%% file: main.tex
\PassOptionsToPackage{numbers,sort&compress}{natbib}
\documentclass[sigconf]{acmart}

\usepackage[normalem]{ulem}
\usepackage{graphicx}
\usepackage{balance}  
\usepackage{color}
\usepackage{multirow}
\usepackage{tabularx, xspace,url}
\usepackage[vlined,linesnumbered,ruled,norelsize]{algorithm2e}
\usepackage{subfigure}
\usepackage{verbatim}
\usepackage{makecell}
\usepackage{amsmath}
\usepackage{multirow}
\usepackage{enumerate}
\usepackage[title]{appendix}
\usepackage{url}

\newcommand{\redc}[1]{{\color{black} #1}}

\newcommand{\rdc}[1]{{\color{red} #1}}

\acmConference[SIGMOD '21]{2021 SIGMOD International Conference on Management of Data}{June 20--25, 2021}{Xi'an, Shaanxi, China}
\acmYear{2021}
 \copyrightyear{2021}
 \acmPrice{15.00}

\begin{document}
\fancyhead{}
\title[Efficient and Effective Algorithms for Revenue Maximization in Social Advertising]{Efficient and Effective Algorithms for Revenue \\Maximization in Social Advertising}


%
\author{Kai Han}
\affiliation{%
  \institution{School of Computer Science and Technology / SuZhou Research Institute, University of Science and Technology of China}
}
\email{hankai@ustc.edu.cn}

\author{Benwei Wu}
\affiliation{%
  \institution{School of Computer Science and Technology, University of Science and Technology of China}
}
\email{wubenwei@mail.ustc.edu.cn}

\author{Jing Tang}
\affiliation{%
  \institution{School of Computing,\\ National University of Singapore}
}
\email{isejtang@nus.edu.sg}

\author{Shuang Cui}
\affiliation{%
  \institution{School of Computer Science and Technology, University of Science and Technology of China}
}
\email{lakers@mail.ustc.edu.cn}

\author{Cigdem Aslay}
\affiliation{%
\institution{Department of Computer Science, Aarhus University}
}
\email{cigdem@cs.au.dk}

\author{Laks V. S. Lakshmanan}
\affiliation{%
  \institution{Department of Computer Science, University of British Columbia}
}
\email{laks@cs.ubc.ca}


\begin{abstract}
We consider the revenue maximization problem in social advertising, where a social network platform owner needs to select seed users for a group of advertisers, each with a payment budget, such that the total expected revenue that the owner gains from the advertisers by propagating their ads in the network is maximized. Previous studies on this problem show that it is intractable and present approximation algorithms. We revisit this problem from a fresh perspective and develop novel efficient approximation algorithms, both under the setting where an exact influence oracle is assumed and under one where this assumption is relaxed. Our approximation ratios significantly improve upon the previous ones. Furthermore, we empirically show, using extensive experiments on four datasets, that our algorithms considerably outperform the existing methods on both the solution quality and computation efficiency.
\end{abstract}

\maketitle


\input{0-command.tex}
\newlength{\textfloatsepsave}
\setlength{\textfloatsepsave}{\textfloatsep}

\begin{sloppy}
\input{1-intro.tex}

\input{2-prelim.tex}

\input{3-approxalg.tex}

\input{4-sampling.tex}

\input{5-exp.tex}

\input{6-related.tex}

\input{7-conclusion.tex}

\clearpage
\bibliographystyle{ACM-Reference-Format}
\bibliography{Mylib}

\end{sloppy}

\setlength{\textfloatsep}{\textfloatsepsave}
\begin{appendices}
\input{8-appendix}
\end{appendices}



\end{document}

%% file: 0-command.tex
\def\done{\hspace*{\fill} {$\square$}}
\def\header{\vspace{2.5mm} \noindent}
\def\extraspacing{\vspace{2mm}}
\def\myspacing{\vspace{2mm}}
\def\figcapup{\vspace{-2mm}}
\def\figcapdown{\vspace{-3mm}}
\def\tblcapup{\vspace{-0mm}}
\def\tblcapdown{\vspace{2mm}}
\def\tbldown{\vspace{-2mm}}
\def\tblup{\vspace{2mm}}
\def\codecapup{\vspace{-3mm}}
\def\codecapdown{\vspace{-2mm}}
\newcommand{\E}{{\mathbb E}\xspace}

\newcommand{\AG}{$\mathsf{AdaptGreedy}$\xspace}
\newcommand{\EP}{$\mathsf{EPIC}$\xspace}
\newcommand{\KM}{\mathit{KM}}
\newcommand{\KC}{\mathit{KC}}
\newcommand{\UB}{\mathit{UB}}
\newcommand{\LB}{\mathit{LB}}
\newcommand{\R}{\mathcal{R}}
\newcommand{\OPT}{\mathrm{OPT}}
\newcommand{\cpe}{\mathit{cpe}}

\newcommand{\wpi}{\widetilde{\pi}}

\newcommand{\feasibles}{\mathcal{F}}
\newcommand{\groundset}{\mathcal{U}}
\newcommand{\posreals}{\mathbb{R}^{\ge 0}}
\newcommand{\ra}{\rightarrow}

\newcommand{\spara}[1]{\vspace{0mm}\noindent\textbf{#1.}}
\newcommand{\tabincell}[2]{\begin{tabular}{@{}#1@{}}#2\end{tabular}}

\newcommand{\jing}[1]{{\color{black}Jing: #1}}
\newcommand{\tj}[1]{{\color{black} #1}}
\newcommand{\hkl}[1]{{\color{black} #1}}
\newcommand{\hk}[1]{{\color{black} #1}}
\newcommand{\hkr}[1]{{\color{black} #1}}
\newcommand{\laks}[1]{{\color{black}{#1}}}
\newcommand{\ca}[1]{{\color{black}{#1}}}

\newcommand{\grn}[1]{{\color{green}{#1}}}

\newcommand{\brn}[1]{{\color{black}{#1}}}

\newcommand{\vlt}[1]{{\color{black}{#1}}}

\newcommand{\revise}[1]{{\color{black}{#1}}}

\newcommand{\XX}{\textcolor{red}{XXX~}}

\def\correspondingauthor{\footnote{Corresponding author.}}

\newcommand{\eat}[1]{\ignorespaces}
\newcommand{\red}[1]{{\color{black}{#1}}}

\newcommand{\toblue}[1]{{\color{black}{#1}}}

\newcommand{\figdir}{figs-revision} 

%% file: 1-intro.tex
\section{Introduction} \label{sec:intro}

With the proliferation of Online Social Networks (OSNs) such as Facebook and Twitter, there emerge great opportunities for social network platform owners and advertisers to gain revenue through placing advertisements on OSNs. The availability of rich information in OSNs, such as user profiles, shared posts, and user behavioral features, brings a tremendous  opportunity for personalized advertising, while the interactions between OSN users make it possible for advertisers to propagate their marketing messages to a large audience in short time. Due to these advantages, the paradigm of \textit{social advertising} on OSNs has aroused great interest from both academia and industry.
\revise{For example, statistics on social advertising across companies worldwide, published in a Hootsuite blogpost\footnote{\url{https://blog.hootsuite.com/social-media-advertising-stats/}.}, include interesting trends, such as ``27\% of internet users say they find new products and brands through paid social ads''.}

Typically, social advertising is implemented by showing ``promoted posts'' in the news feed of OSN users, and these promoted posts can take various forms such as images, videos, and hyper-links \eat{directing a user} to advertisements. After a user sees a promoted post, she may react to it by performing a social action in the form of \textit{comment}, \textit{like}, or \textit{reshare}. Once a user performs a social action, it is counted as an engagement with the advertisement, and the advertiser would pay a unit amount to the OSN platform owner for the engagement. This is a typical marketing paradigm known as Cost-Per-Engagement (CPE) advertising, where advertisers only pay when users engage with the advertisements.

As the users in OSNs often influence each other based on their social affinity, the promoted posts could propagate in the network as a result of the users' social actions.\footnote{\footnotesize It is reported that the promoted posts in Tumblr are reposted more than ten thousand times on average~\cite{Aslay2017}.} Therefore, a large number of users could eventually engage with the advertisement through influence propagation, even when only a small number of ``seed users'' initially engage with the promoted posts inserted into their feed by the platform.
\eat{
Clearly, the revenue (measured by the number of influenced users) gained through this ``multi-hop'' advertising paradigm strictly outperforms that of the traditional one-hop advertising, as can be seen from \figurename~\ref{fig:hops}, where we plot the average number of influenced nodes of a random seed node under different hop constraints in the well-known Flixster and LastFM social networks~\cite{barbieri2014influence,Aslay2017},  by simulating the influence propagation for ten thousand times based on the influence probabilities learned from the TIC influence model~\cite{barbieri2012topic}. }   \eat{To understand the potential marginal revenue gained from this  ``multi-hop'' advertising over traditional one-hop advertising, we conducted an experiment on the real networks of Flixster and LastFM with influence probabilities learned from the TIC influence model~\cite{barbieri2012topic}. For a random seed node, we measured the average number of influenced users by simulating influence propagation 10,000 times, shown in  \figurename~\ref{fig:hops} for different number of hops. It can be clearly seen that after several hops, the influence, and hence social actions on promoted posts, keep propagating to a large number of users.} \toblue{Based on the above observation,} it is possible to boost the revenue of social advertising by intelligently selecting a few most ``influential'' seed users for initial endorsements of the advertisements such that the total engagements through influence propagation, and hence the revenue of social advertising, are maximized. In this paradigm, seed users should be properly incentivized by the advertiser to ensure that they are actually engaged with the advertisement.
%
%
%
\revise{
In fact, the paradigm of ``incentivized multi-hop social advertising'' described above has already been implemented or experimented with in several major companies. For example, Youtube and Twitch currently pay content creators (influencers) selected as seeds, a portion of the revenue from a video ad included in the content~\cite{news1, news2, news4}. In May 2020, Instagram rolled out a new feature, IGTV, providing similar monetization opportunities for influencers: IGTV ads appear when users click on IGTV videos in their feed. In these platforms, user actions (watching/clicking) are visible to their peers and the content propagates to the peers, and with them the ads. Instagram decides which ad will appear in which content, and shares 55\% of the IGTV ad revenue with the content creator~\cite{news21}. Similar situations also appear in FaceBook Stories, a social media application where users can share photos/videos with their friends, allowing the associated advertisements to naturally propagate virally in multi-hops~\cite{news3, news5}. Due to a lot of successful stories on multi-hop advertising (a.k.a. ``viral marketing'' or ``viral advertising'')~\cite{news18}, it also has aroused significant interest from  researchers in the area of advertising~\cite{himelboim2019social,liu2020creating,book1,book2}.
}

In this work, we consider a social advertising scenario where there is a social networking platform owner (referred as the \textit{host}) and a set of $h$ advertisers who need social advertising service provided by the host. Each advertiser $i\in \{1,\dotsc,h\}$ needs to propagate an advertisement in the social network and has a \textit{budget} $B_i$ to pay (i) the host for total engagements with their ad and (ii) the seed users for  incentivizing them. As the products associated with different ads could be competitive, we assume that a seed user can endorse at most one ad. In such a scenario, the host is faced with a \textit{revenue maximization} problem, i.e., how to select the seed users for each advertiser under the constraints described above, such that the revenue gained by the host is maximized.

\vspace{1mm}
\spara{Previous Work and Limitations} We discuss the related work on computational and social advertising and other related topics in Section~\ref{sec:relatedwork} and provide here a brief comparison with the work that is most related to ours.
\eat{A set function $f: 2^{\groundset} \ra \posreals$ is monotone, if $f(S) \le f(T)$, for $S\subset T\subseteq \groundset$; it is submodular, if $f(T \cup \{x\}) -f(T) \le f(S \cup \{x\}) -f(S)$, for $S \subset T\subseteq \groundset$ and $x\in \groundset \setminus T$.}
Aslay et al.~\cite{Aslay2017} studied revenue maximization for social advertising, and to our our knowledge, theirs is the only work that studies the problem within the ``\emph{incentivized}" social advertising framework, where seed users are paid monetary incentives. They study the problem under the \textit{Topic-aware Independent Cascade (TIC)} model~\cite{barbieri2012topic} as we do in our work, allowing ad-specific IC propagation parameters on each edge. They show that the revenue maximization problem in incentivized advertising corresponds to monotone\ca{\footnote{\footnotesize $f(S) \le f(T)$ whenever $S\subseteq T \subseteq  \groundset$.}} submodular\ca{\footnote{\footnotesize $f(T \cup \{x\}) -f(T) \le f(S \cup \{x\}) -f(S)$ whenever $S \subseteq T \subseteq \groundset$ for $x \in \groundset \setminus T$.}} function maximization subject to a partition matroid constraint on the ads-to-seeds allocation, and submodular knapsack constraints on advertisers’ budgets, generalizing the special case of the problem with a single submodular knapsack constraint~\cite{Iyer2013,iyer2015submodular}. \ca{As in \cite{Iyer2013,iyer2015submodular}, they propose cost-agnostic and a cost-sensitive greedy algorithms with provable approximation guarantees.}




\eat{
The algorithms in~\cite{Aslay2017}, however, leave much room for improvements. First, without considering the implementation issues, the theoretical performance guarantees of the algorithms proposed in~\cite{Aslay2017} are rather weak. }
The approximation ratios of the above algorithms depend on the input social network instance and could be arbitrarily small, which can hurt the quality of the approximation achieved. Clearly, an approximation ratio that does not depend on the network instance, if possible, is more desirable.
\eat{This impedes their algorithms from achieving theoretically accurate results even if their performance guarantees can be perfectly implemented. \redc{Second, it is unclear how their algorithms can 
ensure budget-feasibility in practice.}}
Second, the manner in which budget feasibility is ensured in their algorithms appeals to upper bounds on the expected spread, since it is $\#\mathrm{P}$-hard to compute the expected spread exactly. This has the consequence of making their algorithm ``conservative", in that the seed allocation provided by the algorithms may end up \textit{under-utilizing} the budget. Third, the experimental results in~\cite{Aslay2017} reveal that the computational overheads of their algorithms can be significant. 
More detailed discussions on the results in~\cite{Aslay2017} can be found in Section~\ref{sec:prelim-existing}.

\vspace{1mm}
\spara{Contributions} Motivated by the limitations of the existing studies, we propose new approximation algorithms for the revenue maximization problem in social advertising that provide significantly improved approximation ratios, which are \emph{independent of the network instance}. \eat{Moreover, our algorithms can be implemented efficiently in practice while still achieving rigorous performance guarantees.} Furthermore, our algorithms are more efficient and achieve better revenue thanks to a combination of improved approximation guarantees and a \ca{bicriteria approximation} strategy for better utilizing given advertiser budgets. \laks{We further elaborate on the need for bicriteria approximation and how it can be managed in practice \ca{in Sections~\ref{sec:prelim-existing} and ~\ref{sec:onebatchsampling}.}}



\hkl{
{\sl More specifically, our major contributions include the following.} First, we propose (Section~\ref{sec:solvingkmedian}) approximation algorithms with provable performance guarantees under the assumption that there is an \textit{influence spread oracle}, which 
returns the exact influence spread of any set of seed nodes.
Our algorithms are based on several novel methods combining a greedy node-selection strategy and binary search, which fully exploit the special structure of the revenue maximization problem. The approximation ratio $~\lambda$ of our algorithms is \textit{independent} of the input social network and, depending on the number $h$ of advertisers, is characterized as follows: 
\begin{equation}
\lambda=
\begin{cases}
{1}/{3}, &\mathsf{if}~h= 1; \\
\frac{1}{2(h+1)(1+\tau)},&\mathsf{if}~h\in \{2,3\}; \\
\frac{1}{(h+6)(1+\tau)},&\mathsf{if}~h\geq 4.
\end{cases}
\end{equation}
}
\noindent Here $\tau$ is any number in $(0,1)$, which reflects the trade-off between accuracy and efficiency of our revenue maximizing algorithm. \revise{Compared to the network-dependent approximation ratios proposed in \cite{Aslay2017} which could be arbitrarily bad, our approximation ratio $\lambda$ is essentially a \textit{constant} for a given $h$. We believe this  improvement is highly-nontrivial and theoretically interesting.
}

Second, we extend our algorithms (Section~\ref{sec:rmwithoutoracle}) to the practical case where there is no exact influence spread oracle\footnote{Computing the exact influence spread is $\#\mathrm{P}$-hard~\cite{ChenWW2010}.}, by using \laks{novel adaptation of} the notion of Reverse-Reachable Sets proposed by Borgs et al.~\cite{Borgs2014}. We prove that our algorithms can achieve a $\lambda-\epsilon$ approximation ratio under the relaxed budget constraint of $(1+\varrho)B_i$ for each advertiser $i$ with high probability, where $\lambda$ is the approximation ratio shown above, $\epsilon$ is any number in $(0,\lambda)$, and $\varrho$ is any number in $(0,1)$, \laks{which controls how much the budget is overshot. We discuss how budget overshoot can be managed in practice in Section~\ref{sec:limitations}.}

Third, we evaluate the efficiency and effectiveness of our algorithms with extensive experiments on 4 social networks containing up to $\sim$69M edges (Sec.~\ref{sec:pe}). The experimental results show that our algorithms significantly outperform the algorithms in~\cite{Aslay2017} on both the processing time and the achieved revenue under TIC model. 

We provide the necessary preliminaries in the next section and formally define the problem studied. Section~\ref{sec:relatedwork} discusses related work. Section~\ref{sec:concl} summarizes the paper and discusses interesting directions for future work. 

%% file: 2-prelim.tex
\section{Preliminaries} \label{sec:prelim}

\subsection{Problem Statement} \label{sec:prelim-IM}

Following the social advertising model in \cite{Aslay2017}, we assume that there exist a set of $h$ \textit{advertisers} and a \textit{host}, which is the owner of a social networking platform. The host owns a social network represented as a directed graph $G=(V,E)$, where $V$ and $E$ denote the sets  of nodes (i.e., users) and edges in $G$ respectively, with $|V|=n$ and $|E|=m$. Each advertiser $i$ provides the host with an ad $i$, and the host is responsible to select a set of \textit{seed users} $S_i\subseteq V$ to endorse ad $i$. It is assumed that each node $u\in V$ has a cost $c_i(u)$ to be ``activated'' to endorse ad $i$; an influence propagation process will be triggered to activate more nodes in $V$ after the seed users in $S_i$ are activated. Moreover, each activated node in the influence propagation process would bring revenue to the host as it  engages with ad $i$. After the influence propagation process ends, advertiser $i$ should pay an amount of money for:
\begin{enumerate}
\item The incentive cost of activating the users in $S_i$, i.e., $c_i(S_i)=\sum_{u\in S_i}c_i(u)$; this amount is paid to the seed users in $S_i$.
\item A cost-per-engagement amount $\mathit{cpe}(i)$ for each engagement with ad $i$ during the influence propagation process as described above; this amount is paid to the host for its service. 
\end{enumerate}

\revise{
\spara{Discussion} In our problem, we assume that the value of $\cpe(i)$ is agreed upon between advertiser $i$ and the host for each $i\in [h]$. This could happen in the scenarios, e.g., the host posts $\cpe(i)$ to advertiser $i$ as a ``take-it-or-leave-it'' price according to prior marketing studies on the advertised product, and advertiser $i$ clearly would only accept this price if $\cpe(i)$ is less than her/his ``value-per-engagement'', resulting in no negative utility of any part. We note that some excellent studies~\cite{grubenmann2020tsa} consider a scenario where $\cpe(i)$ is unknown and is  determined by truthful auction mechanisms. The auction problem considered in these studies is orthogonal to ours.

We adopt a general model in which seed node costs can be any positive number obtained by any existing pricing strategy for social networks. For example, a simple strategy prices nodes based on their number of  followers.\footnote{\revise{Klear’s survey~\cite{emarketer} shows that brands pay, on average, 114 dollars for each video post on Instagram to nano-influencers ($500{\sim} 5K$ followers) and 775 dollars per video to more powerful users ($30K{\sim} 50K$ followers).}} As another example, a recent study~\cite{zhu2020pricing} proposed another pricing strategy based on the expected influence gains of candidate seeds.}

\spara{Influence Propagation Model} We adopt the \textit{Topic-aware Independent Cascade} (TIC) model proposed  in~\cite{barbieri2012topic} to characterize the influence propagation process for each ad $i$, as described below.
At first, the set of ``seed nodes'' in $S_i$ are activated at time 0. Afterwards, each node $u$ newly activated at time $t-1$ has a single chance to activate each of its inactive out-neighbors $v$ at time $t$, succeeding with probability $p_{(u,v)}^i$. \eat{The node $u$ then stays active but activates no more neighboring nodes at a time after $t$.} The expected number of total activated nodes when the influence propagation ends is denoted by $\sigma_i(S_i)$ and is called the (expected) spread of $S_i$.

The activation probability $p_{(u,v)}^i$ associated with each edge $(u,v)\in E$ under the TIC model is defined as follows. Assume that there exist $L$ latent \textit{topics} for ads and users' interests, and there is a hidden random variable $Z$ ranging over the $L$ topics.  The TIC model then maps ad $i$ to a distribution $\phi_i(\cdot)$ over the $L$ latent topics with $\phi_i(z)=\Pr [Z=z\mid i]$ and $\sum_{z=1}^L \phi_i(z)=1$. 
The influence propagation in social advertising can be topic-dependent, i.e., user $u$'s influence on user $v$ may depend on the topic of the ad that is being propagated.
In the TIC model, the probability that $u$ can activate $v$ for ad $i$ (i.e., $p_{(u,v)}^i$) is  defined as $p_{(u,v)}^i=\sum_{z=1}^L \phi_i(z)\cdot \hat{p}_{(u,v)}^z$, where $\hat{p}_{(u,v)}^z$ is the probability that $u$ can activate $v$ under latent topic $z$.


\spara{The Revenue Maximization (RM) Problem}
\revise{
Following the social advertising model, the host can gain an expected revenue of $\pi_i(S_i)=\mathit{cpe}(i)\cdot \sigma_i(S_i)$ from advertiser $i$, and the \textit{total expected revenue} of the host is $\sum_{i\in [h]}\pi_i(S_i)$. The Revenue Maximization (RM) Problem aims to maximize this total expected revenue under the following constraints: (1) each advertiser $i$ has a budget $B_i$ for its total social ad spend, i.e., the total amount paid to the host and to the seed users in $S_i$; (2) each user in $V$ can endorse at most one ad within a certain time window.\footnote{\revise{Note that limiting the ads endorsed by a seed can increase the credibility for followers. The same constraint has been widely adopted to avoid undesirable situations, e.g., the same celebrity endorsing Nike and Adidas at the same time~\cite{grubenmann2020tsa,Aslay2017,lu2013bang,Chalermsook2015}.}} Formally, the RM problem is defined as follows.}


\vspace{-1ex}
\begin{definition} \label{def:k-median}
The \textit{Revenue Maximization} (RM) problem for social advertising aims to identify an optimal solution $\vec{S} = (S_1,\dotsc, S_h)$ to the following optimization problem:
\vspace{-1ex}
\begin{align*}
        {\mathbf{Maximize}}\quad&\sum\nolimits_{i\in [h]}\pi_i(S_i) \nonumber\\
        \mathbf{s.t.}\quad&\pi_i(S_i)+c_i(S_i)\leq B_i,~~~\forall i\in [h],\\
        \quad&S_i\cap S_j=\emptyset,~~~i\neq j,~~\forall i,j\in [h].
\end{align*}
\label{def:rmproblem}
\vspace{-5ex}
\end{definition}

It is well-known that the influence spread function $\sigma_i(\cdot)$ is monotone and submodular under the TIC model~\cite{barbieri2012topic}, so the revenue function $\pi_i(\cdot)$ for every $i\in [h]$ is also monotone and submodular. Aslay et al.~\cite{Aslay2017} have shown that the RM problem is NP-hard.

We now introduce some notations\eat{(see Table~\ref{tbl:prelim-notations} for a summary)}. We use $\vec{S}$ to represent an allocation, i.e., a list of sets $(S_1,\dotsc, S_h)$, and use $\vec{O}=(O_1,\dotsc, O_h)$ to represent an optimal solution to the revenue maximization problem. We abuse these notations slightly by using $\vec{S}$ to also represent the set $\{(u,i) \mid u\in S_i\wedge i\in [h]\}$ (and also abuse $\vec{O}$ similarly), as these representations are essentially equivalent. For any set $M\subseteq V\times [h]$, we define $\pi(M)=\sum_{i\in [h]}\pi_i(M_i)$ where \laks{$M_i=\{u\in V \mid  (u,i)\in M\}$,} and we define $\mathrm{OPT}=\pi(\vec{O})$. Furthermore, for any set function $f(\cdot)$, we use $f(X \mid Y)=f(X\cup Y)-f(Y)$ to denote the \textit{marginal gain} of $X$ with respect to $Y$. For example, we have $\pi_i(u\mid S_i)=\pi_i(S_i\cup\{u\})-\pi_i(S_i)$ and  $\pi((v,j)\mid \vec{S})=\pi_j(v\mid S_j)$.
\tj{Finally, we use $\zeta_i(u\mid S_i)$ to denote the \textit{marginal rate} of node $u$ upon seed set $S_i$ for advertiser $i$, defined as the ratio of the marginal gain in revenue to the marginal gain in payment, i.e.,
\begin{equation}\label{eqn:marginal-rate}
    \zeta_i(u\mid S_i)=\frac{\pi_i(u\mid S_i)}{c_i(u)+\pi_i(u\mid S_i)}.
\end{equation}}

\vspace{-6mm}
\revise{
\subsection{Existing Solutions} \label{sec:prelim-existing}
To the best of our knowledge, only Aslay et al.~\cite{Aslay2017} have addressed the revenue maximization problem in this framework. They show that the RM problem is essentially a \textit{submodular maximization problem with a partition matroid and multiple submodular knapsack constraints}, and propose two approximation algorithms\textemdash Cost-Agnostic Greedy (i.e., CA-Greedy) and Cost-Sensitive Greedy (i.e., CS-Greedy). 
Both algorithms iteratively select seed nodes under the budget constraint. In each iteration, CA-Greedy (resp. CS-Greedy) greedily selects an element $(u,i)$ such that the marginal gain $\pi_i(u\mid S_i)$ (resp. the marginal rate $\zeta_i(u\mid S_i)$) is maximized.

Aslay et al.~\cite{Aslay2017} prove that the CS-Greedy algorithm has an approximation ratio of
\begin{equation}
1-\frac{R\cdot \rho_{max}}{R\cdot \rho_{max}+(1-\max_{i\in [h]}\kappa_i)\rho_{min}}, \label{eqn:boundcsgreedy}
\end{equation}
and that the approximation ratio of CA-Greedy follows from the result of Conforti et al.~\cite{conforti1984submodular} for submodular maximization subject to an independence system
\begin{equation}
\left(1-(1-{\kappa}/{R})^r\right)/\kappa. \label{eqn:boundcagreedy}
\end{equation}
The parameters $r,R,\kappa,\kappa_i,\rho_{min},\rho_{max}$ in the above ratios all depend on the input social network and the detailed definitions of them can be found in~\cite{Aslay2017}.}


\ca{Note that the approximation ratios in Eqn.~\eqref{eqn:boundcsgreedy} and~\eqref{eqn:boundcagreedy} hold under the assumption that there is an influence spread oracle which can exactly evaluate $\sigma_i(\cdot)$. Given the $\#\mathrm{P}$-hardness of computing $\sigma_i(A)$ for any given $A\subseteq V$~\cite{ChenWW2010}, Aslay et al.~\cite{Aslay2017} further propose  algorithms  TI-CARM and TI-CSRM, as practical versions of CA-Greedy and CS-Greedy, by extending TIM~\cite{TangXS2014} for influence spread estimation.}

\eat{the influence spread function $\sigma_i(\cdot)$ can be exactly evaluated, which is unrealistic as it is $\#\mathrm{P}$-hard to compute  $\sigma_i(A)$ for any given $A\subseteq V$~\cite{ChenWW2010}. Therefore, the approximation ratios in Eqn.~\eqref{eqn:boundcagreedy} and Eqn.~\eqref{eqn:boundcsgreedy} cannot be achieved in practice. To address this problem, Aslay et al.~\cite{Aslay2017} further propose two algorithms called TI-CARM and TI-CSRM that are adaptions of CA-Greedy and CS-Greedy, respectively.}

They prove that these two algorithms can return a solution $\vec{S}$ satisfying the following performance bound:
\begin{equation}
\pi(\vec{S})\geq \beta\cdot \mathrm{OPT} - \epsilon\cdot \sum\nolimits_{i\in [h]}\mathit{cpe}(i)\cdot \sigma_i(N_i), \label{eqn:boundforsampledgreedy}
\end{equation}
\laks{where $N_i\subseteq V$ is a node set that maximizes $\sigma_i(N_i)$ under the condition that $|N_i|$ equals the estimated cardinality of
\eat{an optimal solution to the RM problem when there is only one advertiser $i$.}
the maximum allocation to ad $i$ under the given  budget.}
The value of $\beta$ equals Eqn.~\eqref{eqn:boundcagreedy} and Eqn.~\eqref{eqn:boundcsgreedy} for TI-CARM and TI-CSRM, respectively.

\revise{
\subsubsection{Limitations of the Existing Solutions}
\label{sec:limitations}

Unfortunately, the solutions provided in~\cite{Aslay2017} suffer from the following major shortcomings:
(i) The exact value of the approximation bounds in Eqn.~\eqref{eqn:boundcsgreedy}--\eqref{eqn:boundcagreedy} cannot be computed easily\textemdash there is no obvious way to calculate them in polynomial time. This is a direct consequence of the bounds depending on the network instance and the $\#\mathrm{P}$-hardness of influence spread computation.
(ii) The theoretical approximation ratios of CA-Greedy and CS-Greedy 
could be arbitrarily small. To see this, consider the case $h = 1$ and a network $G=(V,E)$, with $|V|=n$ nodes. Assume that $E$ contains an edge $(u,v)$, where the out-degree of $v$ is $0$, and $p_{(u,v)}^1=1$. On this instance, the approximation ratio of CS-Greedy is at most  $\frac{c_1(v)}{R\cdot \rho_{max}+c_1(v)}$, which can be arbitrarily small as $\rho_{max}/c_1(v)$ can be arbitrarily large. CA-Greedy also has similar problems and Aslay et al.~\cite{Aslay2017} actually indicate in their paper that the worst-case approximation ratio of CA-Greedy is $1/R$, which is in the order of $\mathcal{O}(1/n)$.
\eat{
\redc{Third, it is unclear how budget-feasibility can be achieved in practice by the algorithms proposed in~\cite{Aslay2017}. Due to the space constraint, a more detailed discussion on this issue can be found in our technical report~\cite{techreport}.}
}
}
(iii) The algorithms in~\cite{Aslay2017} incur large computational overheads by their implementation. In fact, experimental results reveal that the running time and memory consumption of the TI-CARM and TI-CSRM algorithms in \cite{Aslay2017} both grow drastically when $\epsilon$ gets small. Therefore, the work in \cite{Aslay2017} has to set $\epsilon$  to a relatively large number (e.g., $0.3$) such that TI-CARM and TI-CSRM can handle a social network with 4.8M nodes using a computer equipped with 264GB memory.
(iv)
The manner in which budget feasibility is ensured by the TI-CARM and TI-CSRM appeals to upper bounds on the expected spread when using estimations from a sample. This results in their seed allocation under-utilizing the budget to a great extent for the sake of not violating budget constraints. Naturally, allowing the host to control how much the budget can be \textit{overshot}, for the sake of fully-utilizing advertisers budgets, is more desirable as this would also imply higher revenue for the host. There can be an agreement between the host and an advertiser that specifies who would pay the excess amount when the budget is overshot: this might be the advertiser as they would be receiving more engagements to their ad compared with when their budget is  under-utilized; \textit{or} this might be the host preferring to give some ``free service" to advertisers since they can fully earn $B_i$. Note that in the latter case, the host can simply control the amount of free service provided, using a parameter $\varrho$. In an extreme, the host could use $B_i/(1+\varrho)$ as the input budget to the algorithms, thus canceling out the effect of the overshoot (details in Section~\ref{sec:rmwithoutoracle}). Our contributions address all four limitations.

%% file: 3-approxalg.tex
\section{Solving RM with an Oracle} \label{sec:solvingkmedian}

%


In this section, we present algorithms for the revenue maximization problem under the assumption that there is an oracle to compute $\sigma_i(A)$ for any $i\in [h]$ and $A\subseteq V$. We will present algorithms without this assumption in the next section.

 \vspace{-2mm}
\subsection{Algorithms for a Single Advertiser} \label{sec:algsa}

\hkl{
We first consider the case where there is only one advertiser $i$. In this case, the RM problem defined in Definition~\ref{def:rmproblem} belongs to the class of submodular maximization problems with a \textit{single submodular knapsack constraint}, 
introduced by Iyer et al.~\cite{Iyer2013}.
\ca{As in \cite{Aslay2017}, the greedy approximation guarantees provided by Iyer et al.~\cite{Iyer2013, iyer2015submodular} are instance dependent and could be arbitrarily small, and in the case of cost-sensitive approximation, the guarantee can be unbounded as they acknowledge.}
\eat{However, Iyer et al.~\cite{Iyer2013} have not provided a constant approximation ratio.}
We now show that a simple $\mathsf{Greedy}$ algorithm (Algorithm~\ref{alg:naivegreedy}) can achieve a constant $\frac{1}{3}$-approximation\eat{,as shown in Theorem~\ref{thm:greedyratioforsa}}. The 
algorithm greedily selects a seed user with the maximum \tj{marginal rate} from the input candidate set $U$, and adds $u$ into $S_i$ if the total cost of the currently selected nodes is no more than $B_i$. It adds $u$ into $D_i$ if $u$ is the first node satisfying $c_i(S_i\cup \{u\})>B_i$ (we call such a node $u$ as a ``stopple node''). Finally, $\mathsf{Greedy}$ returns one of $S_i$ or  $D_i$, whichever has the larger revenue. 
}

\vspace{-1mm}
\hkl{
\begin{theorem}
When there is only one advertiser $i$, the $\mathsf{Greedy}(V,i)$ algorithm returns a solution $S^*_i\subseteq V$ to the revenue maximization problem with approximation ratio of $1/3$.
\label{thm:greedyratioforsa}
\end{theorem}
}

 \vspace{-2mm}
\subsection{Algorithms for Multiple Advertisers} \label{sec:algformulteadvertiser}

In this section, we provide algorithms for the RM problem when the number of advertisers is more than one (i.e., $h>1$). A straightforward approach is to apply the greedy selection rule in the CS-Greedy algorithm~\cite{Aslay2017}, i.e., selecting an element $(u,i)$ at each step with its \tj{marginal rate} being as large as possible. However, this approach could not have a ``nice'' (i.e., network independent) approximation ratio for the case of $h>1$ due to the following reasons. In fact, we cannot guarantee that the \tj{marginal rate} of $(u,i)$ is no less than any unselected element in the optimal solution $\vec{O}$, because there could exist $(v,j)\in \vec{O}$ with a larger \tj{marginal rate} but we have already selected $(v,\ell)$ (for certain $\ell\in [h]$), and hence we have to select $(u,i)$ instead of $(v,j)$.
\eat{
Due to this issue, we cannot use the method in the proof of Theorem~\ref{thm:greedyratioforsa} to derive a nice performance bound when $h>1$.
}
Due to this issue, applying the proof idea of Theorem~\ref{thm:greedyratioforsa} does not lead to non-trivial bounds when $h > 1$.
\textit{This raises a major challenge.}

%
\begin{algorithm} [t]
    \KwIn{Advertiser $i$, a set $U\subseteq V$ of candidate users;}
    \KwOut{A subset of $U$ selected as the seed nodes;}
    $U\gets U-\{v\mid v\in U \wedge c_i(v)+\pi_i(v)> B_i\}$\;
    $S_i\gets \emptyset;D_i\gets \emptyset$\;
    \While{$U\neq \emptyset\wedge D_i= \emptyset$}{
        $u\leftarrow \arg\max_{v\in U} \tj{\zeta_i({v}\mid S_i)}$;~~$U\gets U-\{u\}$\; \label{ln:greedyselectionrule}
        \lIf{$c_i(S_i\cup\{u\})+\pi_i(S_i\cup\{u\})\leq B_i$}{$S_i\leftarrow S_i\cup \{u\}$}
        \lElse{$D_i\leftarrow \{u\}$}
    }
    $S^*_i\leftarrow \arg\max_{X\in \{S_i,D_i\}} \pi_i(X);$~~\Return ${S}^*_i$\;
    \caption{$\mathsf{Greedy}(U,i)$}
    \label{alg:naivegreedy}
\end{algorithm}
\setlength{\textfloatsep}{7pt}
\begin{algorithm} [!t]
    $M\gets \{(v,j)\colon (v,j)\in V\times [h]\wedge c_j(v)+\pi_j(v)\leq B_j\}$\;
    \lForEach{$j\in [h]$}{
        $S_j\gets \emptyset;~~D_j\gets \emptyset;~~A_j\gets \emptyset;~~I\gets \emptyset$
    }
    \While{$M\neq \emptyset\wedge I\neq [h]$}{
        $(u,i)\leftarrow \arg\max_{(v,j)\in M} {\pi_j({v}\mid S_j)}$;~$M\gets M-\{(u,i)\}$\; \label{ln:greedyruleinthreshold}
      \tj{\lIf{$\zeta_i(u \mid S_i\cup D_i)< {\gamma}/{B_i}\vee D_i\neq \emptyset$ \label{ln:startadding}}{\textbf{continue}}}
      \lIf{$u\in \bigcup_{j\in [h]}S_j\cup D_j$}{\textbf{continue}}
      \lIf{$c_i(S_i\cup\{u\})+\pi_i(S_i\cup\{u\})\leq B_i$}{$S_i\leftarrow S_i\cup \{u\}$}
      \lElse{$D_i\leftarrow \{u\}$;~$I\gets I\cup \{i\}$\label{ln:endadding}}
    }
    \If{$|I|=1$ \label{ln:startif}}{
        $i\gets$ the number in $I$;~\hkl{$A_i\gets \mathsf{Greedy}(V-\bigcup_{j\in [h]}S_j, i)$}\;\label{ln:callnaivegreedy}}
    \lForEach{$j\in [h]$}{
        $S_j'\gets \arg\max_{X\in \{S_j,D_j,A_j\} } \pi_j(X)$\label{ln:largestone}}
    $\vec{S}^*\gets \mathsf{Fill}(\vec{S}');~b\gets |I|$\; \label{ln:callfill}
    \Return $\vec{S}^*, b$\;
    \caption{$\mathsf{ThresholdGreedy}(\gamma)$}
    \label{alg:thresholdgreedy}
\end{algorithm}

\subsubsection{A Greedy Algorithm with a Threshold}

Accordingly, we design a new greedy algorithm dubbed $\mathsf{ThresholdGreedy}$ (see  Algorithm~\ref{alg:thresholdgreedy}) which abandons the greedy selection rule used in CS-Greedy, but simply selects the node with the maximum marginal gain at each step (as in CA-Greedy). However, this has the drawback that a lot of nodes with large seeding costs may be selected, quickly depleting the budgets~\footnote{\revise{Here is a toy example to illustrate the intuition. Suppose that $u,v,w$ are nodes with highest singleton revenues 91, 50 and 45, respectively, and that there are no common nodes reached by them. Let the costs of $u,v,w$ be $9$, $3$ and $2$, respectively. Then, for a budget of $100$, CA-Greedy would select $u$ and exhaust the budget for a revenue of $91$, while CS-Greedy would select $v,w$, obtaining a total revenue of $95$.}}. To address this problem, we set an additional rule that the \tj{marginal rate} of any selected node should be no less than a given threshold $\gamma$ (we will discuss in Section~\ref{sec:searchingforthreshold} how to set the value of $\gamma$). We next explain the details of $\mathsf{ThresholdGreedy}$.

$\mathsf{ThresholdGreedy}$ uses $M\subseteq V\times [h]$ to denote all candidate elements to be selected and uses $I$ to denote the set of advertisers whose budgets have been depleted by the already selected elements (i.e., elements in $\vec{S}\cup \vec{D}$). In each step, it removes an element $(u,i)$ from $M$ with the maximum marginal gain (Line~\ref{ln:greedyruleinthreshold}), and adds $u$ into $S_i$ or $D_i$ if and only if all of the three conditions are satisfied (Lines~\ref{ln:startadding}--\ref{ln:endadding}): (1) \tj{marginal rate} of $(u,i)$ is no less than $\gamma/B_i$; (2) the node $u$ has not been assigned to any advertiser yet; (3) the budget of advertiser $i$ has not been depleted by  nodes \laks{already} in $S_i\cup D_i$. This process terminates either when $M$ is empty or when $|I|=h$. \eat{Similar to that in the $\mathsf{Greedy}$ algorithm,} \ca{As in Algorithm~\ref{alg:naivegreedy}, the ``stopple node'' for each $i$ is stored in $D_i$.}

After this greedy procedure terminates, the budgets of the advertisers in $I$ must have been depleted by the nodes in $\bigcup_{i\in I} S_i\cup D_i$. If there is only one such advertiser $i$, we call the $\mathsf{Greedy}$ algorithm again to find a node set $A_i$ (Line~\ref{ln:callnaivegreedy}), and the revenue of $A_i$ can help to derive the approximation ratio of $\mathsf{ThresholdGreedy}$. Next, the $\mathsf{ThresholdGreedy}$ algorithm sets $S_j'$ to be the one in $\{S_j, A_j, D_j\}$ with the largest revenue for all $j\in [h]$ (Line~\ref{ln:largestone}). Finally, the function $\mathsf{Fill}$ is called (Line~\ref{ln:callfill}) to select more seed nodes for the advertisers whose budgets have not been depleted by the elements in $\vec{S}'$, and function $\mathsf{Fill}$ greedily selects nodes with the maximum \tj{marginal rate} until the budgets of all advertisers are depleted. After that, the $\mathsf{ThresholdGreedy}$ algorithm returns the final solution $\vec{S}^*$.

\begin{algorithm} [!t]
    $M\gets \{(v,j)\colon (v,j)\in V\times [h]\wedge c_j(v)+\pi_j(v)\leq B_j\}$\;
    \While{$M\neq \emptyset$}{
        \tj{$(u,i)\leftarrow \arg\max_{(v,j)\in M} \zeta_j(v \mid S_j)$\;}
        $M\gets M-\{(u,i)\}$\;
      \If{$c_i(S_i\cup\{u\})+\pi_i(S_i\cup\{u\})\leq B_i\wedge u\notin \bigcup_{j\in [h]}S_j$}{$S_i\leftarrow S_i\cup \{u\}$\;}
    }
    \Return $\vec{S}=(S_1,\dotsc,S_h)$\;
    \caption{$\mathsf{Fill}(\vec{S})$}
    \label{alg:fill}
\end{algorithm}

The performance bound of $\mathsf{ThresholdGreedy}$ is shown in Theorem~\ref{thm:boundofthresholdgreedy}. Roughly speaking, the main idea in the proof is to classify the elements in the optimal solution into several categories according to their \tj{marginal rates} with respect to the elements selected by $\mathsf{ThresholdGreedy}$; we then bound the ``revenue loss'' caused by missing the elements in each category by $\gamma$ or the revenue of the solution returned by $\mathsf{ThresholdGreedy}$. \eat{\red{For
lack of space, we omit the proof and refer the reader to \cite{RMA_report}~for details.}}




\vspace{-1ex}
\hkl{
\begin{theorem}
Suppose that Algorithm $\mathsf{ThresholdGreedy}(\gamma)$ returns $(\vec{S}^*, b)$. Then we have:
\vspace{-1ex}
\begin{equation}
\pi(\vec{S}^*)\geq
\begin{cases}
b\cdot\gamma/2, &\mathsf{if}~b\geq 2; \nonumber \\
\max\left\{ \frac{1}{6}(\mathrm{OPT}-h\cdot\gamma), \frac{\gamma}{2}\right\},&\mathsf{if}~b=1; \nonumber \\
 \frac{1}{2}(\mathrm{OPT}-h\cdot\gamma),&\mathsf{if}~b=0. \nonumber
\end{cases}
\end{equation}
\label{thm:boundofthresholdgreedy}
\end{theorem}
}
\vspace{-2ex}

\vspace{-3mm}
\subsubsection{Searching for a Good Threshold} \label{sec:searchingforthreshold}

It can be seen from Theorem~\ref{thm:boundofthresholdgreedy} that the 
\laks{approximation quality} of $\mathsf{ThresholdGreedy}$ is affected by the threshold $\gamma$. Specifically, if $\gamma$ is  small, then the $\mathsf{ThresholdGreedy}$ algorithm could select more elements with large marginal gain. On the other hand, if $\gamma$ is  large, then  $\mathsf{ThresholdGreedy}$  could select more elements with large \tj{marginal rate}. Therefore, we propose a novel binary-search process (see  Algorithm~\ref{alg:search}) to find an appropriate $\gamma$ to strike a balance and find  a good approximation ratio.

\begin{algorithm} [!t]

    $\gamma_2\gets (1+\tau)\gamma_{max};~\gamma_1\gets 0;~\mathcal{Q}\gets\emptyset;~\gamma\gets \gamma_1$\;
     $\vec{T}^*_1\gets \emptyset;~\vec{T}^*_2\gets \emptyset;~b_1\gets 0;~b_2\gets 0$\;

%
    \Repeat{$\big((1+\tau)\gamma_1\geq \gamma_2\big)\vee \big(\gamma_2\leq \min_{i\in [h]}\cpe(i)/(h+6)\big)$\label{ln:stoppcondsearch}}{
        $(\vec{T}, b)\gets \mathsf{ThresholdGreedy}(\gamma);~\mathcal{Q}\gets \mathcal{Q}\cup \{\vec{T}\}$ \label{ln:callthresholdgreedy}\;
        \lIf{$b\geq b_{min}$}{
            $(\vec{T}^*_1, b_1,\gamma_1)\gets (\vec{T},b,\gamma)$}  \label{ln:searchright}
        \lElse{
            $(\vec{T}^*_2, b_2,\gamma_2)\gets (\vec{T},b,\gamma)$} \label{ln:searchleft}
        $\gamma\gets (\gamma_1+\gamma_2)/2$\;
    }
    $\vec{S}^*\gets\arg\max_{\vec{T}\in \mathcal{Q}}\pi(\vec{T})$\;
    \Return $\vec{S}^*, (\vec{T}^*_1, b_1,\gamma_1), (\vec{T}^*_2, b_2,\gamma_2)$\;
    \caption{$\mathsf{Search}(\tau,  b_{min})$}
    \label{alg:search}
\end{algorithm}

The $\mathsf{Search}$ algorithm takes two parameters $\tau\in (0,1)$ and $b_{min}\in \{1,2\}$ as input, and maintains an interval $[\gamma_1,\gamma_2]$ which is initialized to $[0,(1+\tau)\gamma_{max}]$ and is halved at each iteration during binary-search, where $\gamma_{max}$ is defined as
\begin{equation}
\gamma_{max}= \max \left\{ B_j\cdot \zeta_j(v\mid \emptyset)\colon v\in V, j\in [h]\right\}.
\end{equation}
Intuitively, if $\gamma>\gamma_{max}$, then no nodes would be selected by $\mathsf{ThresholdGreedy}$; if $\gamma=0$, then $\mathsf{ThresholdGreedy}$ would select nodes purely based on their marginal gains without considering their costs. The $\mathsf{Search}$ algorithm starts searching from $\gamma=0$, such that it can test as many thresholds as possible in the binary search process for the purpose of maximizing the revenue. \revise{The input parameter $b_{min}$ is a threshold used to guide the searching direction in binary search. For example, if $\gamma$ is too large, then it is very likely that no advertisers would deplete their budgets and hence Line~\ref{ln:callthresholdgreedy} returns $b<b_{min}$, which implies that we should try a smaller $\gamma$ (see Line~\ref{ln:searchleft}). More detailed explanations can be found in the sequel.}

\eat{\note[Laks]{Can we motivate right away, why $\gamma_{max}$ is defined in this way? Also, why call ThresholdGreedy the first time at $\gamma=0$? In fact, the explanation of this algorithm could be considerably improved. It is a bit opaque as it stands. A straightforward way to do this is to explain the role of the various parameters in Search and what roles they play.}
}

\brn{
Throughout the searching process, $\mathsf{Search}$ maintains two solutions $(\vec{T}^*_1, b_1)$ and $(\vec{T}^*_2, b_2)$ such that $(\vec{T}^*_1, b_1)=\mathsf{ThresholdGreedy}(\gamma_1)$ and $(\vec{T}^*_1, b_2)=\mathsf{ThresholdGreedy}(\gamma_2)$, and adds $\vec{T}^*_1$ and $\vec{T}^*_2$ into the set $\mathcal{Q}$ during the search process. Note that $b_1$ (resp.\ $b_2$) represents the number of advertisers whose budgets would be depleted by the nodes selected by $\mathsf{ThresholdGreedy}$ under threshold $\gamma_1$ (resp.\ $\gamma_2$). The binary-search process stops when the length of $[\gamma_1,\gamma_2]$ is sufficiently small, and the $\mathsf{Search}$ algorithm returns the solution in $\mathcal{Q}$ that has the maximum revenue. Roughly speaking, the reason that $\mathsf{Search}$ can return a solution with large revenue is that, it keeps adjusting the interval $[\gamma_1,\gamma_2]$ to ensure that $b_1\geq b_{min}$ and $b_2<b_{min}$. Therefore, when the $\mathsf{Search}$ algorithm stops, if $\gamma_1$ is sufficiently large, then we can guarantee that $\pi(\vec{T}^*_1)\geq b_{min}\gamma_1/2$ is also sufficiently large following Theorem~\ref{thm:boundofthresholdgreedy}; on the other hand, if $\gamma_1$ is small, then $\gamma_2$ should also be small thanks to the stopping condition of $\mathsf{Search}$, so we can guarantee that $\mathrm{OPT}-h\gamma_2$ hence $\pi(\vec{T}^*_2)$ is  sufficiently large according to Theorem~\ref{thm:boundofthresholdgreedy}. The parameter $b_{min}\in \{1,2\}$ is just used to control the solution quality according to the above explanation and Theorem~\ref{thm:boundofthresholdgreedy}. Based on these ideas, we show in Theorem~\ref{thm:proofofsearchtau2} that $\mathsf{Search}(\tau,2)$ achieves a good performance ratio. }

\hkl{
\begin{theorem}
$\mathsf{Search}(\tau,2)$ returns a solution $\vec{S}^*$ satisfying $\pi(\vec{S}^*)\geq \frac{1}{(h+6)(1+\tau)}\mathrm{OPT}$.
\label{thm:proofofsearchtau2}
\end{theorem}
}

By an argument analogous to the proof of Theorem~\ref{thm:proofofsearchtau2}, we get:

\begin{theorem}
$\mathsf{Search}(\tau,1)$ returns a solution $\vec{S}^*$ satisfying $\pi(\vec{S}^*)\geq \frac{1}{2(h+1)(1+\tau)}\mathrm{OPT}$.
\label{thm:approxboundofsearchtau1}
\end{theorem}

\subsection{Putting It Together} \label{sec:puttogether}

\hkl{
Theorems~\ref{thm:greedyratioforsa}--\ref{thm:approxboundofsearchtau1} imply that we can make optimizations based on the number of advertisers $h$. As $h+6\leq 2(h+1)$ when $h\geq 4$, we design an algorithm $\mathsf{RM\_with\_Oracle}$ (Algorithm~\ref{alg:rmwithoracle}) to get the best performance.
The following theorem is immediate:


\vspace{-2ex}
\begin{theorem}
Given any $\tau\in (0,1)$, $\mathsf{RM\_with\_Oracle}(\tau)$ can return a solution to the RM problem with the approximation ratio $\lambda$, where $\lambda=\frac{1}{3}$ for $h=1$, $\lambda=\frac{1}{2(h+1)(1+\tau)}$ for $h\in \{2,3\}$, and $\lambda=\frac{1}{(h+6)(1+\tau)}$ for $h\geq 4$.
\label{thm:finalapproxratio}
\end{theorem}
\vspace{-1ex}

}

\begin{algorithm} [t]
    \lIf{$h=1$}{\Return $\mathsf{Greedy}(V,1)$}
    \lIf{$2\leq h\leq 3$}{\Return $\mathsf{Search}(\tau,1)$}
    \lIf{$h\geq 4$}{\Return $\mathsf{Search}(\tau,2)$}
    \caption{$\mathsf{RM\_with\_Oracle}(\tau)$}
    \label{alg:rmwithoracle}
\end{algorithm}

\eat{
\note[Laks]{Somewhere, there needs to be a brief discussion of who chooses $\tau$ and how.}

\vlt{I have added a paragraph below--Kai}
}

\laks{Note that $\tau$ reflects a trade-off between accuracy and efficiency: the $\mathsf{Search}$ algorithm needs $\mathcal{O}(\log\frac{h\gamma_{max}}{\min_{i\in [h]}\cpe(i)})$ iterations in the worst case and needs $\mathcal{O}(\log\frac{1}{\tau})$ iterations in most cases, while the approximation ratio improves with smaller $\tau$. In practice, $\tau$ can be set as a small number (e.g., $\tau=0.1$) while our algorithm still runs very fast, as our experiments in Section~\ref{sec:pe} demonstrate.}

%% file: 4-sampling.tex

\vspace{-1ex}
\section{Solving RM Without an Oracle} \label{sec:rmwithoutoracle}

\laks{
In this section, we present algorithms for the RM problem without the influence spread oracle assumed in last section. Our algorithms are based on 
a novel adaptation to Reverse Reachable Sets~\cite{Borgs2014} and several efficient sampling techniques with guaranteed performance bounds for the RM problem.
}

\vspace{-2mm}
\subsection{Reverse Reachable Sets} \label{sec:rrset}
The concept of Reverse Reachable Set (RR-set) was  first proposed in~\cite{Borgs2014} for the Independent Cascade (IC) influence model. Given a social network $G$ with each edge $(u,v)$ associated with an influence propagation probability $p_{u,v}$ under the IC model, a random RR-set $R$ is generated by first selecting a node $v\in V$ uniformly at random, and then setting $R$ as the set of nodes in $V$ that are reverse-reachable from $v$ in a random graph generated by independently removing each edge $(u,v)\in E$ with probability $1-p_{u,v}$. Given any node set $A$ and a random RR-set $R$, we can define a random variable $X(A,R)$ such that $X(A,R)=1$ when $A$ intersects $R$ and $X(A,R)=0$ otherwise. Borgs et al.~\cite{Borgs2014} show  that the influence spread of $A$ under the IC model equals  $n\cdot \E[X(A,R)]$, and  $\E[X(A,R)]$ can be  estimated in an unbiased manner  by the empirical mean $\sum_{R\in \R} X(A,R)/|\R|$ based on concentration bounds, where $\R$ is a set of generated RR-sets.


\vspace{-2mm}
\subsection{A New Method for Generating RR-Sets} \label{sec:newmethodforrrsets}

In our problem, we need to design a method to estimate $\pi(\vec{S})=\sum_{i\in [h]}\cpe(i)\cdot \sigma_i(S_i)$ for any solution $\vec{S}=(S_1,\dotsc,S_h)$ to the RM problem.
According to Section~\ref{sec:rrset}, a straightforward idea for estimating $\pi(\vec{S})$ is to generate a set $\mathcal{R}_i$ of random RR-sets for each advertiser $i\in [h]$ with $|\R_1|=|\R_2|=\dotsb=|\R_h|$, such that $\sigma_i(S_i)$ can be estimated using $\R_i$ for each $i\in [h]$. However, the estimation accuracy of this method is unsatisfactory, as the random variables in $\{X(S_i,R)\colon i\in [h]\wedge R\in \R_i\}$ have $h$ different distributions while the concentration bounds are generally sharper when the considered random variables are identically distributed. To overcome this hurdle, we propose a \textit{uniform sampling} method for generating a random RR-set, as described below:

\begin{enumerate}
    \item Sample a random advertiser $i\in [h]$ with probability proportional to $\mathit{cpe}(i)$.
    \item Generate an RR-set $R$ for advertiser $i$ selected in the first step using the edge probability $p^i_{(u,v)}$ for each edge $(u,v)\in E$.
\end{enumerate}

\vspace{-1ex}
Given $\vec{S}=(S_1,\dotsc,S_h)$ and a random RR-set $R$ generated as above, we define a random variable $\Lambda(\vec{S},R)$ such that $\Lambda(\vec{S},R)=1$ if $R$ is generated for certain advertiser $j\in [h]$ and \eat{$S_j\cup R\neq \emptyset$} $S_j\cap R\neq \emptyset$, and $\Lambda(\vec{S},R)=0$ otherwise. Let $\Gamma=\sum_{i\in [h]}\cpe(i)$. Then it follows that:
\laks{
\vspace{-.5ex}
\begin{lemma}
$\pi(\vec{S})=n\Gamma\cdot \E [\Lambda(\vec{S},R)]$.
\label{lem:samplingthm}
\end{lemma}
\vspace{-.5ex}
}
Given a set $\R$ of random RR-sets generated by using the uniform sampling method described above, Lemma~\ref{lem:samplingthm} suggests that $\widetilde{\pi}(\vec{S},\mathcal{R})=n\Gamma\cdot \sum\nolimits_{R\in \R} \Lambda(\vec{S}, R)/|\R|$~
is an unbiased estimation of $\pi(\vec{S})$. Moreover, as the random variables in $\{\Lambda(\vec{S}, R)\colon R\in \R\}$ follow the same distribution, it is possible to use sharper concentration bounds to improve the estimation accuracy. Similarly, we \ca{also have} $\widetilde{\pi}_i(S_i,\R)=n\Gamma\cdot \sum\nolimits_{R\in \R} \Lambda(S_i, R)/|\R|$
\ca{as} \eat{is} an unbiased estimation of $\pi_i(S_i)$ for any $i\in [h]$, where \laks{$\Lambda(S_i, R)=\min\{|S_i\cap R|, 1\}$} if $R$ is generated for advertiser $i$, and $\Lambda(S_i, R)=0$ otherwise.

\vspace{-2mm}
\subsection{One-Batch Sampling}
\label{sec:onebatchsampling}

With the uniform sampling method described above, a simple \textit{one-batch} sampling algorithm can be used to address the RM problem: we first generate a set $\mathcal{R}$ of RR-sets, then call the $\mathsf{RM\_with\_Oracle}$ algorithm with the function $\pi_i(\cdot)$ replaced by ${\wpi}_i(\cdot,\R)$ and the budget $B_i$ replaced by $(1+\varrho/2)B_i$ for all $i\in [h]$, where $\varrho\in (0,1)$ is an input parameter \ca{for bicriteria approximation}. \brn{Note that we have to seek  a bicriteria approximation for the RM problem as the optimal solution $\vec{O}$ may violate the budget constraint in the sampling space due to the sampling error. Therefore, it is hopeless to find an approximate solution as we can only derive an approximation ratio through sampling, \laks{unless we under-utilize the available budget, based on estimates and concentration bounds.\footnote{Recall, the effect of budget overshoot can be canceled out by the host using a ``corrected'' budget $B_i := B_i/(1+\varrho)$, if desired (see Section~\ref{sec:limitations}).}}} The following theorem shows that such a one-batch algorithm can return an approximate \eat{a bi-criteria approximate} solution to the RM problem when $|\R|$ is sufficiently large:



\begin{theorem}
	The one-batch approach described above can return a solution $\vec{S}^*$ satisfying $c_i(S_i^*)+\pi_i(S_i^*)\leq (1+\varrho)B_i$ for all $i\in [h]$ and $\pi(\vec{S}^*)\geq (\lambda-\epsilon)\mathrm{OPT}$ with probability at least $1-\delta$, as long as $|\mathcal{R}|\geq \theta_{max}=\max\{\hat{\theta}_{max}, \bar{\theta}_{max}\}$ by setting
	\begin{align*}
	&\hat{\theta}_{max}={\frac{2n}{\epsilon^{2}}\left(\lambda\sqrt{\ln \frac{4}{\delta}}+\sqrt{\lambda\left(\ln \frac{4}{\delta}+\sum\nolimits_{i\in[h]}{\mu_i}\ln\frac{\mathrm{e}n}{\mu_i}\right) } \right)^2 },\\
	&\bar{\theta}_{max}=\frac{8n\Gamma(1+\varrho)}{\varrho^2 B_{min}}(\ln \frac{4h}{\delta}+{\mu}\ln\frac{\mathrm{e}n}{{\mu}}),
	\end{align*}
	\brn{where $\mu_i$ is the maximum number of nodes that can be selected by advertiser $i$ without exceeding the \laks{relaxed} budget of $(1+\varrho)B_i$, $\mu=\max\{\mu_i\colon i \in [h]\}$, and $B_{min}=\min\{B_i \colon i\in [h]\}$.}
	\label{thm:upperboundofrrsets}
\end{theorem}

The proof of Theorem~\ref{thm:upperboundofrrsets} is highly non-trivial compared to the existing results for the traditional influence maximization problem~\cite{TangXS2014,TangSX2015}, as the RM problem is more complex. \brn{In a nutshell, the proof of Theorem~\ref{thm:upperboundofrrsets} shows that, when the number of RR-sets in $\R$ is sufficiently large (i.e., $\geq \theta_{max}$), the one-batch algorithm can achieve the claimed performance guarantee because all the following conditions simultaneously hold with high probability:}


\begin{enumerate}[(i)]
\item The optimal solution $\vec{O}$ is budget-feasible in the ``sampling space'', i.e., $c_i(O_i)+\wpi_i(O_i, \R)\leq (1+\varrho/2)B_i$ for all $i\in [h]$.
\item The approximate solution $\vec{S}^*$ is ``almost'' budget-feasible, i.e., $c_i(S^*_i)+\pi_i(S^*_i, \R)\leq (1+\varrho)B_i$ for all $i\in [h]$.
\item The approximate solution $\vec{S}^*$ satisfies the $\lambda-\epsilon$ approximation ratio, i.e., $\pi(\vec{S}^*)\geq (\lambda-\epsilon)\OPT$.
\end{enumerate}

Roughly speaking, Condition (i) ensures that the optimal solution is comparable to the approximate solution $\vec{S}^*$, as we require $\forall i\in [h]\colon c_i(S_i^*)+\wpi_i(S^*_i, \R)\leq (1+\varrho/2)B_i$ in searching $\vec{S}^*$; Conditions (ii)--(iii) ensure that $\vec{S}^*$ is a valid bi-criteria approximate solution.
\eat{\note[Laks]{Is (ii) supposed to use $\wpi_i$ instead of $\pi_i$?}

\vlt{the current (ii) is correct, it ensures that the solution is ``trully'' budget-feasible (not in the sampling space, as in (i)). -kai}}



\begin{algorithm}[t]
	\setlength{\hsize}{0.94\linewidth}
	$\lambda\gets$ the approximation ratio shown in Theorem~\ref{thm:finalapproxratio};\\
	$\delta'\gets \delta/4$ and compute $\theta_{max}$ by replacing $\delta$ with $\delta'$ in $\hat{\theta}_{max}$ and $\bar{\theta}_{max}$ defined in Theorem~\ref{thm:upperboundofrrsets}\;
 		$\theta_0\gets \frac{4n\Gamma(2+\varrho/3)}{\varrho^2 B_{min}}\ln \frac{h}{\delta^\prime}$;~$t_{\max}\gets \lceil \log_2 \frac{\theta_{max}}{\theta_0}\rceil$;~~$q\gets \ln \frac{h+2}{\delta^\prime t_{\max}}$\;\label{ln:initial}
	generate two sets $\mathcal{R}_1$ and  $\mathcal{R}_2$ of random RR sets, with $|\mathcal{R}_1|=|\mathcal{R}_2|=\theta_0$\;
	\While{true}{
		$\vec{S}^*, (\vec{T}^*_1, b_1,\gamma_1), (\vec{T}^*_2, b_2,\gamma_2)\leftarrow\mathsf{RM\_with\_Oracle}(\tau)$ with function $\pi_i(\cdot)$ replaced by $\widetilde{\pi}_i(\cdot,\mathcal{R}_1)$ and $B_i$ replaced by $(1+\frac{\varrho}{2})B_i$ for all $i\in [h]$\; \label{ln:callrmwithoracle}
		$z\gets \mathsf{SeekUB}(\vec{S}^*, \vec{T}^*_1, b_1, \gamma_1, \vec{T}^*_2, b_2, \gamma_2, b_{min}, \lambda, \mathcal{R}_1)$\;\label{ln:upperofoptimum}
		$\mathit{Feasible}\gets \mathrm{True};~B_{min}\gets \min\{B_i\colon i\in [h]\}$\;
		\ForEach{$i\in [h]$ \label{ln:starttestfeasible}}{
			$\mathit{UB}(S_i^*)\gets\bigg(\sqrt{ \frac{\widetilde{\pi}_i({S}_i^*,\mathcal{R}_2)|\mathcal{R}_2|}{n\Gamma} +\frac{q}{2}}+\sqrt{\frac{q}{2}}\bigg)^{2}\cdot\frac{n\Gamma}{|\mathcal{R}_2|}$\;
			\lIf {$\mathit{UB}(S_i^*)>(1+\varrho)B_i-c_i(S_i^*)$}{$\mathit{Feasible}\gets \mathrm{False}$\label{ln:endtestfeasible}}
		}
		$\mathit{LB}(\vec{S}^*)\gets\bigg(\Big(\sqrt{\frac{\widetilde{\pi}(\vec{S}^*,\mathcal{R}_2)|\mathcal{R}_2|}{n\Gamma}+\frac{2q}{9}}-\sqrt{\frac{q}{2}}\Big)^{2}-\frac{q}{18}\bigg)\cdot\frac{n\Gamma}{|\mathcal{R}_{2}|}$\; \label{ln:derivelbstar}
		$\mathit{UB}(\vec{O})\gets\bigg(\sqrt{\frac{z|\mathcal{R}_1|}{n\Gamma}+\frac{q}{2}}+\sqrt{\frac{q}{2}}\bigg)^{2}\cdot\frac{n\Gamma}{|\mathcal{R}_1|}$\;\label{ln:deriveubo}
		$\beta \gets \mathit{LB}(\vec{S}^*)/\mathit{UB}(\vec{O})$\;
		\lIf{$(\beta \geq \lambda -\epsilon \wedge \mathit{Feasible})\vee |\mathcal{R}_1|\geq \theta_{max}$ \label{ln:judgeappratio}}{\Return $\vec{S}^*$}
		double the sizes of $\mathcal{R}_1$ and $\mathcal{R}_2$ with new random RR sets\;\label{ln:doublerrsets}
	}
	\caption{$\mathsf{RM\_without\_Oracle}(\epsilon, \delta, \tau, \varrho)$}
	\label{alg:rmwithoutoracle}
\end{algorithm}

\vspace{-2ex}

\subsection{Progressive Sampling} \label{sec:progressivesampling}

Theorem~\ref{thm:upperboundofrrsets} implies that $\theta_{max}$ is an upper  bound on the required number of RR-sets
for guaranteed performance. In this section, we propose a progressive sampling algorithm in Algorithm~\ref{alg:rmwithoutoracle} that generates fewer RR-sets in practice without compromising the performance guarantee. The design of Algorithm~\ref{alg:rmwithoutoracle} is similar in spirit to the OPIM-C framework in~\cite{Tang2018} for the  influence maximization problem, but it involves more complex operations as RM problem has more stringent requirements on bounding the sampling errors.

Instead of generating $\theta_{max}$ RR-sets in one batch, Algorithm~\ref{alg:rmwithoutoracle} first generates two sets of RR-sets (i.e., $\mathcal{R}_1$ and $\mathcal{R}_2$) with $|\mathcal{R}_1|=|\mathcal{R}_2| =\theta_0$, where $\theta_0$ is much smaller than $\theta_{max}$. It then uses $\mathcal{R}_1$ as the input to the one-batch algorithm to find a solution $\vec{S}^*$ (Line~\ref{ln:callrmwithoracle}). Afterwards, it tests whether Conditions (i)--(iii) listed in Section~\ref{sec:onebatchsampling} can be satisfied by $\vec{S^*}$, $\R_1$ and $\R_2$ with high probability. Specifically, Line~\ref{ln:initial} ensures that $|\R_1|$ is sufficiently large such that Condition (i) can be satisfied; Lines~\ref{ln:starttestfeasible}--\ref{ln:endtestfeasible} check whether Condition (ii) can be satisfied, where $\UB(S_i^*)$ is an upper bound of $\pi_i(S_i^*)$ computed using the concentration bounds; Line~\ref{ln:judgeappratio} checks whether Condition (iii) is satisfied (i.e., whether $\LB(\vec{S}^*)/\UB(\vec{O}) \geq \lambda-\epsilon$), where $\LB(\vec{S}^*)$ and $\UB(\vec{O})$ are lower bound and upper bound of $\pi(\vec{S}^*)$ and $\pi(\vec{O})$ with high probability, respectively. When all the three conditions are satisfied, the solution $\vec{S}^*$ is returned immediately. Otherwise, the algorithm doubles the sizes of $\mathcal{R}_1$ and $\mathcal{R}_2$ and repeats the above process until a satisfying solution is returned or the number of generated RR-sets reaches $\theta_{max}$ (Lines~\ref{ln:judgeappratio}--\ref{ln:doublerrsets}). \brn{Although $\mathsf{RM\_without\_Oracle}$ may theoretically generate $\theta_{max}$ RR-sets in the worst case, our experimental results in Section~\ref{sec:pe} show that it runs very fast in practice.}

\brn{
A key challenge in Algorithm~\ref{alg:rmwithoutoracle} is that we need to make the upper bound $\mathit{UB}(\vec{O})$ and lower bound $\mathit{LB}(\vec{S}^*)$ as tight as possible, so that the condition in Line~\ref{ln:judgeappratio} can be met more easily, making the algorithm  more likely to generate  fewer RR-sets and stop early.}
\eat{
\note[Laks]{I agree tight bounds are important, but if you keep doubling the size of $\mathcal{R}_1$, it will eventually exceed $\theta_{max}$ and the algorithm will stop.}
\vlt{Please see the brown text revised above. This concern also suits for the OPIM-C framework in Jing's SIGMOD'18.}
}

\brn{
Although it is relatively easy to get $\mathit{LB}(\vec{S}^*)$ based on concentration bounds, finding a tight $\mathit{UB}(\vec{O})$ is non-trivial as $\vec{O}$ is unknown. To address this problem, Algorithm~\ref{alg:rmwithoutoracle} calls the $\mathsf{SeekUB}$ function to find an upper bound $z$ of $\widetilde{\pi}(\vec{O}, \mathcal{R}_1)$ (Line~\ref{ln:upperofoptimum}), which is further used to derive $\mathit{UB}(\vec{O})$ based on concentration bounds (Line~\ref{ln:deriveubo}). The $\mathsf{SeekUB}$ function adopts a novel method to find a tight upper bound of $\widetilde{\pi}(\vec{O}, \mathcal{R}_1)$ based on the special binary-search process of the $\mathsf{Search}$ algorithm. More specifically, as $\widetilde{\pi}(\cdot, \R_1)$ is also a monotone submodular function,  $\mathsf{ThresholdGreedy}$ called by the $\mathsf{Search}$ algorithm also satisfies Theorem~\ref{thm:boundofthresholdgreedy} in the sampling space with $\OPT$ replaced by $\widetilde{\pi}(\vec{O}, \mathcal{R}_1)$, which can be used for $\mathsf{SeekUB}$ to derive an upper bound of $\widetilde{\pi}(\vec{O}, \mathcal{R}_1)$. For example, when the $\mathsf{Search}$ algorithm returns $(\vec{T}_2^*, b_2,\gamma_2)$ with $b_2=0$, we can know from Theorem~\ref{thm:boundofthresholdgreedy} that $\widetilde{\pi}(\vec{T}_2^*,\R_1)\geq \frac{1}{2}(\widetilde{\pi}(\vec{O}, \mathcal{R}_1)-h\cdot\gamma_2)$, so $2\widetilde{\pi}(\vec{T}_2^*,\mathcal{R}_1)+h\gamma_2$ is an upper bound of $\widetilde{\pi}(\vec{O},\R_1)$ and it could be tighter than the naive upper bound of $\wpi(\vec{S}^*,\R_1)/\lambda$. Based on all the methods decribed above, we can get the following theorem:}




\vspace{-1ex}
\begin{theorem}
$\mathsf{RM\_without\_Oracle}(\epsilon, \delta, \tau, \theta_0, \theta_{max})$ returns a solution $\vec{S}^*$ satisfying  $c_i(S_i^*)+\pi_i(S_i^*)\leq (1+\varrho)B_i$ for all $i\in [h]$ and $\pi(\vec{S}^*)\geq (\lambda-\epsilon)\mathrm{OPT}$ with probability  at least $1-\delta$ for any $\delta \in (0,1)$, where~$\lambda$ is the approximation ratio shown in Theorem~\ref{thm:finalapproxratio}.
\label{thm:highprobratio}
\end{theorem}
\vspace{-1ex}

\brn{\ca{The intuition is that as in the proof of Theorem~\ref{thm:upperboundofrrsets}, the proof of Theorem~\ref{thm:highprobratio} similarly shows that the conditions (i)--(iii) given in Section~\ref{sec:onebatchsampling} can be satisfied by $\mathsf{RM\_without\_Oracle}$ with high probability.} \eat{The proof idea of Theorem~\ref{thm:highprobratio} is similar to that of Theorem~\ref{thm:upperboundofrrsets}, as it also tries to show that conditions (i)--(iii) that we have mentioned in Section~\ref{sec:onebatchsampling} can be satisfied by $\mathsf{RM\_without\_Oracle}$ with high probability.} The major difference is that, \ca{since} $\mathsf{RM\_without\_Oracle}$ adopts a ``trial-and-error'' approach and returns an approximate solution immediately in each trial if it judges that the current solution has already satisfied the performance guarantee, we need to show that the total probability \ca{of $\mathsf{RM\_without\_Oracle}$ making wrong judgements in all the trials is no more than $\delta$}. We  prove this by using concentration bounds.} \eat{Due to space constraint, we provide the detailed proof in \eat{Appendix~\ref{sec:othermissingproofs} of} \cite{RMA_report}.}

\toblue{\spara{Time complexity} We provide the theoretical time complexity of $\mathsf{RM\_without\_Oracle}$ (RMA), and of the algorithms of \cite{Aslay2017},  left open in \cite{Aslay2017}. RMA has an expected time complexity of $O\big(\frac{m \sum_{i\in[h]}\E[\pi_i(\{v^\ast\}) (\ln \frac{1}{\delta}+n\ln h)}{\epsilon^2 B_{min}}\big)$, where $v^\ast$ denotes a random node selected from $V$ with probability proportional to its in-degree. Both algorithms of \cite{Aslay2017} on the other hand have a time complexity of
$O\big(\frac{n (1 + \frac{\ln 1/\delta}{\ln n})(m + n) \ln n}{\epsilon^2})$. These results show that the running time of RMA is dominated by the factor $mn$ while the running time of the algorithms of \cite{Aslay2017} are dominated by the factor $n(m + n)$, translating to the superiority of the RMA algorithm in terms of asymptotic worst-case running time.} 




\begin{algorithm} [t]
    \lIf{$h=1$}{\Return $\widetilde{\pi}(\vec{S}^*,\mathcal{R}_1)/\lambda$}
    \lIf{$b_1<b_{min}$}{
        $z\gets 6\widetilde{\pi}(\vec{T}_2^*,\mathcal{R}_1)$}
    \If{$b_1\geq b_{min}\wedge \vec{T}^*_2\neq \emptyset$}{
        \lIf{$b_{2}=0$}{$z\gets 2\widetilde{\pi}(\vec{T}_2^*,\mathcal{R}_1)+h\gamma_2$} \label{ln:setupperbound1}
        \lIf{$b_{2}=1$}{$z\gets 6\widetilde{\pi}(\vec{T}_2^*,\mathcal{R}_1)+h\gamma_2$} \label{ln:setupperbound2}
    }
    \lIf{$b_1\geq b_{min}\wedge \vec{T}^*_2=\emptyset$}{
        $z\gets \widetilde{\pi}(\vec{T}_1^*,\mathcal{R}_1)/\lambda$}
    \Return $\min\{z, \widetilde{\pi}(\vec{S}^*,\mathcal{R}_1)/\lambda\}$\;
    \caption{$\mathsf{SeekUB}(\vec{S}^*, \vec{T}^*_1, b_1, \gamma_1, \vec{T}^*_2, b_2, \gamma_2, b_{min}, \lambda,\mathcal{R}_1)$}
    \label{alg:seekub}
\end{algorithm}

\revise{
\vspace{1mm}
\spara{Discussion} We note that although Algorithm~\ref{alg:rmwithoutoracle} borrows some ideas from the OPIM-C framework~\cite{Tang2018}, it embodies necessary and nontrivial extensions (e.g., the sampling method in Sec.~\ref{sec:newmethodforrrsets}) because OPIM-C was originally designed for the simpler Influence Maximization (IM) problem. We also note that a recent study~\cite{Guo_IM_2020} on the IM problem also used the OPIM-C framework, but it presented an algorithm, dubbed SUBSIM, to accelerate the generation of a single RR-set, which can also be used by RMA.} 
~\toblue{Another useful extension for improving the empirical performance of RMA is as follows. Before RMA returns a solution $\vec{S}^*$, check  whether the ratio of $\widetilde{\pi}(\vec{S}^*,\R_2)$ to $\widetilde{\pi}(\vec{S}^*,\R_1)$ is too small (e.g., less than 80\%); if so, then  generate more RR-sets to enlarge the sizes of $\R_1$ and $\R_2$ by a constant factor  (e.g., 15$\times$), and then repeat the solution-seeking process as before to find a new solution $\vec{S}^c$ using the new collections of RR-sets $\R_1$ and $\R_2$. Finally, return $\vec{S}^c$
if it also satisfies the stopping condition in Line~\ref{ln:judgeappratio}, and otherwise return $\vec{S}^*$. Clearly,  such an extension does not affect the  theoretical performance bound of RMA, and it enables RMA to potentially output a solution with an empirical  (instance-dependent) approximation ratio better than $\lambda-\epsilon$ (and also minimize biases), by possibly generating more RR-sets.}


%% file: 5-exp.tex





\vspace{-1ex}
\section{Performance Evaluation} \label{sec:pe}
In this section, we compare the performance of our algorithm with the state-of-the-art algorithms proposed in~\cite{Aslay2017}.  The performance of considered algorithms is evaluated in terms of their revenue, seeding costs and running time. All algorithms are implemented using C++ and all experiments are run on a Linux server with Intel Xeon 2.20GHz CPU and 192GB memory.


\vspace{-2ex}
\subsection{Experimental Setting} \label{sec:expsetting}

\spara{Datasets} We use several public datasets in our experiments. Flixster~\cite{Aslay2017} is from a social movie rating network website (www.flixster.com), where each node represents a user and two users are connected by a directed edge if they are friends, both rating the same movies. LastFM~\cite{barbieri2014influence} is a social network where people can specify their interests on music types and make friends. 
As both Flixster and LastFM have action logs that record users' activities of rating movies or music (i.e., “a log of past propagation” in \cite{barbieri2012topic}), we use the method provided in ~\cite{barbieri2012topic} to learn the topic-dependent influence probabilities (i.e., $\hat{p}_{u,v}^z$) on each edge $(u,v)\in E$. We also follow the settings in~\cite{Aslay2017} to set the default numbers $L=10$ and $h=10$ for the Flixster and LastFM datasets, and use the same topic distributions as that in~\cite{Aslay2017} for the Flixster dataset. The topic distributions used for LastFM are learned from its action logs. \revise{As a result, more than $95\%$ (resp. $77\%$) of the influence probabilities generated for Flixster (resp. LastFM) are positive.}

\begin{table}
\centering
\caption{Datasets}
\vspace{-4mm}
 \setlength{\tabcolsep}{4.0mm}{
\begin{tabular}{|c|c|c|c|}
\hline
 \bf{Dataset}&  \bf{$|V|$} &  \bf{$|E|$} & \bf{Type}   \\ \hline
{Lastfm}&  1.3K&  14.7K&  directed \\ \hline
 {{Flixster}}&  30K& 425K&  directed \\ \hline
 {{DBLP}}&  317K& 1.05M&  undirected \\ \hline
{{LiveJournal}}&  4.8M& 69M &  directed\\ \hline
\end{tabular}}
\label{tab:datasets}
\vspace{-3ex}
\end{table}

\begin{table}
\centering
\caption{Advertiser budgets and CPE values}
\vspace{-4mm}
 \setlength{\tabcolsep}{1.5mm}{
\begin{tabular}{|c|c|c|c|c|c|c|}
\hline
 \multirow{2}{*}{\bf{Dataset}}&  \multicolumn{3}{c|}{\bf{Budgets}}&\multicolumn{3}{c|}{ \bf{CPEs}}\cr
 \cline{2-7}
  &mean&max&min &mean&max&min\cr
 \hline
{Lastfm} &320&1200&100&1.5&2&1\\
 \hline
{{Flixster}}&10.1K&20K&6K&1.5&2&1\\
 \hline
\end{tabular}}
\label{tab:budgets}
\vspace{-1ex}
\end{table}


We also use the DBLP and LiveJournal Datasets to test the scalability of the implemented algorithms. DBLP~\cite{leskovec2014snap} is a collaboration network where each node represents an author and co-authors are adjacent in the network. LiveJournal~\cite{leskovec2014snap} is a free on-line blogging community where users declare friendship with each other. These two datasets are also used in~\cite{Aslay2017}. The details of our datasets used in the experiments are listed in Table~\ref{tab:datasets}. 
\spara{Seed Incentive Models} \revise{Similar to~\cite{Aslay2017}, we use three seed incentive models (i.e., node seeding cost models) in the experiments. Given a fixed constant $\alpha>0$ and any pair $(u,i)\in V\times [h]$, these models set the cost of node $u$ for advertiser $i$ as follows:}



\revise{
\begin{itemize}
	\item {Linear incentive model}: the cost of $u$ is proportional to its influence spread, i.e., $c_{i}(u)=\alpha \cdot\sigma_{i}(\{u\})$.
	\item {QuasiLinear incentive model}: the cost of $u$ is a quasi-linear function, i.e., $c_{i}(u)=\alpha \cdot \sigma_{i}(\{u\})\ln\left(\sigma_{i}(\{u\})\right)$. 
\item {SuperLinear incentive model}: the cost of $u$ is a quadratic function of its influence spread. i.e., $c_{i}(u)=\alpha \cdot (\sigma_{i}(\{u\}))^2$.
\end{itemize}}

\vspace{-1mm}
\spara{Baseline Algorithms} As mentioned in Section~\ref{sec:prelim-existing}, only Aslay et al.~\cite{Aslay2017} have addressed the revenue maximization problem considered in this paper \ca{and proposed TI-CSRM and TI-CARM that can be implemented in practice.}
Therefore, we use  TI-CSRM and TI-CARM to compare with our algorithm $\mathsf{RM\_without\_Oracle}$ (RMA for short) (Section~\ref{sec:progressivesampling}).
\eat{\toblue{In our implementation, we also employ a ``boosting method'' to improve the empirical performance of RMA, as described in the following. Before RMA returns a solution $\vec{S}^*$, it checks whether the ratio of $\widetilde{\pi}(\vec{S}^*,\R_2)$ to $\widetilde{\pi}(\vec{S}^*,\R_1)$ is too small (e.g., less than 85\%); if yes, then RMA generates more RR-sets to enlarge the sizes of $\R_1$ and $\R_2$ by a constant times (e.g., 15 times), and then repeats the solution-seeking process as before to find a new solution $\vec{S}^c$ using the new collections of RR-sets $\R_1$ and $\R_2$. Finally, $\vec{S}^c$ is returned instead of $\vec{S}^*$ if it also satisfies the stopping condition in Line~\ref{ln:judgeappratio} (otherwise $\vec{S}^*$ is returned). It can be easily seen that such a boosting method does not affect any theoretical performance bound of RMA, but it makes it possible to output a solution with a practical (instance-dependent) approximation ratio better than $\lambda-\epsilon$ (and also avoid large biases), by generating more RR-sets.}}


\spara{Parameter Settings} In all the experiments, we set $\epsilon=0.02$, $\delta=1/n$, $\varrho=0.1$ and $\tau=0.1$ for our RMA algorithm unless otherwise stated, \revise{so RMA always returns an approximate solution with non-negative approximation ratios, according to Theorem~\ref{thm:highprobratio}}. However, TI-CSRM and TI-CARM algorithms cannot terminate successfully on all the four datasets due to memory issues and high running time when  $\epsilon$ is set to $0.02$, where $\epsilon$ is the parameter in Eqn.~\eqref{eqn:boundforsampledgreedy}. Therefore, we follow the same setting in~\cite{Aslay2017} to set their parameter $\epsilon=0.1$ for the {Flixster} and {LastFM} datasets, and set $\epsilon=0.3$ for {DBLP} and {LiveJournal}. 
\revise{
In all experiments, for fair comparison, we set the budget input to each advertiser in TI-CSRM and TI-CARM to $(1+\varrho) \times$ 
the budget to the same advertiser in the RMA algorithm, due to the consideration that RMA is a bi-criteria approximation algorithm. With this implicit rule, we will only 
cite the budget setting of TI-CSRM and TI-CARM in the experiments. Note that this budget setting is practically equivalent to the setting that the budget used by RMA is $(1+\varrho)^{-1}$ fraction of that used by TI-CSRM and TI-CARM.} \toblue{In all our experiments, we measure the revenue of the implemented algorithms by using $10^7$ RR-sets, generated independently of the considered algorithms. Alternatively, we could measure the revenue using (the much slower) Monte-Carlo simulations, but we found in the experiments that the accuracy of these methods does not have a noticeable difference as long as the samples used are independent of the considered algorithms.}
%
\revise{Besides conducting experiments under the implicit parameter settings described above, we will also study the impact of varying these parameters in Sec.~\ref{sec:impactofpara}. }


\eat{
Note that these parameter settings are in favor of TI-CSRM and TI-CARM due to the following reasons. \redc{First, TI-CSRM and TI-CARM do not have a performance guarantee respecting budget-feasibility in practice (see Sec.~\ref{sec:prelim-existing}), while RMA must achieve the provable performance bounds shown in Theorem~\ref{thm:highprobratio}.} Second, TI-CSRM and TI-CARM are granted larger budgets than RMA in the experiments, which is propitious to their revenue performance.
}

\eat{\note[Laks]{In view of the fixes recently applied to the proofs, I have removed the criticism that was here and have rephrased it. BTW, the claim that the budget input to RMA is $(1+\varrho)$ times smaller than for TI-CSRM is misleading since RMA, being a bicriteria approximation algorithm, is allowed to exceed its budget by the same factor, making the field level. }

\vlt{We agree.-Kai}}

\vspace{-2ex}
\subsection{Experimental Results}


\begin{figure}[!t]
    \centering
    \includegraphics[width=0.485\textwidth]{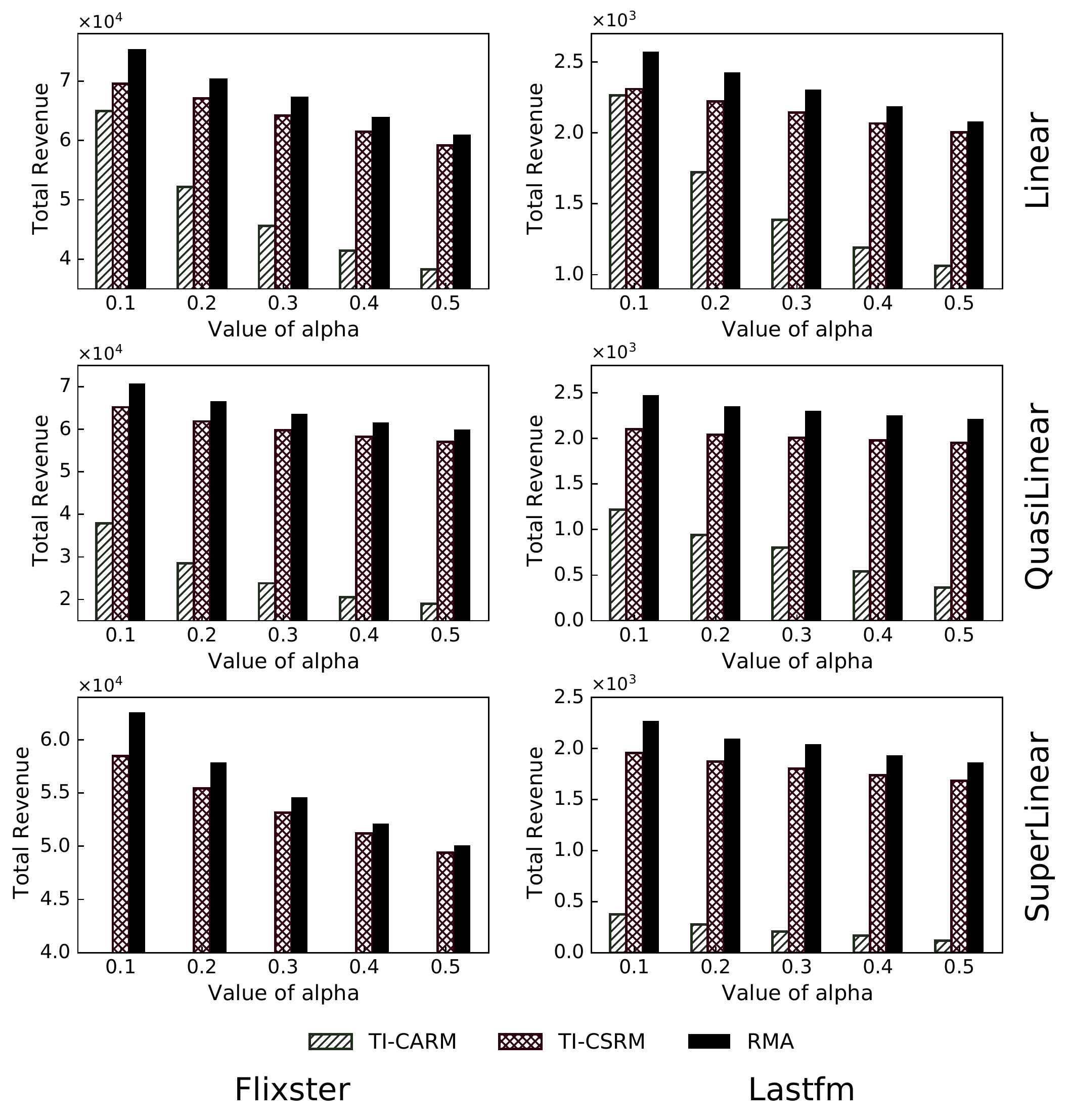}
    \vspace{-0.8cm}
    \caption{Total revenue as a function of $\alpha$, on $\texttt{\texttt{Flixster}}$(left) and $\texttt{Lastfm}$(right), for linear (top), quasi-linear (middle), and super-linear (bottom) incentive models.}
    \label{fig:totalrevenue}
    \vspace{-3mm}
\end{figure}

%

\subsubsection{General Comparisons} \label{sec:pegeneralcompare}

In this section, we compare the algorithms under different node cost models on LastFM and Flixster datasets. Following the experimental settings in~\cite{Aslay2017}, the advertisers are assigned heterogeneous budgets and CPE values for seed selection in TI-CARM and TI-CSRM, as shown in Table~\ref{tab:budgets}. We first plot the the revenue performance of TI-CARM, TI-CSRM and the RMA algorithm in Fig.~\ref{fig:totalrevenue} under the linear, quasilinear and superlinear cost modes, with varying levels of the parameter $\alpha$. It can be seen that the revenue of all algorithms decreases when $\alpha$ increases, which can be explained by the fact that the costs of all nodes increase with $\alpha$, so fewer seed nodes can be selected by all algorithms when $\alpha$ gets larger. 
The results in Fig.~\ref{fig:totalrevenue} also reveal that our RMA algorithm consistently outperforms TI-CARM and TI-CSRM algorithms under all three seed incentive models and for different values of $\alpha$. Specifically, the RMA algorithm can achieve up to \toblue{15.81$\times$  (resp. 17.68\%)}  gain on the revenue compared to TI-CARM  (resp. TI-CSRM).  \revise{Moreover, we observe that TI-CARM performs very poorly under the superlinear cost model. The intuition is that TI-CARM greedily selects elements purely based on their marginal gains while neglecting the costs, so it may quickly meet an element violating the budget constraint and hence terminate with very few seeds selected. This situation is analogous to choosing items without considering their costs in the traditional knapsack problem. These results demonstrate the effectiveness of our methods used in RMA for revenue optimization.}

\begin{figure}[!t]
    \centering
    \includegraphics[width=0.485\textwidth]{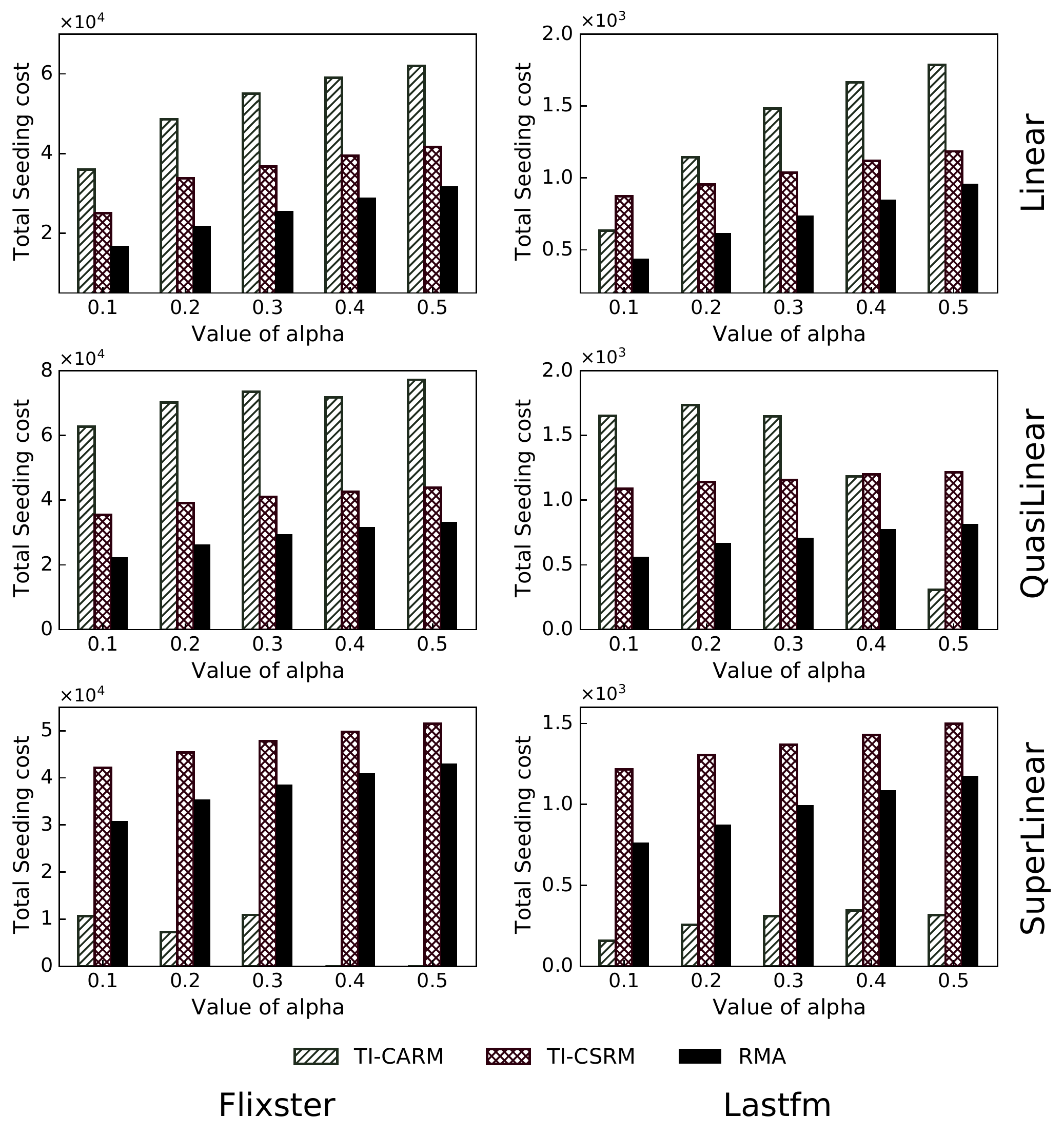}
    \vspace{-0.8cm}
    \caption{Total seeding cost as a function of $\alpha$, on $\texttt{\texttt{Flixster}}$(left) and $\texttt{Lastfm}$(right), for linear (top), quasi-linear (middle), and super-linear (bottom) incentive models.}
    \vspace{-3mm}
    \label{fig:totalcost}
\end{figure}

\revise{
In Fig.~\ref{fig:totalcost}, we plot the total seed costs (i.e., the amount paid by all advertisers for incentivizing seeds) assigned by different algorithms. RMA produces consistently lower seeding cost than TI-CSRM and the lowest cost for most cases under the linear and quasilinear models.
\eat{It can be seen that RMA produces the lowest total seed costs for most cases under linear and quasilinear models. Meanwhile, we observe that} Under the superlinear cost model, TI-CARM achieves very low seeding costs, which can be explained by a reason similar to that described above for the revenue performance of TI-CARM. \eat{These results demonstrate that RMA can achieve a superior expected revenue under a more economic seed incentive cost, compared to TI-CARM and TI-CSRM.} These experiments show that RMA achieves the best revenue performance. Its seeding cost is always lower than TI-CSRM. While TI-CARM has lower costs in some settings (especially under the superlinear model), that does not translate to better revenue.}

\revise{In Fig.~\ref{fig:seedsize}, we study the impact of $\alpha$ on the number of seeds selected by the implemented algorithms under the linear cost model, where the other settings are the same with Fig.~\ref{fig:totalrevenue}. It can be seen that the seed size decreases when $\alpha$ increases, as the node costs increase and hence fewer nodes can be selected under given budgets. While the seed set sizes for RMA and TI-CSRM are comparable, TI-CARM only manages to select very few seeds.}

\eat{First, as explained in Sec.~\ref{sec:expsetting} and Sec.~\ref{sec:intro}, the budget input to the RMA algorithm for any advertiser is always $1+\varrho$ times smaller than that to the TI-CSRM and TI-CARM algorithms to achieve provable budget-feasibility, while TI-CSRM and TI-CARM do not provide a method to guarantee that their algorithms are actually budget-feasible. Therefore, the RMA algorithm has a more stringent constraint for selecting seed nodes.
Second, as
The RMA algorithm achieves better revenue performance than TI-CSRM and TI-CARM (as shown in Fig.~\ref{fig:totalrevenue}), so it has to spend smaller seed costs to meet the budget constraint, even if all algorithms have the same budget input.}

In Table~\ref{tab:runningtimelinearincentive}, we show the running time of the implemented algorithms under the linear cost model. It can be seen that RMA runs faster than TI-CARM and TI-CSRM under all the settings of $\alpha$ \toblue{(1.04$\times$ to 38.5$\times$ faster)}, which can be explained by the fact that our RMA algorithm has leveraged the sampling algorithms proposed in Section~\ref{sec:rmwithoutoracle} to reduce the number of RR-sets to be generated, while still achieving the guaranteed performance ratio. \revise{The experimental results on the running time of the implemented algorithms under the other cost models are qualitatively similar and hence are omitted due to space constraints. It is also noted that~\cite{Aslay2017} exhibits faster running time than that shown in Table~\ref{tab:runningtimelinearincentive}, because only half of the budgets shown in Table 2 of~\cite{Aslay2017} were used in their experiments, resulting in fewer selected nodes and hence faster running time.}


\begin{figure}[!t]
    \centering
    \includegraphics[width=0.485\textwidth]{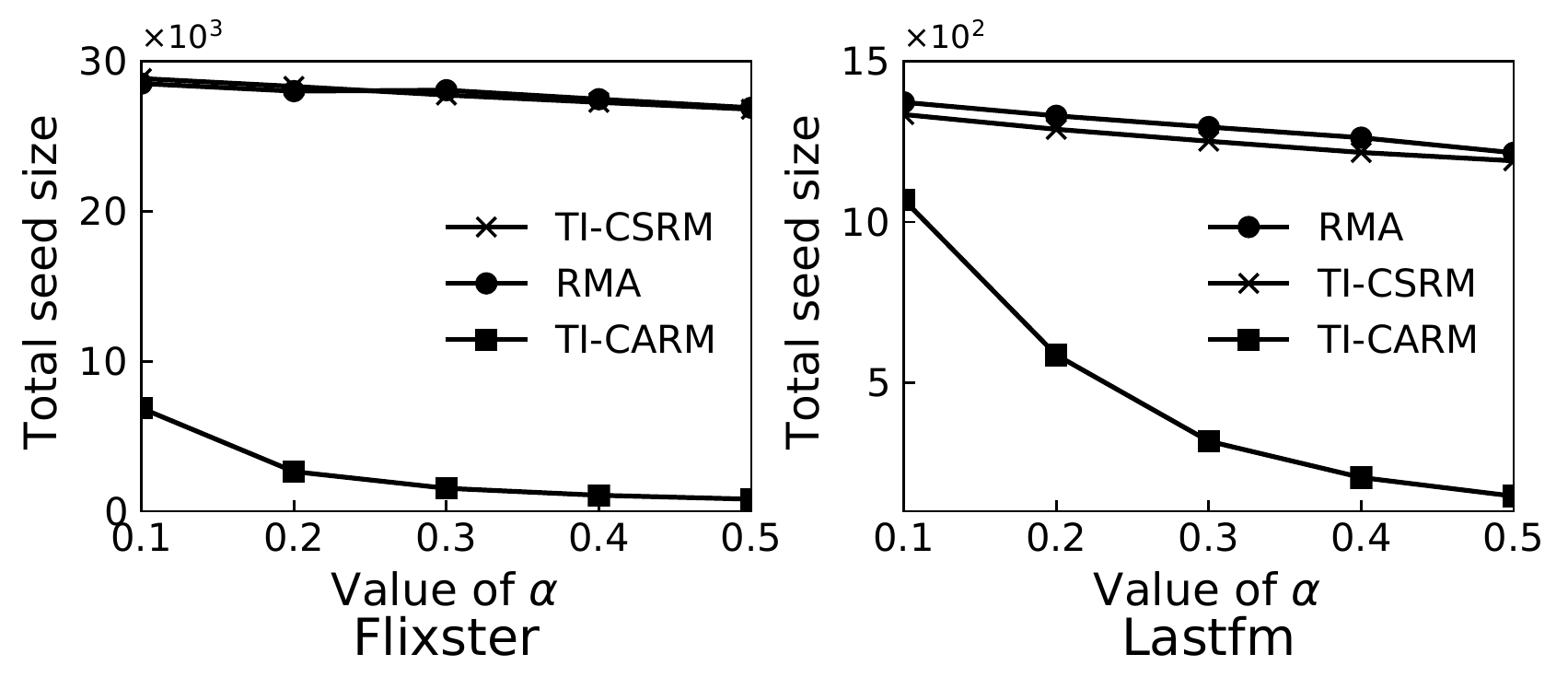}
    \vspace{-0.8cm}
    \caption{Impact of $\alpha$ on seed size}
    \label{fig:seedsize}
\vspace{-2mm}
\end{figure}

\begin{figure}[!t]
    \centering
    \vspace{-3mm}
    \includegraphics[width=0.485\textwidth]{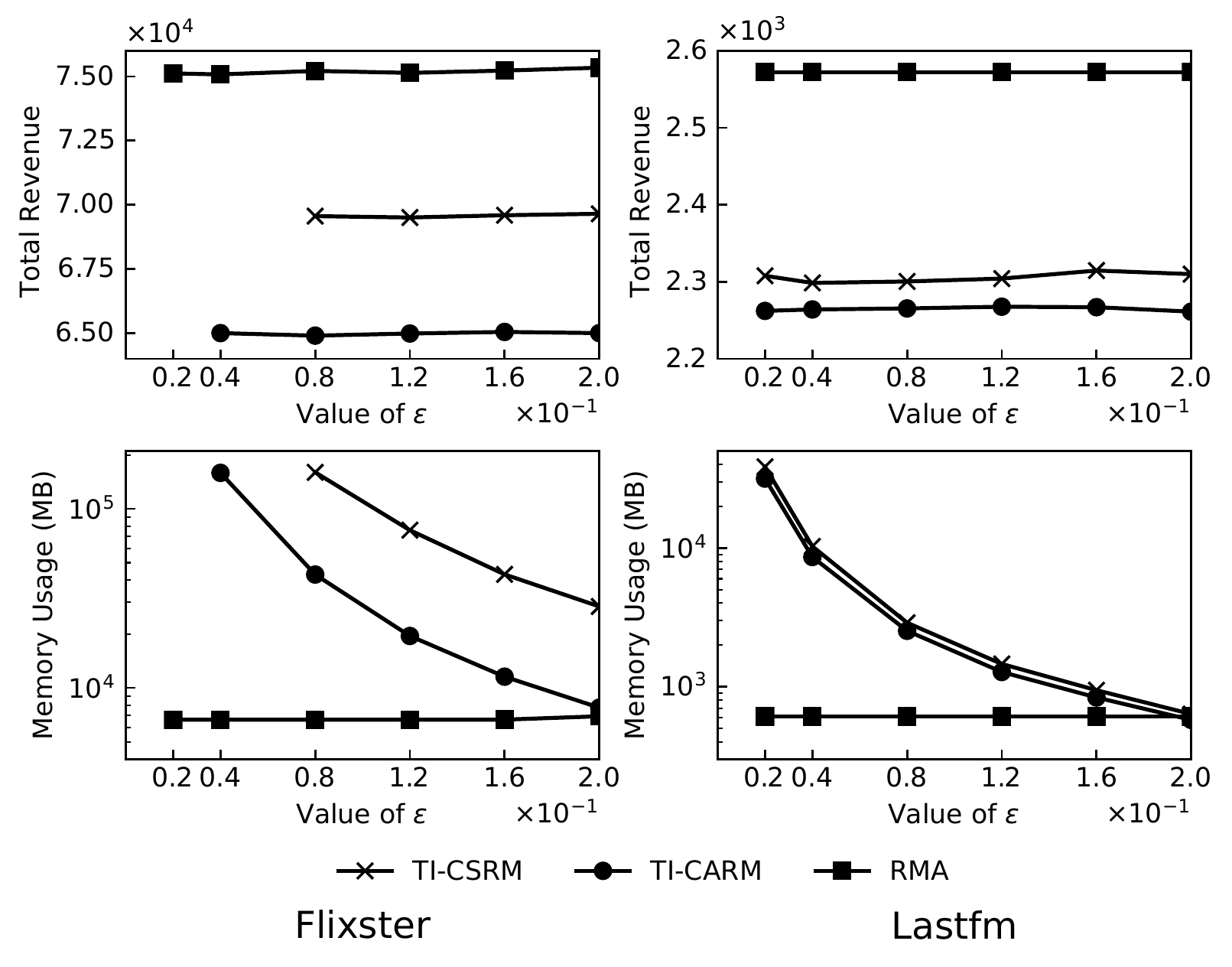}
    \vspace{-0.8cm}
    \caption{The impact of $\epsilon$ on revenue and memory usage} 
    \label{fig:eps}
\end{figure}

\eat{, which can be found in our online technical report~\cite{techreport} due to the space constraint.}

\eat{\note[Laks]{If these results are available, they can be included in the TR.}

\vlt{Yes, they are already included. -Kai}}


\begin{table}[!t]
\centering
\vspace{-4mm}
\caption{Running time (seconds) under linear cost model}
\vspace{-3mm}
 \setlength{\tabcolsep}{2.30mm}{
\begin{tabular}{|c|c|c|c|c|c|}
\hline
\texttt{Flixster} &$\alpha=0.1$&$0.2$&$0.3$&$0.4$&$0.5$ \\ \hline
RMA&736&738&664&695&681\\ \hline
TI-CARM&3803&1609&1074&880&710\\ \hline
TI-CSRM&16255&18798&25572&24109&23473\\ \hline
\hline
\texttt{Lastfm}&$\alpha=0.1$&$0.2$&$0.3$&$0.4$&$0.5$\\
\hline
RMA&26&25&23&24&23\\ \hline
TI-CARM&108&91&75&65&56\\ \hline
TI-CSRM&130&147&145&152&153\\ \hline
\end{tabular}
}
\label{tab:runningtimelinearincentive}
\vspace{-2mm}
\end{table}

\begin{figure*}[!t]
\centering
\includegraphics[width=1\textwidth]{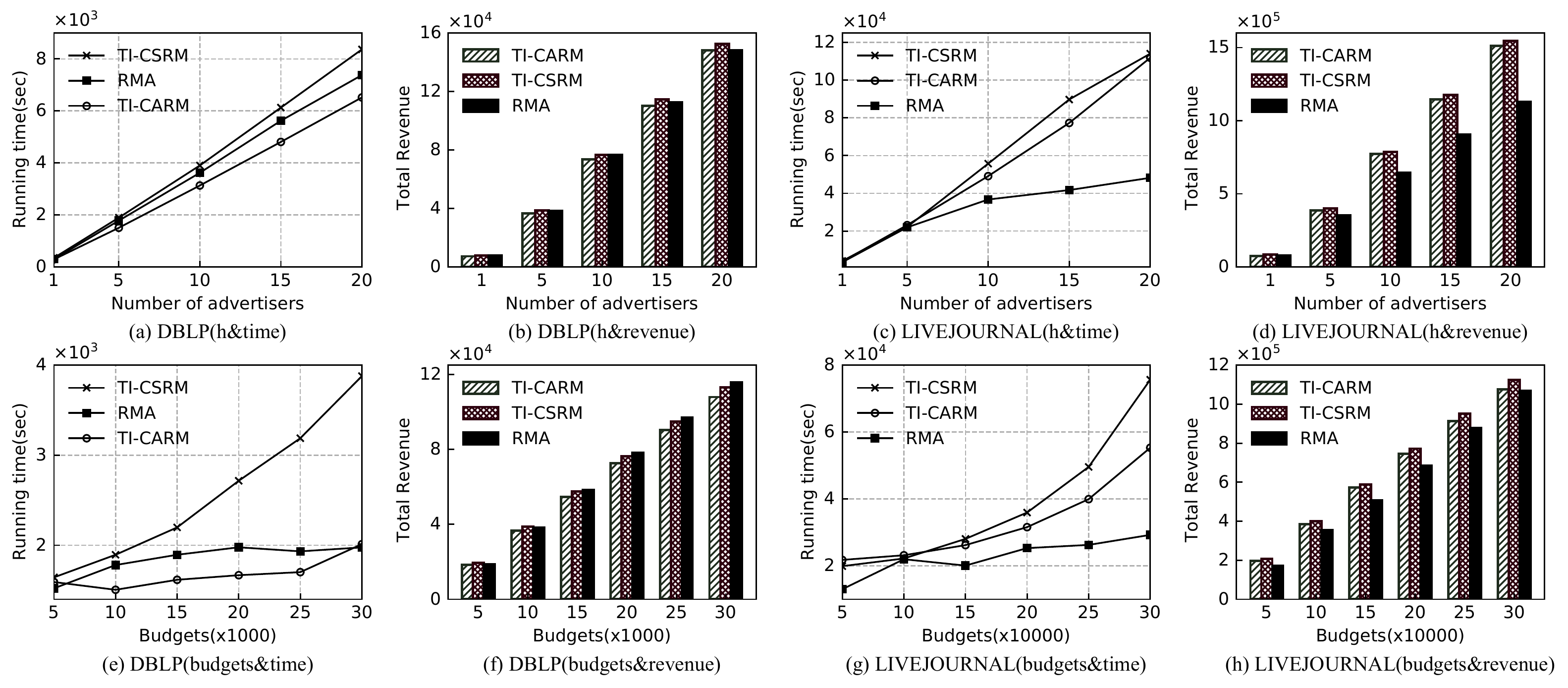}
\vspace{-7mm}
\caption{Running time and total revenue of RM, TI-CARM and TI-CSRM on \texttt{DBLP} and \texttt{LiveJournal}}
\vspace{-5mm}
\label{fig:scalability}
\end{figure*}

\vspace{-2mm}
\revise{
\subsubsection{Impact of Parameters} \label{sec:impactofpara}
Fig.~\ref{fig:eps} compares the revenue and memory consumption of implemented algorithms by varying $\epsilon$ from $0.02$ to $0.2$ in Flixster and LastFM, using the linear cost model, with $\alpha=0.1$. The results reveal that the revenue of RMA does not vary much and it consistently outperforms both TI-CSRM and TI-CARM over the range of values of $\epsilon$ considered. The reason is as follows. Although RMA has a theoretical approximation ratio of $\lambda-\epsilon$ (Theorem~\ref{thm:upperboundofrrsets}) which is affected by $\epsilon$, this ratio is just worst case and the actual performance ratio of RMA on specific datasets could be much better. Indeed, we observed in our experiments that RMA can practically achieve an approximation ratio $\beta$,  where $\beta$ is even larger than $\lambda$ in most cases, so its revenue performance is quite ``robust'' to the variation of $\epsilon$ due to the stopping rule in Line~\ref{ln:judgeappratio} of Algorithm~\ref{alg:rmwithoutoracle}. However, Fig.~\ref{fig:eps} also shows that the memory consumption of TI-CARM and TI-CSRM significantly increases due to the large number of generated RR-sets with  decreasing  $\epsilon$, which eventually causes memory overflow problems on the Flixster dataset when $\epsilon \leq 0.02$ (for TI-CARM) or $\epsilon\leq 0.04$ (for TI-CSRM).}

\eat{Additionally, we have conducted more experiments to study the impact of $\varrho$ and $\tau$, and the results are deferred to \cite{RMA_report}~due to space constraints. In a nutshell, these experiments demonstrate that our setting on parameters $\epsilon,\varrho$, and $\tau$, described in Sec.~\ref{sec:expsetting}, enables TI-CSRM and TI-CARM to output a solution under reasonable memory consumption and running time for  a fair comparison with RMA.}


\begin{figure}[t]
    \centering
    \includegraphics[width=0.48\textwidth]{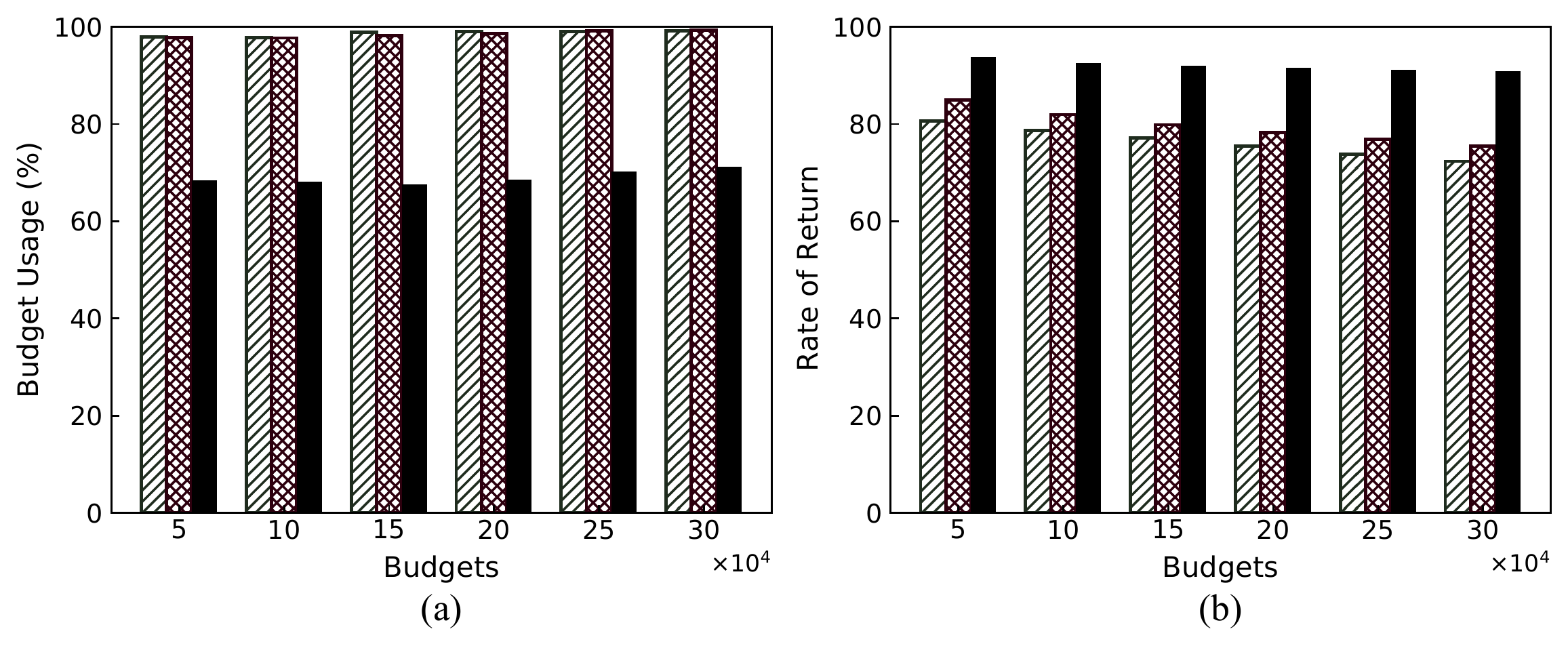}
    \vspace{-0.8cm}
    \caption{Budget usage and rate of return of RMA, TI-CARM, TI-CSRM on LiveJournal}
    \label{fig:rate}
\vspace{-2mm}
\end{figure}

\vspace{-4mm}
\subsubsection{Scalability Test}
In this section, we follow the same settings as those in~\cite{Aslay2017} to test the scalability of the implemented algorithms on DBLP and LiveJournal datasets. As there are no action logs for these two datasets, we cannot use the method proposed in~\cite{barbieri2012topic} to learn the influence probabilities. So we follow~\cite{Aslay2017} to use the Weighted-Cascade model, i.e., setting $p_{u,v}^i=1/|N^{in}(v)|$ for all $i\in [h]$ and $(u,v)\in E$, where $N^{in}(v)$ denotes the set of in-neighbors of node $v$. We also follow~\cite{Aslay2017} to use the Linear incentive model with $\alpha=0.2$ and identical budgets input to all advertisers. \toblue{Note that advertisers usually have heterogeneous budgets and the influence propagation depends on specific topics/items in practice (like the settings adopted in Fig.~\ref{fig:totalrevenue}). So the settings of uniform budgets and Weighted-Cascade model in this section are less practical and are only used for a fair comparison with~\cite{Aslay2017}.}

\toblue{In Figs.~\ref{fig:scalability}(a)--(d), we compare the running time and revenue of the implemented algorithms by scaling the number of advertisers $h$ from 1 to 20, where the budget of each advertiser is set to 10K (100K) for DBLP (resp. for LiveJournal). In Figs.~\ref{fig:scalability}(e)--(h), we compare the  algorithms by scaling the budget of advertisers, with the number of advertisers  fixed to 5. Notice that fixing number of advertisers and scaling up budgets is similar to increasing number of advertisers with a fixed budget, w.r.t. the number of seeds selected.}

\toblue{It can be seen from Fig.~\ref{fig:scalability} that RMA runs faster than TI-CARM and TI-CSRM in almost all cases.
\eat{On  DBLP, TI-CARM generally performs the worst on revenue, while RMA achieves approximately the same revenue as  TI-CSRM.}  On DBLP, all three algorithms attain almost the same revenue.
However, on LiveJournal, although RMA still runs faster than TI-CSRM/TI-CARM, it  achieves  smaller revenue than TI-CSRM/TI-CARM. We explain the reason below.  Note that the budget used by RMA is $(1 + \varrho)^{-1}$ times that used by TI-CSRM and TI-CARM, where we set $\varrho=0.1$. Although this does not affect RMA's superiority  in most cases (as shown in other figures), it could occasionally make TI-CSRM/TI-CARM perform better on revenue in some networks (depending on network structures and  propagation models). This happens especially when the total budget is large, causing a large budget overshoot for a given $\varrho$. 
We analyzed the \textit{rate of
actual budget usage}  $(\pi(\vec{S})+\sum_{i\in [h]}c_i(S_i))/(\sum_{i\in [h]}B_i)$ and the \textit{rate of return}  ${\pi(\vec{S})}/{[\pi(\vec{S})+\sum_{i\in [h]}c_i(S_i)]}$ of all algorithms for Fig.~\ref{fig:scalability}(h): results shown in Fig.~\ref{fig:rate} (results for Fig.~\ref{fig:scalability}(d) are similar). It can be seen from Fig.~\ref{fig:rate} that, RMA uses smaller budgets than TI-CSRM/TI-CARM, while its rate of return is clearly higher. This implies that RMA is much more ``profitable'' than TI-CSRM/TI-CARM, which could be important  from a practical point of view.}

\eat{
\vspace{-2mm}
\subsubsection{Scalability Test}


In this section, we follow the same settings as those in~\cite{Aslay2017} to test the scalability of the implemented algorithms on DBLP and LiveJournal datasets. As there are no action logs for these two datasets, we cannot use the method proposed in~\cite{barbieri2012topic} to learn the influence probabilities. So we follow~\cite{Aslay2017} to use the Weighted-Cascade model, i.e., setting $p_{u,v}^i=1/|N^{in}(v)|$ for all $i\in [h]$ and $(u,v)\in E$, where $N^{in}(v)$ denotes the set of in-neighbors of node $v$. We also follow~\cite{Aslay2017} to use the Linear incentive model with $\alpha=0.2$ and identical budgets input to all advertisers. \toblue{Note that advertisers usually have heterogeneous budgets and influence propagation depends on specific topics/items in practice (like the settings adopted in Fig.~\ref{fig:totalrevenue}). So the settings of uniform budgets and Weighted-Cascade model in this section are less practical and are only used for a fair comparison with~\cite{Aslay2017}. Without harming the fairness of comparison, we also allow the value of $\varrho$ to range from 0.01 to 0.02 in this section, such that $\varrho$ decreases with the increment of total budgets to avoid a large budget overshoot (a detailed discussion on $\varrho$ can be found in \cite{RMA_report}).}


%
%

In Figs.~\ref{fig:scalability}(a)--(d), we compare the running time and revenue of the implemented algorithms by scaling the number of advertisers (i.e., $h$) from 1 to 20, where the budget of each advertiser is set to 10K and 100K for DBLP and LiveJournal, respectively. The experimental results show that the running time of all algorithms grows with the number of advertisers, \laks{since as the effective number of seeds to be selected grows with the number of advertisers, which in turn necessitates  more random RR-sets to be generated. \eat{However, the RMA algorithm is \rdc{one to two orders of magnitude} faster than the TI-CSRM and TI-CARM algorithms under all settings of $h$, and it also consistently outperforms TI-CSRM and TI-CARM on the revenue.}
\toblue{It can be observed that the RMA algorithm consistently outperforms TI-CSRM and TI-CARM on the revenue, under approximately the same running time. Specifically, compared to TI-CSRM and TI-CARM, RMA can achieve up to 3.23\% gain on the revenue, which demonstrates the superiority of our RMA algorithm.}} 

In Figs.~\ref{fig:scalability}(e)--(h), we compare the implemented algorithms by scaling the budget input to advertisers, where the the number of advertisers is fixed to 5. \laks{Notice that fixing number of advertisers and scaling up budgets is similar to increasing number of advertisers while fixing the budget, w.r.t. the number of seeds selected.} We can see that the revenue of all algorithms increases with the budget, as all the algorithms can select more seed nodes when the budget increases. However, the running time of our RMA algorithm slightly decreases when the budget increases. This can be explained by the reason that, \hkr{both the parameters $\theta_0$ and $\theta_{max}$ in Algorithm~\ref{alg:rmwithoutoracle} decrease when the budget increases, which results in the phenomenon that the RMA algorithm generates fewer random RR-sets before it stops.} Overall, the results in Figs.~\ref{fig:scalability}(e)--(h) show similar trends with those in Figs.~\ref{fig:scalability}(a)--(d), \toblue{and the RMA algorithm outperforms TI-CSRM and TI-CARM both on the running time and on the revenue performance (up to 4.5\% gain on revenue).}}

\vspace{-4mm}
\subsubsection{Studying the Scenario with a Holistic Demand} \label{sec:holisticdemand}
In this section, we consider a  practical scenario where social advertising demands are controlled holistically. An advertiser budget $B_i$ includes seeding costs as well as user engagements. Thus, $B_i/(cpe(i)*n)$ is the maximum percentage of user engagements the advertiser can expect, so we can regard this as a proxy for the demand from   advertiser $i$.
We use $M=\sum_{i\in [h]}M_i$ to denote the total demand of a social advertising market, where $M_i=B_i/(n\cdot\cpe(i))$. For simplicity, we assume $\cpe(i)=1$ for all $i\in [h]$. In Fig.~\ref{fig:demand}(a)-(b), we use the Flixster dataset to study the impact of $M$ on revenue and total seed cost under the linear cost model, where we set $h=10$, $\alpha=0.1$, and the individual demands $M_i: i\in [h]$ are all randomly generated such that they sum to $M$. The results in Fig.~\ref{fig:demand}(a)-(b) show that, the revenues of all algorithms increase with $M$, as more elements can be selected under a larger demand, while RMA always achieves a better revenue with smaller seed costs than the other algorithms, for all values of $M$ tested. In Fig.~\ref{fig:demand}(c)-(d), we further study the relationship between parameter $\alpha$ and total revenue and seeding cost, for a fixed total demand $M=2.5$, with  all other settings being the same as those in Fig.~\ref{fig:demand}(a)-(b). The results 
show that (i) RMA outperforms the other algorithms again  and (ii) the revenue of all algorithm decreases when $\alpha$ increases, since the node costs increase with $\alpha$, causing fewer seed nodes to be selected. In summary, the results in Fig.~\ref{fig:demand} demonstrate that from the perspective of total advertising demand, RMA still exhibits a performance superior to the baselines.
\eat{
the superiority of RMA maintains when the social advertising demands are controlled holistically, and the change of $\alpha$ does not reverse the superiority of RMA.}

\begin{figure}[!t]
    \centering
    \includegraphics[width=0.485\textwidth]{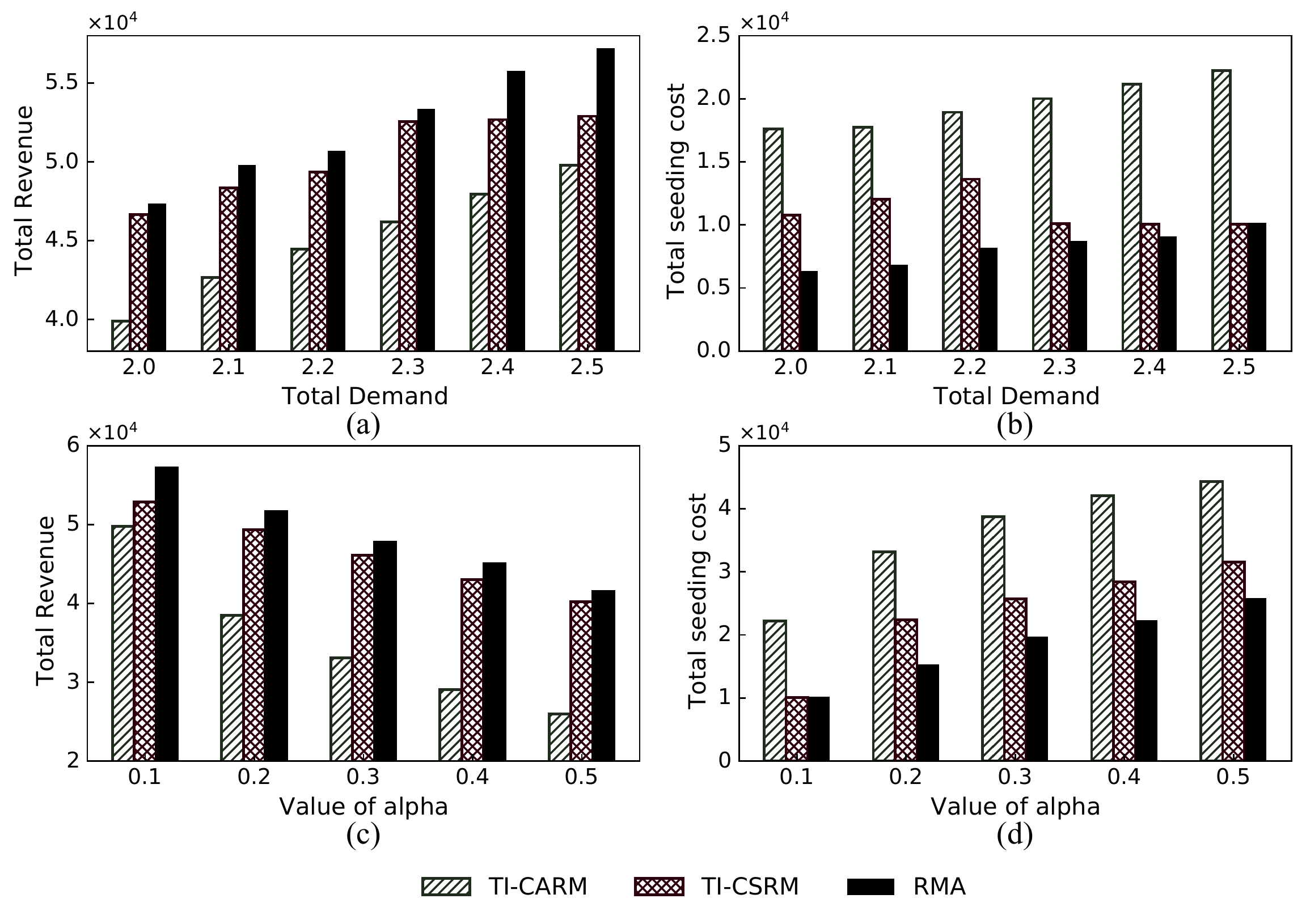}
    \vspace{-0.8cm}
    \caption{Comparing revenue and seed cost under the case where the advertising demand is controlled holistically.}
    \vspace{-2mm}
    \label{fig:demand}
\end{figure}

%% file: 6-related.tex
\vspace{-3mm}
\section{Related Work} \label{sec:relatedwork}

\eat{\note[Laks]{References need to be beefed up. Since references don't count toward space, we have no excuse not to cite any major works including targeted IM, SSA/DSSA, and various other related works that Jing can tell you about. We need to cite them. }
\jing{RW is revised by citing more relevant works and discussing our advantages or differences compared with existing studies.} }

\spara{Influence Maximization}
\citet{Kempe2003} study influence maximization (IM) where the aim is to select $k$ seed nodes in a social network such that the expected spread is maximized. They propose a simple greedy algorithm with $(1-1/e-\epsilon)$ approximation, but their algorithms are based on Monte-Carlo sampling hence have high time complexity.
Since then, there has been considerable research on improving algorithms for IM \cite{Borgs2014,Tang2018,TangSX2015,TangXS2014,ChenWW2010,Arora_benchmark_2017,Chen_degreeDiscount_2009,Chen_LDAG_2010,Cheng_IMRank_2014,Cheng_StaticGreedy_2013,Cohen_SKIM_2014,Galhotra_EaSyIM_2016,Galhotra_ASIM_2015,Goyal_CELF_2011,Goyal_infMax_2011,Goyal_SIMPATH_2011,Huang_SSA_2017,Jung_IRIE_2012,Lee_hop_2014,Leskovec_CELF_2007,Lu_IMstreaming_2015,Nguyen_DSSA_2016,Ohsaka_prunedMC_2014,Zhou_UBB_2014,Zhou_UBLF_2013,Wang_IM_2017,Guo_IM_2020,Ohsaka_IMsolution_2020,bian2020efficient}. In particular, Borgs et al.~\cite{Borgs2014} proposed  Reverse-Reachable Sets that can efficiently estimate influence spread with accuracy guarantee, based on which several studies~\cite{Tang2018,TangSX2015,TangXS2014,Nguyen_DSSA_2016,Guo_IM_2020} propose more efficient algorithms for IM  while still achieving $(1-1/e-\epsilon)$-approximation. Moreover, several variants of IM have been studied, such as topic-aware~\cite{barbieri2012topic,Chen_TAI_2015}, competition~\cite{Bharathi_CIM_2007,Lu_comIC_2015}, adaptive solutions~\cite{Han_AIM_2018,Huang_AIM_2020,Tang_ASM_2019}. However, these studies concentrate on seed selection for submodular optimization with a single cardinality or knapsack constraint.

\spara{Social Advertising} Compared to influence maximization, the studies on social advertising are relatively few. Chalermsook et al. \cite{Chalermsook2015} study the revenue maximization problem for a social network platform with multiple advertisers, where each advertiser has a cardinality constraint on their seed set.
The presence of cardinality constraint on considerably simplifies the problem, while its absence in our setting poses a significant challenge. \citet{Aslay2015}~study the regret minimization problem in social advertising, where the regret is defined as the the difference between the advertisers' budgets and the expected revenue achieved by social advertising. Their setting do not consider seed user costs.
Alon et al.~\cite{Alon2012} and Abbassi et al.~\cite{Abbassi2015} investigate the channel allocation and user ordering problems in social advertising, respectively. Both works do not consider viral propagation. Moreover, none of these studies~\cite{Chalermsook2015,Aslay2015,Alon2012,Abbassi2015} consider seed user costs. Some recent work~\cite{Tang_profitMaxUS_2018,Huang_ATPM_2020} focus on profit maximization combining the benefit of influence spread with the cost of seed selection or information propagation, albeit in a single advertiser setting.
Due to the differences in problem definitions, techniques developed for these problems are inapplicable to our  RM  problem.
The work closest to ours is by ~\citet{Aslay2017}. As discussed in Section~\ref{sec:prelim-existing},  their solutions have several limitations. In this paper, we develop efficient approximation algorithms that are theoretically and empirically superior to their solutions on the solution quality and computational efficiency.

\eat{\spara{Submodular Optimization} The theory of submodular optimization has been extensively studied. Nemhauser et al.~\cite{Nemhauser1978} investigate the submodular function maximization problem under a single cardinality constraint and provide a greedy algorithm with $(1-1/e)$-approximation. Khuller et al.~\cite{Khuller1999} consider a single knapsack constraint in submodular maximization and propose a partial enumeration greedy algorithm with $(1-1/e)$-approximation. Other studies~\cite{conforti1984submodular,Grucia2011,Ashwinkumar2014,Chekuri2014,Filmus2012,Kulik2009} propose submodular optimization algorithms under more complex constraints such as matroid and multi-linear constraints. However, as indicated in~\cite{Aslay2017}, our problem is intrinsically a submodular maximization problem under a matroid constraint and multiple submodular knapsack constraints,
which are more complex than the constraints in the previous proposals: although we may use an independence system to model our constraints as in \cite{Aslay2017}, the resulting approximation could be instance dependent and arbitrarily small as that in \cite{Aslay2017}. A close work by Iyer et al.~\cite{Iyer2013} studies the submodular maximization problem under a single submodular knapsack constraint. However, as we have mentioned in Section~\ref{sec:algsa}, if we apply the algorithms in~\cite{Iyer2013} to a special case of our problem with a single advertiser (i.e., $h=1$), we can only get an instance dependent approximation guarantee that could be arbitrarily small. In contrast, our algorithms achieve a much better constant approximation ratio of $1/3$ under the case of $h=1$ by exploiting the special problem structure of social advertising.   }

\spara{Submodular Optimization} The theory of submodular optimization has been extensively studied. Nemhauser et al.~\cite{Nemhauser1978} and Khuller et al.~\cite{Khuller1999} investigate the submodular function maximization problem under a single cardinality or knapsack constraint. Other studies~\cite{conforti1984submodular,Grucia2011,Ashwinkumar2014,Chekuri2014,Filmus2012,Kulik2009} propose submodular optimization algorithms under more complex constraints such as matroid and multi-linear constraints. However, as indicated in~\cite{Aslay2017}, our problem is intrinsically a submodular maximization problem under a matroid constraint and multiple submodular knapsack constraints,
which are more complex than the constraints in the previous proposals: although we may use an independence system to model our constraints as in \cite{Aslay2017}, the resulting approximation could be instance dependent and arbitrarily small as that in \cite{Aslay2017,Iyer2013}. In contrast, our algorithms achieve a much better approximation ratio which is independent of the network instance, by exploiting the special problem structure of social advertising.


\eat{\note[Laks]{Should we mention that our result for $h=1$ improves on their result in their setting?}
\jing{Added. Note that our objective and constraint share a common submodular function $\pi(\cdot)$. However, this result may be inapplicable if the submodular functions in objective and constraint are distinct (which is their setting).}

\note[Laks]{Sorry to belabor this point, but RW section is quite weak. Need to beef up each of the 3 topics above by citing additional relevant works and positioning our work w.r.t. all of them better.} }

%
%
%
%
%
%
%
%
%
%
%
%
%
%

%% file: 7-conclusion.tex
\vspace{-1ex}
\section{Conclusion} \label{sec:concl}

We have studied the revenue maximization problem in social advertising where multiple advertisers pay the social network platform to disseminate their ads. \redc{Previous work on this problem presents algorithms with weak approximation ratios, and they incur large computational overheads in practice.}  
We provide algorithms with significantly better approximation ratios, which can be efficiently implemented. We also conduct extensive experiments using four public datasets to compare our algorithms with the existing ones, and the experimental results demonstrate the superiority of our algorithms both on the running time and on the revenue gained by the social network platform, which our algorithm achieves at considerably less seed incentive cost compared to previous algorithms.
For future work, we aim to take into account the natural competitive and complementary relationships between different propagating entities from multiple advertisers, e.g., competition between iPhone 11 and Samsung Galaxy S20, and  complementarity between iPhone 11 and Apple Watch.





%% file: 8-appendix.tex


\begin{table}[t]
	\centering
	\renewcommand{\arraystretch}{1.2}
	\small
    \caption{Frequently used notations}
    \label{tbl:prelim-notations}
    \vspace{-1mm}
	\begin{tabular}{|p{1.5cm}|p{6.1cm}|}\hline
        \textbf{Notation} & \textbf{Description} \\ \hline
        $G=(V,E)$ & A social network with node set $V$ and edge set $E$.\\ \hline
        $n,m$ & The numbers of nodes and edges in $G$, respectively.\\ \hline
        $h$ & The number of advertisers.\\ \hline
        $B_i$ & The budget of advertiser $i$.\\ \hline
        $\sigma_i(\cdot)$ & The influence spread function for advertiser $i$.  \\ \hline
        $\mathit{cpe}(i)$ & The cost-per-engagement amount for advertiser $i$.\\ \hline
        $\pi_i(\cdot)$ & The revenue function for advertiser $i$, $\pi_i(A)=\mathit{cpe}(i)\cdot \sigma_i(A)$ for any $A\subseteq V$. \\ \hline
        $c_i(\cdot)$ & A cost function, $c_i(A)=\sum_{u\in A} c_i(u)$ denotes the total cost for selecting $A$ as seed nodes for advertiser $i$.\\ \hline
        $\vec{S}$ & The collection of sets $(S_1,\dotsc, S_h)$, also denoting the set $\{(u,i)\colon u\in S_i\wedge i\in [h]\}$. \\ \hline
        $\pi(\cdot)$ & A set funtion, $\pi(M)=\sum_{i\in [h]}\pi_i(M_i)$ where $M_i=\{u\colon \exists(u,i)\in M\}$ for any $M\subseteq V\times [h]$.\\ \hline
        $\vec{O}$ & An optimal solution, $\vec{O}=(O_1,\dotsc, O_h)$. \\ \hline
        $\mathrm{OPT}$ & The revenue of an optimal solution, $\mathrm{OPT}=\pi(\vec{O})$. \\ \hline
        $f(X \mid Y)$ & The marginal gain of $X$ with respect to $Y$ for any set function $f(\cdot)$, i.e., $f(X\mid Y)=f(X\cup Y)-f(Y)$. \\ \hline
    \end{tabular}
\end{table}

%

\section{Proof of Theorem~\ref{thm:boundofthresholdgreedy}} \label{sec:proofoftheorem32}

\hkl{


In this section, we provide the proof for  Theorem~\ref{thm:boundofthresholdgreedy}. \eat{Due to the space constraint, we leave the other proofs, additional experimental results and more discussions on the results of~\cite{Aslay2017} in~\cite{techreport}.}


As the proof is a bit long, we decompose Theorem~\ref{thm:boundofthresholdgreedy} into three equivalent theorems: Theorem~\ref{thm:casebgeq2} for the case of $b\geq 2$, Theorem~\ref{thm:casebeq1} for the case of $b=1$ and Theorem~\ref{thm:casebeq0} for the case of $b=0$. We also introduce some Definitions (Definitions~\ref{def:mapping}--\ref{def:splitoptsolution}) and Lemmas (Lemma~\ref{lma:diemptyset} to Lemma~\ref{lma:ostarlesssd}) to prove Theorems~\ref{thm:casebgeq2}--\ref{thm:casebeq0}. Throughout the proof, the symbols $\vec{S}$, $\vec{D}$, $I$, $b$ denote the corresponding variables in $\mathsf{ThresholdGreedy}(\gamma)$ when it terminates.

\begin{definition}
We define a mapping $\psi$ from $\vec{S}\cup \vec{D}$ to the optimal solution $\vec{O}$ as follows. Whenever a node $u$ is added into $S_i$ or $D_i$, we set
$\psi(u,i)=(u,j)$ if there exists $(u,j)\in \vec{O}$ for certain $j\in [h]$, otherwise, we set $\psi(u,i)=\mathrm{NULL}$.
\label{def:mapping}
\end{definition}

\begin{definition}
Define $\vec{O}'=\vec{O}-\vec{S}$ and
\begin{equation*}
    Z=\Big\{(u,j)\colon (u,j)\in \vec{O}'\wedge \exists (u,i)\in \vec{S}, \psi(u,i)=(u,j)\Big\}.
\end{equation*}
Define a partition $\{\vec{O}^-,\vec{O}^*,\vec{O}^+\}$ of $\vec{O}'$ as:
\tj{\begin{align*}
&\vec{O}^-=\Big\{(v,j)\colon (v,j)\in \vec{O}'\wedge \zeta_j( v \mid S_j)< {\gamma}/{B_j} \Big\},\\
&\vec{O}^*=\Big\{(v,j)\colon (v,j)\in  Z \wedge \zeta_j( v \mid S_j)\geq {\gamma}/{B_j} \Big\},\\
&\vec{O}^+=\Big\{(v,j)\colon (v,j)\in \vec{O}'-Z\wedge \zeta_j( v \mid S_j)\geq {\gamma}/{B_j} \Big\}.
\end{align*}
}
For any given $i\in [h]$, we define $O_i'=\{u\colon (u,i)\in \vec{O}'\}$, and define $O^*_i,O^+_i,O^-_i$ similarly.
\label{def:splitoptsolution}
\end{definition}

Intuitively, ${O}_i^-$ is the set of nodes in ${O}_i$ whose \tj{marginal rates} with respect to $S_i$ are smaller than $\frac{\gamma}{B_i}$; ${O}_i^*$ is the set of nodes in ${O}_i$ whose \tj{marginal rates} are no less than $\frac{\gamma}{B_i}$, but all nodes in ${O}_i^*$ have already been selected into $\bigcup_{j\in [h]}{S_j}$ by our algorithm; ${O}_i^+$ is the set of elements in ${O}_i$ whose \tj{marginal rate} are no less than $\frac{\gamma}{B_i}$, and none of the nodes in ${O}_i^+$ has been selected into $\bigcup_{j\in [h]}{S_j}$. Therefore, Definition~\ref{def:mapping} and Definition~\ref{def:splitoptsolution} actually provide a way to partition the set of elements in the optimal solution. In the sequel, Lemma~\ref{lma:diemptyset} to Lemma~\ref{lma:ostarlesssd} find some quantitative relationships between the revenue of the elements in these partitions and the revenue of $\vec{S}$:

\begin{lemma}
$\sum\nolimits_{i\colon D_i=\emptyset} \pi_i(O_i^+\mid S_i)\leq \sum\nolimits_{j\in [h]} \pi_j(D_j \mid S_j)$.
\label{lma:diemptyset}
\end{lemma}
\begin{proof}

We first prove that: given any $i\in [h]$ satisfying $D_i=\emptyset$, if there exists certain node $u\in O_i^+$ , then we must have $u\in \bigcup_{j\in [h]}D_j$. This is due to the reason that, if $u\notin \bigcup_{j\in [h]}D_j$, then we must have $u\notin \bigcup_{j\in [h]}S_j\cup D_j$ according to the definition of $O_i^+$, and hence $u$ can be added into $S_i$ or $D_i$ as $u$'s \tj{marginal rate} with respect to $S_i$ is larger than $\frac{\gamma}{B_i}$, but this contradicts the fact that $O_i^+\cap S_i=\emptyset$ and $D_i=\emptyset$.

Now suppose that the node $u$ described above is in $D_j$ for certain $j\in [h]$. According to the greedy selection rule of our algorithm, we must have $\pi_i(u\mid S_i)\leq \pi_j(u\mid S_j)=\pi_j(D_j\mid S_j)$. Therefore, by considering all the nodes in $\bigcup_{i\colon D_i=\emptyset}O_i^+$ in the same way as described above, we can prove the lemma.
\end{proof}

\begin{lemma}
For any $i\in [h]$, we have $\pi_i(O^-_i\mid S_i) \leq \gamma$.
\label{lma:oiminuslessgamma}
\end{lemma}
\begin{proof}
We sort the nodes in ${O}^-_i$ into $\{v_1,\dotsc,v_q\}$ such that $r(v_1)\geq r(v_2)\geq \dotsb \geq r(v_q)$, where $r(v_j)$ is defined as:
\begin{equation*}
\tj{r(v_j) = \zeta_i(v_j \mid S_i\cup \{v_1,\dotsc, v_{j-1}\})},~\text{ for every }j\in [q]
\end{equation*}
Therefore, we must have \tj{$r(v_1)=\zeta_i(v_1 \mid S_i)\leq {\gamma}/{B_i}$} as $v_1\in O^-_i$. Moreover, using similar reasoning with that in Theorem~\ref{thm:greedyratioforsa}, we get
\begin{equation*}
{\gamma}/{B_i}\geq r(v_1) \geq \big({\pi_i(S_i\cup O_i^-)- \pi_i(S_i)}\big)/{B_i},
\end{equation*}
so the lemma follows.
\end{proof}

\begin{lemma}
We have $\pi(\vec{O}^*\mid \vec{S})\leq \pi(\vec{S})$.
\label{lma:ostarlesssd}
\end{lemma}
\begin{proof}
For any $(u,i)\in \vec{S}$, we use $\mathrm{Pre}(u,i)$ to denote the set of all elements added into $\vec{S}$ before $(u,i)$ by $\mathsf{ThresholdGreedy}(\gamma)$. Clearly, we have $\sum_{(u,i)\in \vec{S}}\pi\big((u,i)\mid \mathrm{Pre}(u,i)\big)=\pi(\vec{S})$ and hence
\begin{align}
\pi(\vec{S})
&\geq \sum_{(u,i)\in \vec{S}\wedge \psi(u,i)\neq \mathrm{NULL}} \pi\big((u,i)\mid \mathrm{Pre}(u,i)\big) \nonumber\\
&\geq \sum_{(u,i)\in \vec{S}\wedge \psi(u,i)\neq \mathrm{NULL}} \pi\big(\psi(u,i)\mid \mathrm{Pre}(u,i)\big) \label{eqn:duetosub12}\\
&\geq \sum_{(u,i)\in \vec{S}\wedge \psi(u,i)\neq \mathrm{NULL}} \pi\big(\psi(u,i)\mid \vec{S}\big) \label{eqn:duetogreedy1}\\
&\geq \sum\nolimits_{(u,j)\in \vec{O}^*} \pi\big((u,j) \mid \vec{S}\big) \geq \pi(\vec{S}\cup \vec{O}^*)- \pi(\vec{S}), \nonumber
\end{align}
where \eqref{eqn:duetosub12} and \eqref{eqn:duetogreedy1} are due to the greedy rule of Algorithm~\ref{alg:thresholdgreedy} in Line~\ref{ln:greedyruleinthreshold} and the submodularity of $\pi(\cdot)$. So the lemma follows.
\end{proof}

Based on the quantitative relationships found in Lemma~\ref{lma:diemptyset} to Lemma~\ref{lma:ostarlesssd}, we then prove Theorems~\ref{thm:casebgeq2}--\ref{thm:casebeq0} in the following, which completes the proof of Theorem~\ref{thm:boundofthresholdgreedy}.

\begin{theorem}
We have $\pi(\vec{S}^*)\geq b\gamma/2$ when $b\geq 2$.\label{thm:casebgeq2}
\end{theorem}
\begin{proof}
Note that $|I|=b$. Consider any $i\in I$ and suppose that the nodes sequentially added into $S_i$ are $u_1,\dotsc, u_{k-1}$ and $D_i=\{u_{k}\}$. 
Then we have $c_i(\{u_1,\dotsc, u_k\})+\pi_i(\{u_1,\dotsc, u_k\})\geq B_i$ and
\begin{equation*}
\tj{\zeta_i(u_{j}\mid \{u_1,\dotsc, u_{j-1}\})\geq {\gamma}/{B_i} \text{ for every }j\in [k].}
\end{equation*}
Therefore, we have
\begin{align*}
\pi_i({S}_i\cup {D}_i)
&= \sum\nolimits_{j=1}^k \pi_i(u_{j}\mid \{u_1,\dotsc, u_{j-1}\})\\
&\geq\sum\nolimits_{j=1}^k\frac{\gamma}{B_i}\big(c_i(u_j)+\pi_i(u_{j}\mid \{u_1,\dotsc, u_{j-1}\})\big)\\
&=\frac{\gamma}{B_i}\big(c_i(\{u_1,\dotsc, u_k\})+\pi_i(\{u_1,\dotsc, u_k\})\big)\geq \gamma,
\end{align*}
and hence
\begin{equation*}
\pi(\vec{S}^*)\geq \sum_{i\in I}\max\{\pi_i(S_i),\pi_i(D_i)\}\geq \sum_{i\in I}\frac{\pi_i(S_i\cup D_i)}{2}\geq \frac{b\gamma}{2}.
\end{equation*}
So the theorem follows.
\end{proof}

\begin{theorem}
We have $\pi(\vec{S}^*)\geq \max\{\frac{1}{6}(\mathrm{OPT}-h\gamma),\frac{\gamma}{2}\}$ when $b=1$.\label{thm:casebeq1}
\end{theorem}
\begin{proof}
We can prove $\pi(\vec{S}^*)\geq \gamma/2$ by similar reasoning with that in Theorem~\ref{thm:casebgeq2}. Next, we prove $\pi(\vec{S}^*)\geq \frac{1}{6}(\mathrm{OPT}-h\gamma)$.  According to Lemma~\ref{lma:ostarlesssd}, Lemma~\ref{lma:diemptyset} and Lemma~\ref{lma:oiminuslessgamma}, we can get $\pi(\vec{O}^*\mid \vec{S})\leq \pi(\vec{S}),~\pi(\vec{O}^-\mid \vec{S})\leq h\gamma$
and
\begin{equation*}
\sum\nolimits_{j\neq i} \pi_j(O_j^+\mid S_j)\leq \sum\nolimits_{\ell\in [h]} \pi_\ell(D_\ell \mid S_\ell)=\pi_i(D_i\mid S_i),
\end{equation*}
where $i$ is the advertiser in $I$. Besides, we have $\pi_i(O_i^+)\leq 3\pi_i(S_i^*)$ due to Theorem~\ref{thm:greedyratioforsa} and and Line~\ref{ln:callnaivegreedy} of $\mathsf{ThresholdGreedy}(\gamma)$. Using the above results, we get
\begin{align*}
\pi(\vec{O})-\pi(\vec{S})
&\leq \pi(\vec{O}\mid \vec{S})\\
&\leq \pi(\vec{O}^* \mid \vec{S}) + \pi(\vec{O}^- \mid \vec{S})+ \pi(\vec{O}^+\mid \vec{S}) \\
&\leq \pi(\vec{S})+ h\gamma + \sum\nolimits_{j\neq i} \pi_j(O_j^+\mid S_j) + \pi_i(O_i^+\mid S_i)\\
&\leq \pi(\vec{S})+ h\gamma + \pi_i(D_i\mid S_i) + 3\pi_i(S_i^*)\\
&\leq \pi(\vec{S})+ h\gamma + 4\pi(\vec{S}^*).
\end{align*}
So the theorem follows by combining the above inequality with $\pi(\vec{S})\leq \pi(\vec{S}^*)$.
\end{proof}

\begin{theorem}
We have ${\pi}(\vec{S}^*)\geq \frac{1}{2}(\mathrm{OPT}-h\gamma)$ when $b=0$.\label{thm:casebeq0}
\end{theorem}
\begin{proof}
When $b=0$, we must have $\vec{D}=\emptyset$. So we can use Lemma~\ref{lma:diemptyset} to Lemma~\ref{lma:ostarlesssd} to get $\pi(\vec{O}^*\mid \vec{S})\leq \pi(\vec{S})$ and
\begin{align*}
&\pi(\vec{O}^+\mid \vec{S})=\sum\nolimits_{i\colon D_i=\emptyset} \pi_i(O_i^+\mid S_i)=0,\\
&\pi(\vec{O}^- \mid \vec{S})=\sum\nolimits_{i=1}^h \pi_i(O^-_i\mid S_i)\leq h\gamma.
\end{align*}
Combining the above equations, we get
\begin{align}
\pi(\vec{O})-\pi(\vec{S}^*) &\leq\pi(\vec{O})-\pi(\vec{S})\leq \pi(\vec{O}\mid \vec{S}) \nonumber\\
&\leq \pi(\vec{O}^- \mid \vec{S}) + \pi(\vec{O}^*\mid \vec{S}) + \pi(\vec{O}^+\mid \vec{S}) \nonumber\\
&\leq h\gamma +\pi(\vec{S})\leq h\gamma +\pi(\vec{S}^*).\label{eqn:toberearanging}
\end{align}
So the theorem then follows by re-arranging \eqref{eqn:toberearanging}.
\end{proof}

}



\section{Other Missing Proofs} \label{sec:othermissingproofs}

\subsection{Proof of Theorem~\ref{thm:greedyratioforsa}}
\begin{proof}
Suppose that the nodes sequentially added into $S_i$ are $u_1,\dotsc, u_{k-1}$ and $D_i=\{u_{k}\}$. We sort the nodes in ${O}_i\setminus S_i$ into $\{v_1,\dotsc,v_q\}$ such that $r(v_1)\geq r(v_2)\geq \dotsb \geq r(v_q)$, where $r(v_j)$ for every $j\in [q]$ is defined as:
\begin{equation*}
r(v_j) = \frac{\pi_i(v_j \mid S_i\cup \{v_1,\dotsc, v_{j-1}\})}{c_i(v_j)+\pi_i(v_j \mid S_i\cup \{v_1,\dotsc, v_{j-1}\})}.
\end{equation*}
Using the submodularity of $\pi_i(\cdot)$ \laks{and elementary algebra}, we get
\begin{align}
r(v_1)
&\geq\frac{ \sum_{j=1}^q r(v_j)\cdot \big(c_i(v_j)+\pi_i(v_j \mid S_i\cup \{v_1,\dotsc,v_{j-1}\})\big)}{\sum_{j=1}^q \big(c_i(v_j)+\pi_i(v_j \mid S_i\cup \{v_1,\dotsc,v_{j-1}\})\big)}\nonumber\\
&\geq \frac{ \sum_{j=1}^q \pi_i(v_j \mid S_i\cup  \{v_1,\dotsc, v_{j-1}\})}{\sum_{j=1}^q \big(c_i(v_j)+\pi_i(v_j \mid \{v_1,\dotsc, v_{j-1}\})\big)} \nonumber\\
&\geq \frac{\pi_i(S_i\cup O_i)- \pi_i(S_i)}{c_i(O_i)+\pi_i(O_i)}\geq \frac{\pi_i(O_i)- \pi_i(S_i)}{B_i}. \label{eqn:tocombineeqn1}
\end{align}
\revise{
Notice that by the greedy selection rule (Line 4 of Algorithm~\ref{alg:naivegreedy}), we have $u_k=\arg\max_{v\in U \setminus S_i} \zeta_i(v\mid S_i)$. And by definition, we have $v_1=\arg\max_{v\in O_i \setminus S_i} \zeta_i(v \mid S_i)$. Given that $(O_i \setminus S_i) \subseteq (U \setminus S_i)$, we have $r(v_1) = \zeta_i(v_1 \mid S_i) \le  \zeta_i(u_k \mid S_i)$.
Thus, we have}
%
\begin{equation}
r(v_1)\leq \frac{\pi_i(u_k \mid S_i)}{c_i(u_k)+\pi_i(u_k \mid S_i)}\leq \frac{\pi_i(S_i\cup D_i)}{c_i(S_i\cup D_i)+\pi_i(S_i\cup D_i)}.\label{eqn:tocombineeqn2}
\end{equation}
Note that $c_i(S_i\cup D_i)+\pi_i(S_i\cup D_i)\geq B_i$. \eat{so that $B_i\cdot r(v_1)\leq \pi(S_i\cup D_i)$.} Thus, we can combine Eqn.~\eqref{eqn:tocombineeqn1} and Eqn.~\eqref{eqn:tocombineeqn2} to get
\begin{equation*}
\pi_i(O_i)\leq \pi_i(S_i)+ \pi(S_i\cup D_i)\leq 2\pi_i(S_i)+ \pi_i(D_i)\leq 3\pi_i(S_i^*),
\end{equation*}
which completes the proof.
\end{proof}

\subsection{Proof of Theorem~\ref{thm:proofofsearchtau2}}

\begin{proof}

When the $\mathsf{Search}(\tau,2)$ algorithm stops, one of the three cases holds: 

Case 1: $b_1<2$. This case implies that $\mathsf{Search}(\tau,2)$ returns $b_2<2$ and $\vec{T}_2^*=\mathsf{ThresholdGreedy}(0)$. Using Theorem~\ref{thm:boundofthresholdgreedy}, we have $\pi(\vec{S}^*)\geq \pi(\vec{T}^*_2)\geq {\mathrm{OPT}}/6$.

Case 2: $b_1\geq 2$ and $\vec{T}_2^*= \emptyset$. This implies that $\mathsf{Search}(\tau,2)$ returns $\vec{T}_1^*=\mathsf{ThresholdGreedy}(\gamma_1)$ and $(1+\tau)\gamma_1\geq \gamma_2\geq\gamma_{max}$. \ca{Note that we have $\frac{\mathrm{OPT}}{h}\leq \frac{\sum_{i\in [h]}\pi_i(O_i)}{\sum_{i\in h}\big(c_i(O_i)+\pi_i(O_i)\big)/B_i}.$ Besides, by the submodularity of $\pi_i(\cdot)$ we get
\begin{align*}
\frac{\pi_i(O_i)}{\big(c_i(O_i)+\pi_i(O_i)\big)/B_i}
&\leq \frac{B_i\cdot\sum_{v\in O_i} \pi_i(v)}{\sum_{v\in O_i} \big(c_i(v)+ \pi_i(v)\big)}\\
&\leq \max \left\{B_j\cdot \zeta_j({v}\mid \emptyset)\colon v\in V, j\in [h]\right\}\\
&= {\gamma_{max}}.
\end{align*}
Thus, we have $\mathrm{OPT}/h\leq \gamma_{max}$.} Using these results and Theorem~\ref{thm:boundofthresholdgreedy}, we can get
\eat{Moreover, we can prove $\mathrm{OPT}\leq h\cdot \gamma_{max}$~(the proof can be found in~\cite{techreport}).}
\begin{equation}\label{eqn:gammamaxisgood}
\pi(\vec{S}^*)\geq\pi(\vec{T}^*_1)\geq \frac{b_1{\gamma_1}}{2}\geq \frac{\gamma_{max}}{(1+\tau)}\geq  \frac{\mathrm{OPT}}{h(1+\tau)}.
\end{equation}

Case 3: $b_1\geq 2$ and $\vec{T}_2^*\neq \emptyset$. This case implies that $\mathsf{Search}(\tau,2)$ returns $\vec{T}_1^*=\mathsf{ThresholdGreedy}(\gamma_1)$, $\vec{T}_2^*=\mathsf{ThresholdGreedy}(\gamma_2)$ and $b_2<2$. Therefore, if $\gamma_1\geq \frac{\mathrm{OPT}}{(h+6)(1+\tau)}$, then we can use Theorem~\ref{thm:boundofthresholdgreedy} to get
\begin{equation}
\pi(\vec{S}^*)\geq \pi(\vec{T}^*_1)\geq \frac{b_1{\gamma_1}}{2}\geq \frac{\mathrm{OPT}}{(h+6)(1+\tau)};
\end{equation}
otherwise we must have $\gamma_2\leq {\mathrm{OPT}}/{(h+6)}$ according to the stopping condition in Line~\ref{ln:stoppcondsearch} and hence
\begin{equation}
\pi(\vec{T}^*_2)\geq \frac{1}{6}(\mathrm{OPT}-h\gamma_2)\geq \frac{\mathrm{OPT}}{h+6}.
\end{equation}

The theorem follows by combining the  above cases.
\end{proof}

\subsection{Proof of Theorem~\ref{thm:upperboundofrrsets}}

We first introduce the following concentration bounds:

\begin{lemma}
Given any solution $\vec{S}$ to the revenue maximization problem and any set $\mathcal{R}$ of RR-sets, we have
\begin{align*}
&\Pr \left[ \widetilde{\pi}(\vec{S},\mathcal{R})-{\pi}(\vec{S})\geq \frac{n\Gamma t}{|\mathcal{R}|} \right]\leq \exp\left(\frac{-t^2}{2{\pi}(\vec{S})\cdot \frac{|\mathcal{R}|}{n\Gamma}+\frac{2}{3}t}\right),\\
&\Pr \left[ \widetilde{\pi}(\vec{S},\mathcal{R})-{\pi}(\vec{S})\leq \frac{-n\Gamma t}{|\mathcal{R}|} \right]\leq \exp\left(\frac{-t^2}{2{\pi}(\vec{S})\cdot \frac{|\mathcal{R}|}{n\Gamma}}\right).
\end{align*}
\label{lma:concentrationbound}
\end{lemma}
\begin{proof}
The lemma is an extension of Corollary 1 and Corollary 2 in \cite{TangSX2015}. So we omit the proof.
\end{proof}

Next, we introduce two parameters $\epsilon_1$ and $\epsilon_2$ defined as:
\begin{align}
&\epsilon_1=\frac{\epsilon\cdot \sqrt{\ln\frac{4}{\delta}} }{ \lambda\cdot \sqrt{\ln\frac{4}{\delta}} + \sqrt{\lambda\left(\ln\frac{4}{\delta}+\sum_{i\in [h]}\mu_i\ln\frac{\mathrm{e}n}{\mu_i}\right) } },\\
&\epsilon_2=\epsilon-\lambda\cdot \epsilon_1.
\end{align}
These parameters are similar in spirit to those in~\cite{TangSX2015}. With these parameters, we further propose the following lemma, which reveals several key conditions for the one-batch algorithm to return a valid bi-criteria approximate solution:
\begin{lemma}
Define
\begin{align*}
&\mathcal{Q}_1^i=\{S_i\subseteq V\colon c_i(S_i)\leq (1+\varrho)B_i <c_i(S_i)+\pi_i(S_i) \},\\
&\mathcal{Q}_2^i=\{S_i\subseteq V\colon c_i(S_i)+\pi_i(S_i)\leq (1+\varrho)B_i \},\\
&\mathcal{Q}_3=\{\vec{S}\colon \big(S_i\in \mathcal{Q}_2^i\big)\land \big(\pi(\vec{S})<(\lambda-\epsilon)\mathrm{OPT}\big)\}.
\end{align*}
and define four events $\mathcal{E}_1,\mathcal{E}_2,\mathcal{E}_3,\mathcal{E}_4$ as follows: 
\begin{enumerate}
\item $\mathcal{E}_1$ denotes the event that every node set $S_i\in \mathcal{Q}_1^i$ satisfies $c_i(S_i)+\widetilde{\pi}_i(S_i,\mathcal{R})> (1+\varrho/2)B_i$ for all $i\in [h]$.
\item $\mathcal{E}_2$ denotes the event that the optimal solution $\vec{O}$ satisfies $\forall i\in [h]\colon \widetilde{\pi}_i(O_i,\mathcal{R})\leq \pi_i(O_i)+\varrho B_i/2$.
\item $\mathcal{E}_3$ denotes the event that $\widetilde{\pi}(\vec{O},\mathcal{R})\geq (1-\epsilon_1)\cdot \mathrm{OPT}$
\item $\mathcal{E}_4$ denotes the event that $\pi(\vec{S})\geq \widetilde{\pi}(\vec{S},\mathcal{R}) - \epsilon_2 \cdot \mathrm{OPT}$ for every $\vec{S}\in \mathcal{Q}_3$
\end{enumerate}
If $\mathcal{E}_1,\mathcal{E}_2,\mathcal{E}_3$ and $\mathcal{E}_4$ all happen, then the one-batch algorithm must return a valid bi-criteria approximate solution $\vec{S}^*$ satisfying $c_i(S_i^*)+\pi_i(S_i^*)\leq (1+\varrho)B_i$ for all $i\in [h]$ and $\pi(\vec{S}^*)\geq (\lambda-\epsilon)\mathrm{OPT}$.
\end{lemma}
\begin{proof}
Clearly, the one-batch algorithm must return a $(1+\varrho)\vec{B}$ budget-feasible solution when $\mathcal{E}_1$ holds (i.e., $\vec{S}^*_i\in \mathcal{Q}_2^i$ for all $i\in [h]$). When $\mathcal{E}_2$ holds, we have
\begin{equation*}
	\widetilde{\pi}(\vec{S}^*,\mathcal{R})+c_i(O_i)\leq  \pi_i(O_i)+c_i(O_i)+\varrho B_i/2\leq (1+\varrho/2)B_i.
\end{equation*}
Thus, Algorithm~\ref{alg:rmwithoracle} returns $\vec{S}^*$ satisfying
\begin{equation*}
\widetilde{\pi}(\vec{S}^*,\mathcal{R})\geq \lambda\cdot \widetilde{\pi}(\vec{O},\mathcal{R}).
\end{equation*}
Now suppose by contradiction that the one-batch algorithm returns $\vec{S}^*\in \mathcal{Q}_3$. When the events $\mathcal{E}_1,\mathcal{E}_2,\mathcal{E}_3,\mathcal{E}_4$ all hold, we have
\begin{align*}
\pi(\vec{S}^*)
&\geq \widetilde{\pi}(\vec{S}^*,\mathcal{R}) - \epsilon_2 \cdot \mathrm{OPT} \nonumber\\
&\geq \lambda\cdot \widetilde{\pi}(\vec{O},\mathcal{R}) - \epsilon_2 \cdot \mathrm{OPT} \label{eqn:duetoapproximationratio}\\
&\geq \lambda\cdot (1-\epsilon_1)\cdot \mathrm{OPT} - \epsilon_2 \cdot \mathrm{OPT} \nonumber\\
&= (\lambda-\epsilon) \mathrm{OPT}, \nonumber
\end{align*}
which contradicts $\vec{S}^*\in \mathcal{Q}_3$. Therefore, we have $\vec{S}^*\notin \mathcal{Q}_3$, which immediately concludes the lemma that $\pi(\vec{S}^*)\geq (\lambda-\epsilon)\mathrm{OPT}$.
\end{proof}

With the above lemma, Theorem~\ref{thm:upperboundofrrsets} follows as long as the probability that at least one event in $\{\mathcal{E}_1,\mathcal{E}_2,\mathcal{E}_3,\mathcal{E}_4\}$ does not happen is no more than $\delta$ when $|\mathcal{R}|\geq \theta_{max}$. Indeed, the following lemmas (Lemmas~\ref{lma:boundE1}--\ref{lma:boundE4}) prove that the failure probability of each event in $\{\mathcal{E}_1,\mathcal{E}_2,\mathcal{E}_3,\mathcal{E}_4\}$ is no more than $\delta/4$ when $|\mathcal{R}|\geq\theta_{max}$. These lemmas together with the union bound complete the proof of Theorem~\ref{thm:upperboundofrrsets}.

\begin{lemma}
When $|\mathcal{R}|\geq \bar{\theta}_{max}$, we have $\Pr[\neg \mathcal{E}_1]\leq \delta/4$.
\label{lma:boundE1}
\end{lemma}
\begin{proof}
For each $S_i\in \mathcal{Q}_1^i$ with $c_i(S_i)+\pi_i(S_i)>(1+\varrho)B_i$, we have
\begin{align*}
&\Pr [c_i(S_i)+\widetilde{\pi}_i(S_i,\mathcal{R})\leq  (1+\varrho/2)B_i] \\
&\leq\Pr \left[c_i(S_i)+\widetilde{\pi}_i(S_i,\mathcal{R})\leq  \frac{1+\varrho/2}{1+\rho}(c_i(S_i)+{\pi}_i(S_i)) \right] \\
&\leq \Pr \left[\widetilde{\pi}_i(S_i,\mathcal{R})-{\pi}_i(S_i)\leq -\frac{\varrho}{2(1+\varrho)}(c_i(S_i)+\pi_i(S_i)) \right] \\
&\leq \exp\left(\frac{-\varrho^2 |\mathcal{R}| (c_i(S_i)+\pi_i(S_i))^2}{8n\Gamma (1+\varrho)^2 \pi_i(S_i)}\right) \\
&\leq \exp\left(\frac{-\varrho^2 |\mathcal{R}|B_i}{8n\Gamma (1+\varrho)}\right)
\leq \frac{\delta}{4h}\cdot\left(\frac{{\mu}}{\mathrm{e}n}\right)^{{\mu}}.
\end{align*}
On the other hand,
\begin{align*}
	|Q_1^i|
	&\leq\sum_{j=0}^{\mu_i} \binom{n}{j}
	\leq \sum_{j=0}^{\mu_i} \frac{n^j}{j!}
	\leq\sum_{j=0}^{\mu} \left( \frac{{\mu}^j}{j!}\frac{n^j}{{\mu}^j}\right)\\
	&\leq\frac{n^\mu}{{\mu}^\mu}\sum_{j=0}^{\mu}\frac{{\mu}^j}{j!}
	\leq \left(\frac{\mathrm{e}n}{\mu}\right)^{\mu}.
\end{align*}
This implies that $\sum_{i\in[h]}|Q_1^i|\leq h\cdot\left(\frac{\mathrm{e}n}{{\mu}}\right)^{{\mu}}$. Thus,
\begin{align*}
	\Pr[\neg \mathcal{E}_1]
	&\leq \sum_{i\in[h]}\sum_{S_i\in \mathcal{Q}_1^i} \Pr [c_i(S_i)+\widetilde{\pi}_i(S_i,\mathcal{R})\leq  (1+\varrho/2)B_i]\\
	&\leq h\cdot \left(\frac{\mathrm{e}n}{{\mu}}\right)^{{\mu}}\cdot \frac{\delta}{4h}\cdot \left(\frac{{\mu}}{\mathrm{e}n}\right)^{{\mu}}=\frac{\delta}{4}.
\end{align*}
This completes the proof.
\end{proof}

\begin{lemma}
When $|\mathcal{R}|\geq \bar{\theta}_{max}$, we have $\Pr[\neg \mathcal{E}_2]\leq \delta/4$.
\label{lma:boundoptinsamples}
\end{lemma}
\begin{proof}
For any $i\in [h]$, we have
\begin{equation*}
\Pr [\widetilde{\pi}_i(O_i,\mathcal{R})>{\pi}_i(O_i)+ \varrho B_i/2]
\leq \exp\left(\frac{-\varrho^2 |\mathcal{R}| B_i}{4n\Gamma (2+\varrho/3)}\right)\leq \frac{\delta}{4h}.
\end{equation*}
So the lemma follows by using the union bound.
\end{proof}

\begin{lemma}
When $|\mathcal{R}|\geq \hat{\theta}_{max}$, we have $\Pr [\neg \mathcal{E}_3]\leq \frac{\delta}{4}$.
\label{lma:E3}
\end{lemma}
\begin{proof}
Using the concentration bounds proposed in Lemma~\ref{lma:concentrationbound}, we can get
\begin{align*}
&\Pr [\widetilde{\pi}(\vec{O},\mathcal{R})< (1-\epsilon_1)\cdot \mathrm{OPT}]\\
&\leq \exp\left( \frac{-\epsilon_1^2\mathrm{OPT}^2|\mathcal{R}|^2/(n\Gamma)^2}{2\mathrm{OPT}|\mathcal{R}|/(n\Gamma)} \right)\\
&\leq \exp\left( \frac{-\epsilon_1^2|\mathcal{R}|}{2n} \right)\leq \frac{\delta}{4}.
\end{align*}
\footnote{We assume that $\mathrm{OPT} \geq \Gamma$, i.e., there exists a feasible solution $\vec{S}$ such that $S_i\neq\emptyset$ for all $i$.}So the lemma follows.
\end{proof}

\begin{lemma}
When $|\mathcal{R}|\geq \hat{\theta}_{max}$, we have $\Pr [\neg \mathcal{E}_4]\leq \frac{\delta}{4}$.
\label{lma:boundE4}
\end{lemma}
\begin{proof}
We have
\begin{equation*}
|\mathcal{Q}_3|\leq \prod_{i\in [h]} \sum_{j=0}^{\mu_i} {n\choose j}\leq \prod_{i\in [h]} \left(\frac{\mathrm{e}n}{\mu_i}\right)^{\mu_i}.
\end{equation*}
So we can use Lemma~\ref{lma:concentrationbound} and the union bound to get
\begin{align*}
\Pr [\neg \mathcal{E}_4]
&\leq\sum\nolimits_{\vec{S}\in \mathcal{Q}_3}\Pr [ \pi(\vec{S})< \widetilde{\pi}(\vec{S}) - \epsilon_2 \cdot \mathrm{OPT} ] \\
&\leq \sum\nolimits_{\vec{S}\in \mathcal{Q}_3} \exp\left( \frac{-\epsilon_2^2\mathrm{OPT}|\mathcal{R}|/(n\Gamma)}{2(\lambda-\epsilon)+\frac{2\epsilon_2}{3}} \right)\\
&\leq \sum\nolimits_{\vec{S}\in \mathcal{Q}_3} \exp\left(\frac{-\epsilon_2^2|\mathcal{R}|}{2n\lambda} \right)\\
&\leq \sum\nolimits_{\vec{S}\in \mathcal{Q}_3} \bigg(\frac{\delta}{4}\cdot \prod_{i\in [h]} \left(\frac{\mu_i}{\mathrm{e}n}\right)^{\mu_i} \bigg)\\
&\leq \delta/4. \label{eqn:boundbadsolutionprob}
\end{align*}
So the lemma follows.
\end{proof}

\subsection{Proof of Theorem~\ref{thm:highprobratio}}

In the following, we first introduce Lemma~\ref{lma:hugeeqn} and Lemma~\ref{lma:correctubofopt}, and then use them to prove Theorem~\ref{thm:highprobratio}.

\begin{lemma}
Given any set $\mathcal{R}$ of RR-sets, a solution $\vec{S}$ to the RM problem and $t_1\geq \widetilde{\pi}(\vec{S},\mathcal{R})\geq t_2$, we have
\begin{align*}
&\!\!\!\Pr \left[\pi(\vec{S})>\left(\sqrt{\frac{t_1|\mathcal{R}|}{n\Gamma}+\frac{a}{2}}+\sqrt{\frac{a}{2}}\right)^{2}\cdot\frac{n\Gamma}{|\mathcal{R}|} \right]\leq \mathrm{e}^{-a},\\
&\!\!\!\Pr \left[\pi(\vec{S})<\Bigg(\bigg(\sqrt{\frac{t_2|\mathcal{R}|}{n\Gamma}+\frac{2a}{9}}-\sqrt{\frac{a}{2}}\bigg)^{2}-\frac{a}{18}\Bigg)\cdot\frac{n\Gamma}{|\mathcal{R}|}\right]\leq \mathrm{e}^{-a}.
\end{align*}
The above inequalities also hold when we replace $\pi(\vec{S})$ by $\pi_i({S_i})$ and require $t_1\geq \widetilde{\pi}({S_i},\mathcal{R})\geq t_2$.
\label{lma:hugeeqn}
\end{lemma}
\begin{proof}
The proof of the lemma is similar to Lemma 4.2 and Lemma 4.3 in~\cite{Tang2018}, so we omit the proof.
\end{proof}

\hkr{
\begin{lemma}
The function $\mathsf{SeekUB}$ can return a correct upper bound of $\widetilde{\pi}(\vec{O}, \mathcal{R}_1)$.
\label{lma:correctubofopt}
\end{lemma}

\begin{proof}
Let ${\mathrm{OPT}}'=\widetilde{\pi}(\vec{O},\mathcal{R}_1)$. As $\lambda$ is the approximation ratio of $\mathsf{RM\_with\_Oracle}$, we know that $\widetilde{\pi}(\vec{S}^*,\mathcal{R}_1)/\lambda$ is a trivial upper bound of ${\mathrm{OPT}}'$. Note that $\mathsf{SeekUB}$ returns this trivial upper bound immediately when $h=1$. So we only need to prove the lemma under the case of $h>1$. Consider the following cases when $h> 1$:

(1) $b_1<b_{min}$: In this case, we must have $\vec{T}_2^*=\mathsf{ThresholdGreedy}(\gamma_2)$, $\gamma_2=0$ and $b_2<b_{min}$ according to the $\mathsf{Search}$ algorithm. Therefore, we know that $z=6\widetilde{\pi}(\vec{T}_2^*,\mathcal{R}_1)$ must be an upper bound of ${\mathrm{OPT}}'$ due to Theorem~\ref{thm:boundofthresholdgreedy}.

(2) $b_1\geq b_{min}$ and $\vec{T}_2^*=\emptyset$: In this case, we must have $\vec{T}_1^*=\mathsf{ThresholdGreedy}(\gamma_1)$, $\gamma_2\geq \gamma_{max}$ and $(1+\tau)\gamma_1\geq \gamma_2$ according to the $\mathsf{Search}$ algorithm. Similar to \eqref{eqn:gammamaxisgood}, we have $\gamma_1\geq \mathrm{OPT}'/\big(h(1+\tau)\big)$. Note that $b_{min}=2$ when $h\geq 4$ and $b_{min}=1$ when $2\leq h\leq 3$. Therefore, we can use Theorem~\ref{thm:boundofthresholdgreedy} to get that when $h\geq 4$,
\begin{equation*}
\widetilde{\pi}(\vec{T}_1^*,\mathcal{R}_1)\geq \gamma_1\geq \mathrm{OPT}'/\big(h(1+\tau)\big)\geq \lambda\cdot \mathrm{OPT}',
\end{equation*}
and when $2\leq h\leq 3$,
\begin{equation*}
\widetilde{\pi}(\vec{T}_1^*,\mathcal{R}_1)\geq \gamma_1/2\geq \mathrm{OPT}'/\big(2h(1+\tau)\big)\geq \lambda\cdot \mathrm{OPT}'.
\end{equation*}
Therefore, $z=\widetilde{\pi}(\vec{T}_1^*,\mathcal{R}_1)/\lambda$ is always an upper bound of $\mathrm{OPT}'$ in this case.

(3) $b_1\geq b_{min}$ and $\vec{T}_2^*\neq \emptyset$: In this case, we must have $\vec{T}_2^*=\mathsf{ThresholdGreedy}(\gamma_2)$ and $b_2< b_{min}$ according to the $\mathsf{Search}$ algorithm. Using Theorem~\ref{thm:boundofthresholdgreedy}, it can be easily seen that the values of $z$ set in Line~\ref{ln:setupperbound1} and Line~\ref{ln:setupperbound2} are correct upper bounds of $\mathrm{OPT}'$ under the cases of $b_{2}=0$ and $b_{2}=1$, respectively.

According to the above reasoning, we know that the variable $z$ in the $\mathsf{SeekUB}$ algorithm is always set as a correct upper bound of $\widetilde{\pi}(\vec{O},\mathcal{R}_1)$. Moreover, the upper bound returned by $\mathsf{SeekUB}$ could be tighter than the trivial upper bound $\widetilde{\pi}(\vec{S}^*,\mathcal{R}_1)/\lambda$, as it returns the smaller one between $z$ and the trivial upper bound. So the lemma follows.
\end{proof}}

\hkr{

\begin{proof} [Proof of Theorem~\ref{thm:highprobratio}]
  Suppose that the algorithm terminates after $T$ iterations. Note that we have $|\R_1|\geq \theta_{max}$ in the $(t_{\max}+1)$-th iteration. Let $\mathcal{E}_{sucess}$ denote the event that the algorithm returns a valid solution $\vec{S}^*$ satisfying  $c_i(S_i^*)+\pi_i(S_i^*)\leq (1+\varrho)B_i$ for all $i\in [h]$ and $\pi(\vec{S}^*)\geq (\lambda-\epsilon)\mathrm{OPT}$. Therefore, we can use Theorem~\ref{thm:upperboundofrrsets} to get
	\begin{equation}\label{eqn:last-iteration}
		\Pr \big[\mathcal{E}_{sucess}\land (T=t_{\max}+1)\big]\leq  \delta^\prime.
	\end{equation}
 On the other hand, given any fixed $t\leq t_{\max}$, we have
	\begin{align*}
	&\Pr\Big[(\bigvee_{i\in [h]}\widetilde{\pi}_i(O_i,\mathcal{R})> \pi_i(O_i)+\varrho B_i/2)\land (T=t)\Big]\leq  {\delta^\prime}^{2^{t-1}},\\
	&\Pr\Big[(\bigvee_{i\in [h]}\pi({S}_i^*) > \mathit{UB}(S_i^*))\land(T=t)\Big]\leq \frac{h\delta^\prime}{(h+2)t_{\max}},\\
	&\Pr\Big[(\pi(\vec{O})> \mathit{UB}(\vec{O}))\land (T=t)\Big]\leq  \frac{\delta^\prime}{(h+2)t_{\max}},\\
	&\Pr\Big[(\pi(\vec{S}^*)< \mathit{LB}(\vec{S}^*))\land (T=t)\Big]\leq \frac{\delta^\prime}{(h+2)t_{\max}},
	\end{align*}
where the first inequality is obtained by similar reasoning with the proof of Lemma~\ref{lma:boundoptinsamples} and the union bound, the second and fourth inequalities can be obtained by Lemma~\ref{lma:hugeeqn} and the union bound, and the third inequality can be obtained by Lemma~\ref{lma:hugeeqn} and Lemma~\ref{lma:correctubofopt}.

	Recall that the algorithm terminates only when $\mathit{Feasible}=\mathrm{True}$ and $\LB(\vec{S}^*)/\UB(\vec{O})\geq \lambda-\epsilon$. Therefore, one can verify that $\mathcal{E}_{sucess}$ does not happen for $T\leq t_{\max}$ only if at least one of the four events considered above occurs in certain $t\in \{1,2,\dotsc,t_{\max}\}$. By union bound, the probability of such a scenario is at most
	\begin{equation}\label{eqn:ealier-stop}
	\begin{split}
	&\Pr \big[(\pi(\vec{S}^*)< (\lambda-\epsilon)\OPT)\land (T\leq t_{\max})\big]\\
	&\leq\sum_{t=1}^{t_{\max}}\Big( {\delta^\prime}^{2^{t-1}} + \frac{h\delta^\prime+\delta^\prime+\delta^\prime}{(h+2)t_{\max}}\Big)\leq \frac{\delta^\prime}{1-\delta^\prime}+\delta^\prime\leq 3\delta^\prime,
	\end{split}
	\end{equation}
	where the last inequality is due to $\delta^\prime=\delta/4\leq 1/4$. Therefore, by \eqref{eqn:last-iteration} and \eqref{eqn:ealier-stop}, no matter when the algorithm stops, we always have $\Pr [\mathcal{E}_{sucess}]\geq 1-4\delta^\prime=1-\delta$. So the theorem follows.
\end{proof}}

\revise{

\section{Time Complexity Analysis} \label{sec:timecompana}
\toblue{In this section, we provide the theoretical time complexity of $\mathsf{RM\_without\_Oracle}$ (RMA), as well as, of TI-CARM and TI-CSRM algorithms of \cite{Aslay2017}, which were left as an open problem in \cite{Aslay2017}.}

We start by providing the theoretical time complexity of RMA in the following theorem.

\begin{theorem}
$\mathsf{RM\_without\_Oracle}$ has an expected time complexity of $O\big(\frac{m\bar{\pi}(\ln \frac{1}{\delta}+n\ln h)}{\epsilon^2 B_{min}}\big)$, where $\bar{\pi}=\sum_{i\in[h]}\E[\pi_i(\{v^\ast\})]$ and $v^\ast$ denotes a random node selected from $V$ with probability proportional to its in-degree.
\end{theorem}

\begin{proof}
    The time complexity of Algorithm~\ref{alg:rmwithoutoracle} mainly determined by (i) the time for generating an RR set, and (ii) the number of RR sets generated. Under the triggering model \cite{Kempe2003}, which generalizes both IC and LT models, the expected time complexity of generating an RR set for advertiser $i$ is $O(\frac{m}{n}\E[\sigma_i(\{v^\ast\})])$, where the expectation is over the randomness of $v^\ast$ being randomly chosen from $V$ with probability proportional to its in-degree. Note that  $\bar{\pi}=\sum_{i\in[h]}\E[\pi_i(\{v^\ast\})]=\E[\sum_{i\in[h]}\pi_i(\{v^\ast\})]$, where $\sum_{i\in[h]}\pi_i(\{v^\ast\})$ represents $\{v^\ast\}$'s total revenue over all $h$ advertisements and the expectation is again over the randomness of $v^\ast$. Then, the expected time complexity for generating an RR set of our uniform sampling method is $O(\frac{m\bar{\pi}}{n\Gamma})$.
	
	On the other hand, it is easy to get that
	\begin{equation*}
		\bar{\theta}_{max}\leq \frac{8n^2\Gamma(2+\varrho/3)}{\varrho^2 B_{min}}\ln \frac{4h}{\delta^\prime}\leq 2n\theta_0,
	\end{equation*}
	and similarly
	\begin{align*}
	\hat{\theta}_{max}
	&\leq \frac{8n}{\epsilon^{2}}\left(\ln \frac{4}{\delta^\prime}+\sum\nolimits_{i\in[h]}{\mu_i}\ln\frac{\mathrm{e}n}{\mu_i}\right)\\
	&\leq \frac{8nh}{\epsilon^{2}}\left(\ln \frac{4}{\delta^\prime}+{\mu}\ln\frac{\mathrm{e}n}{\mu}\right)\\
	&\leq n\theta_0\cdot \frac{\varrho^2hB_{min} }{\epsilon^{2}\Gamma}.
	\end{align*}
	Thus, $\ln t_{\max}\in O( \frac{1}{\delta}+\sum\nolimits_{i\in[h]}{\mu_i}\ln\frac{\mathrm{e}n}{\mu_i})$ and $\ln t_{\max}\in O(\ln \frac{h}{\delta}+{\mu}\ln\frac{\mathrm{e}n}{{\mu}})$. Hence, one can verify that the expected number of RR sets generated is
	\begin{equation*}
	O\Big(\frac{n\Gamma(\ln \frac{1}{\delta}+\sum\nolimits_{i\in[h]}{\mu_i}\ln\frac{\mathrm{e}n}{\mu_i})}{\epsilon^{2}\OPT}+\frac{n\Gamma(\ln \frac{h}{\delta}+{\mu}\ln\frac{\mathrm{e}n}{{\mu}})}{\varrho^2 B_{min}}\Big).
	\end{equation*}
	According to Wald’s equation, the expected time complexity is
	\begin{equation*}
	O\Big(\frac{m\bar{\pi}(\ln \frac{1}{\delta}+\sum\nolimits_{i\in[h]}{\mu_i}\ln\frac{\mathrm{e}n}{\mu_i})}{\epsilon^{2}\OPT}+\frac{m\bar{\pi}(\ln \frac{h}{\delta}+{\mu}\ln\frac{\mathrm{e}n}{{\mu}})}{\varrho^2 B_{min}}\Big).
	\end{equation*}
	Note that $\sum\nolimits_{i\in[h]}{\mu_i}\ln\frac{\mathrm{e}n}{\mu_i}$ can be replaced by $n\ln h$, since each node can only be selected by one advertiser. Then, the expected time complexity is bounded by
	\begin{equation*}
	O\Big((m\bar{\pi}(\ln \frac{1}{\delta}+n\ln h)\cdot\big(\frac{1}{\epsilon^{2}\OPT}+\frac{1}{\varrho^2 B_{min}}\big)\Big).
	\end{equation*}
	Usually, $\epsilon=\Theta(\varrho)$ and $\OPT\geq B_{min}$, for which the expected time complexity is $O\big(\frac{m\bar{\pi}(\ln \frac{1}{\delta}+n\ln h)}{\epsilon^2 B_{min}}\big)$.
\end{proof}

\toblue{Next, we provide the time complexity analysis for the TI-CARM and TI-CSRM algorithms proposed in  \cite{Aslay2017}.}

\toblue{
\begin{theorem}
TI-CARM and TI-CSRM algorithms of \cite{Aslay2017} both have a time complexity of
$O\big(\frac{n (1 + \ell)(m + n) \ln n}{\epsilon^2})$ where $\ell = \frac{\ln \frac{1}{\delta}}{\ln n}$.
\end{theorem}

\begin{proof}
First, we remind that the TI-CARM and TI-CSRM algorithms of \cite{Aslay2017} utilize TIM~\cite{TangXS2014} algorithm as a subroutine for each advertiser, coupled with a latent seed set size estimation procedure, as TIM~\cite{TangXS2014} requires seed set size as input. Thus, both algorithms start with a latent seed set size $s_i = 1$ for each advertiser $i$, and iteratively revise the latent seed set size and expected spread  estimates whenever the size of the current solution $S_i$ reaches $s_i$. Specifically,
whenever $|S_i| = s_i$, the algorithms re-estimate $s_i$ as an upper bound on the possible final size of $S_i$ and re-derive the sample size, via the KPT-estimation procedure~\cite{TangXS2014}, that is sufficiently large for estimating the expected spread of any $s_i$ seeds.

Recall that TIM's time complexity is given by $O\big(\frac{(k + \ell)(m + n) \ln n}{\epsilon^2})$ with the  $(\ell (m + n) \ln n)$ term accounting for the running time of the KPT-estimation procedure and the rest accounting for the selection of $k$ seeds from the RR-sets sample.

In the worst-case, TI-CARM and TI-CSRM would perform the KPT-estimation procedure at each $s_i = 1, 2, .., k_i$ where $$k_i = \left\lfloor\frac{B_i}{\max_{u \in V} c_i(u) + \cpe(i) \, \max_{u \in V} \sigma_i(\{u\})} \right\rfloor.$$

Then, the worst-case time complexity, by factoring in $k_i$ KPT-estimation procedures, followed by selection of $|S_i| << k_i$ seeds from a sample size devised for the input $k_i$ for each advertiser $i$, is given by $O\big(\frac{\sum_{i \in [h]} k_i \cdot (1 + \ell)(m + n) \ln n}{\epsilon^2})$. By using the fact that $\sum_{i \in [h]} k_i < n$, we conclude that TI-CARM and TI-CSRM have the time complexity of $O\big(\frac{n (1 + \ell)(m + n) \ln n}{\epsilon^2})$.
\end{proof}
}

\toblue{\spara{Discussion} First, we would like to emphasize that the complexity results we have provided for RMA and TI-CARM / TI-CSRM are asymptotic worst-case results, thus, a direct comparison of such worst-case results would not necessarily allow us to conclude whether an algorithm is ``always" more efficient than another. A healthier comparison to draw such conclusion requires to find a lower bound on the complexity (i.e., comparison over the $\Omega(\cdot)$ function). While such analysis is interesting, it is beyond the scope of this paper.

Being mindful of this, still we can compare the asymptotic worst-case results: we see that the running time of RMA is dominated by the factor $mn$ while the running time of TI-CARM and TI-CSRM are dominated by the factor $n(m + n)$, translating to the superiority of the RMA algorithm in terms of asymptotic worst-case running time.
}

\section{Additional Experimental Results}

\subsection{Impact of Parameters $\tau$ and $\varrho$} \label{sec:ontauvarrho}
In this section, we study how the performance of implemented algorithms can be affected by the parameters $\tau$ and $\varrho$.

We first study the impact of $\tau$ in Fig.~\ref{fig:tau} (for revenue) and Table~\ref{table:tau} (for running time) under the linear cost model with $\alpha=0.1$, where all the other settings are the same with those in Fig.~\ref{fig:totalrevenue}. \toblue{It can be seen that both the revenue and running time of RMA generally show a slight decreasing trend when $\tau$ increases,} which corroborates the intuition we provided in Sec.~\ref{sec:puttogether} that $\tau$ is a parameter controlling the tradeoff between efficiency and accuracy in RMA. However, Fig.~\ref{fig:tau} and Table~\ref{table:tau} also show that the superiority of RMA maintains when $\tau$ changes, which demonstrates that it is reasonable to use $\tau=0.1$ as the default setting in our experiments.
\begin{figure}[t]
    \centering
    \includegraphics[width=0.485\textwidth]{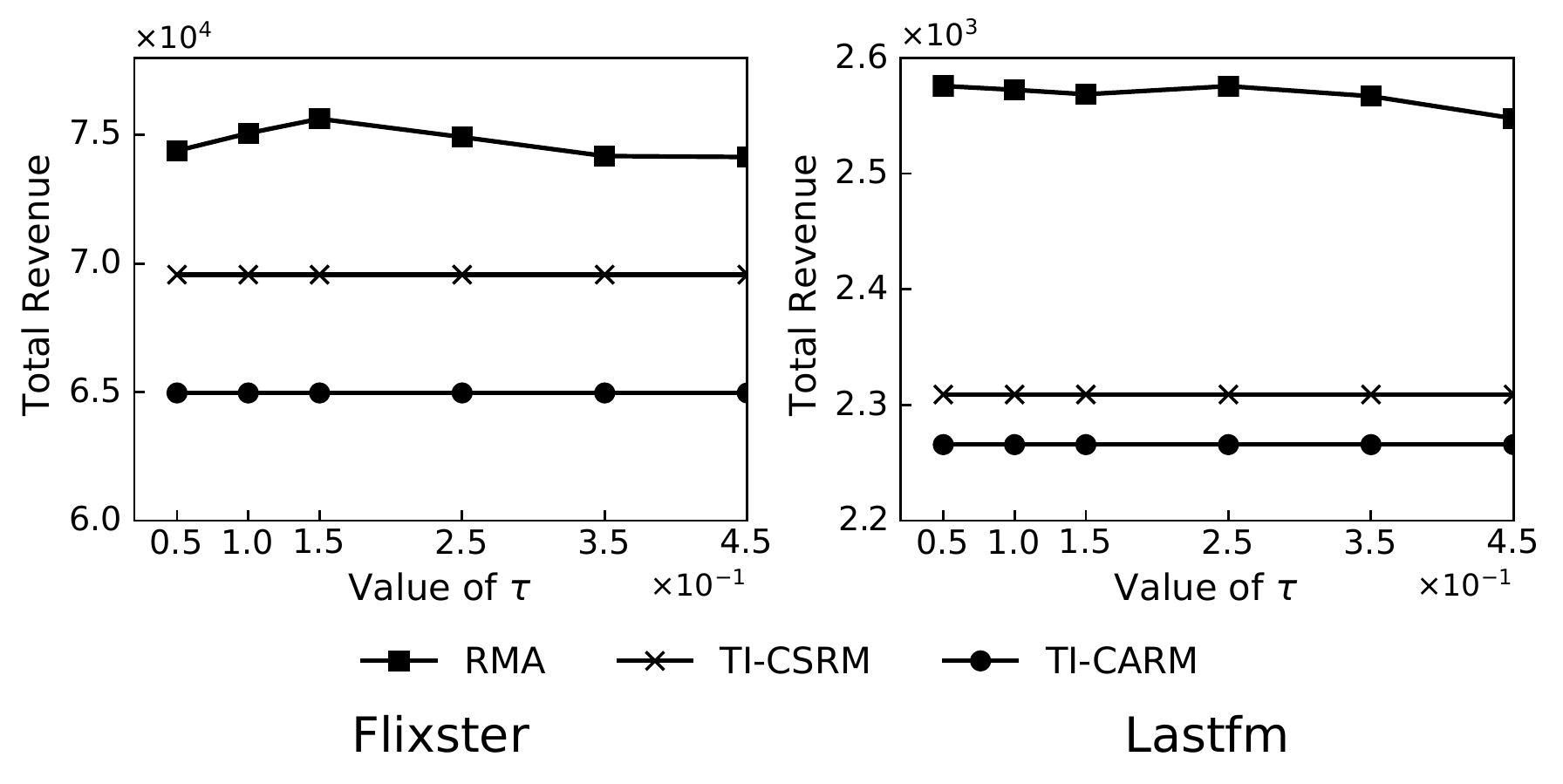}
    \vspace{-0.8cm}
    \caption{Total revenue as a function of $\tau$, on $\texttt{\texttt{Flixster}}$(left) and $\texttt{Lastfm}$(right), for linear incentive model.}
    \label{fig:tau}
\end{figure}

\begin{table}
    \centering
    \caption{Running time (seconds) when $\tau$ changes}
     \setlength{\tabcolsep}{1mm}{
    \begin{tabular}{|c|c|c|c|c|c|c|c|c|}
    \hline
    \multirow{2}*{\texttt{Lastfm}}&  \multicolumn{6}{|c|}{Running time (second)}  \\
    \cline{2-7}
     &$\tau=0.05$ &0.10&0.15& 0.25&0.35&0.45\\
     \hline
    RMA&27&26&25&24&23&23 \\ \hline
    TI-CARM &108&108&108&108&108&108\\ \hline
    TI-CSRM&130&130&130&130&130&130\\ \hline
    \hline
    \multirow{2}*{\texttt{Flixster}}&  \multicolumn{6}{|c|}{Running time (second)}  \\
    \cline{2-7}
     &$\tau=0.05$ &0.10&0.15& 0.25&0.35&0.45\\
     \hline
    RMA&710&589&603&524&533&540 \\ \hline
    TI-CARM &3803&3803&3803&3803&3803&3803\\ \hline
    TI-CSRM&16255&16255&16255&16255&16255&16255\\ \hline
    \end{tabular}}
    \label{table:tau}
\end{table}

Next, we study the impact of $\varrho$ on the revenue performance of RMA in Fig.~\ref{fig:varrho} under the linear cost model with $\alpha=0.1$, where all the other settings are the same with those in Fig.~\ref{fig:totalrevenue}. It can be seen that the revenue of RMA decreases when $\varrho$ increases, which is not surprising as the budgets used by RMA is only $(1+\varrho)^{-1}$ fraction of those for TI-CSRM/TI-CARM (as explained in Sec.~\ref{sec:expsetting}). However, according to the comparison method described in Sec.~\ref{sec:expsetting}, $\varrho$ reflects the ``budget overshoot'' of RMA and can be set to any positive number without harming the fairness of our experiment, because the actual budget used by RMA is guaranteed to be no more than that used by TI-CSRM/TI-CARM. As such, it is natural and reasonable to set $\varrho$ as a small number (e.g. $\varrho=0.1$) in our experiments to avoid a large budget overshoot.

\begin{figure}[t]
    \centering
    \includegraphics[width=0.485\textwidth]{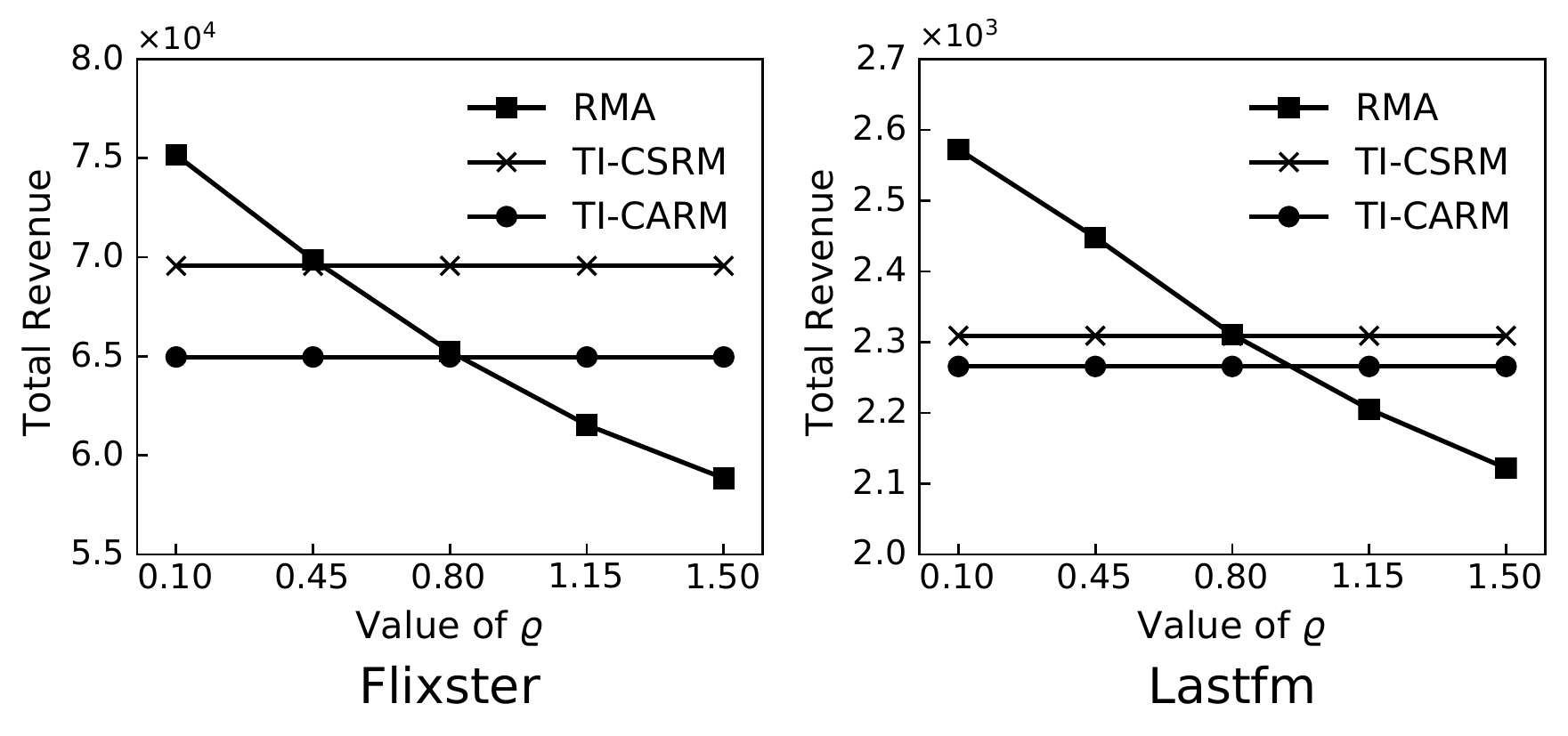}
    \vspace{-0.8cm}
    \caption{Total revenue as a function of $\varrho$, on $\texttt{\texttt{Flixster}}$(left) and $\texttt{Lastfm}$(right), for linear incentive model.}
    \label{fig:varrho}
\end{figure}

\subsection{Experiments on Further Acceleration}\label{sec:subsimexp}

Recently, the work in \cite{Guo_IM_2020} proposed a nice algorithm for accelerating the traditional influence maximization problem, where the OPIM-C framework in~\cite{Tang2018} is also leveraged but a novel SUBSIM algorithm is used to accelerate the generation of a single RR-set. We note that that the SUBSIM can also be plugged into the RMA, TI-CARM, and TI-CSRM algorithms for acceleration, so we perform experiments in Fig.~\ref{fig:subsim} under the linear cost model by calling SUBSIM in RMA/TI-CARM/TI-CSRM for generating a single RR-set, and all the other parameter settings are the same with those in Fig.~\ref{fig:totalrevenue}. Meanwhile, we also list the running time of the compared algorithms in Table~\ref{table:subsim} accordingly. These experimental results show that: (1) the revenues achieved by all compared algorithms are almost identical to those shown in Fig.~\ref{fig:subsim}, as the RR-sets generates by SUBSIM are essentially the same as those generated before, and (2) all the algorithms are speeded up by calling SUBSIM, but RMA is still significantly faster than TI-CARM/TI-CSRM. In summary, these experimental results demonstrate that the superiority of RMA still remains when SUBSIM is used for acceleration.

    \begin{figure}[t]
        \centering
        \includegraphics[width=0.485\textwidth]{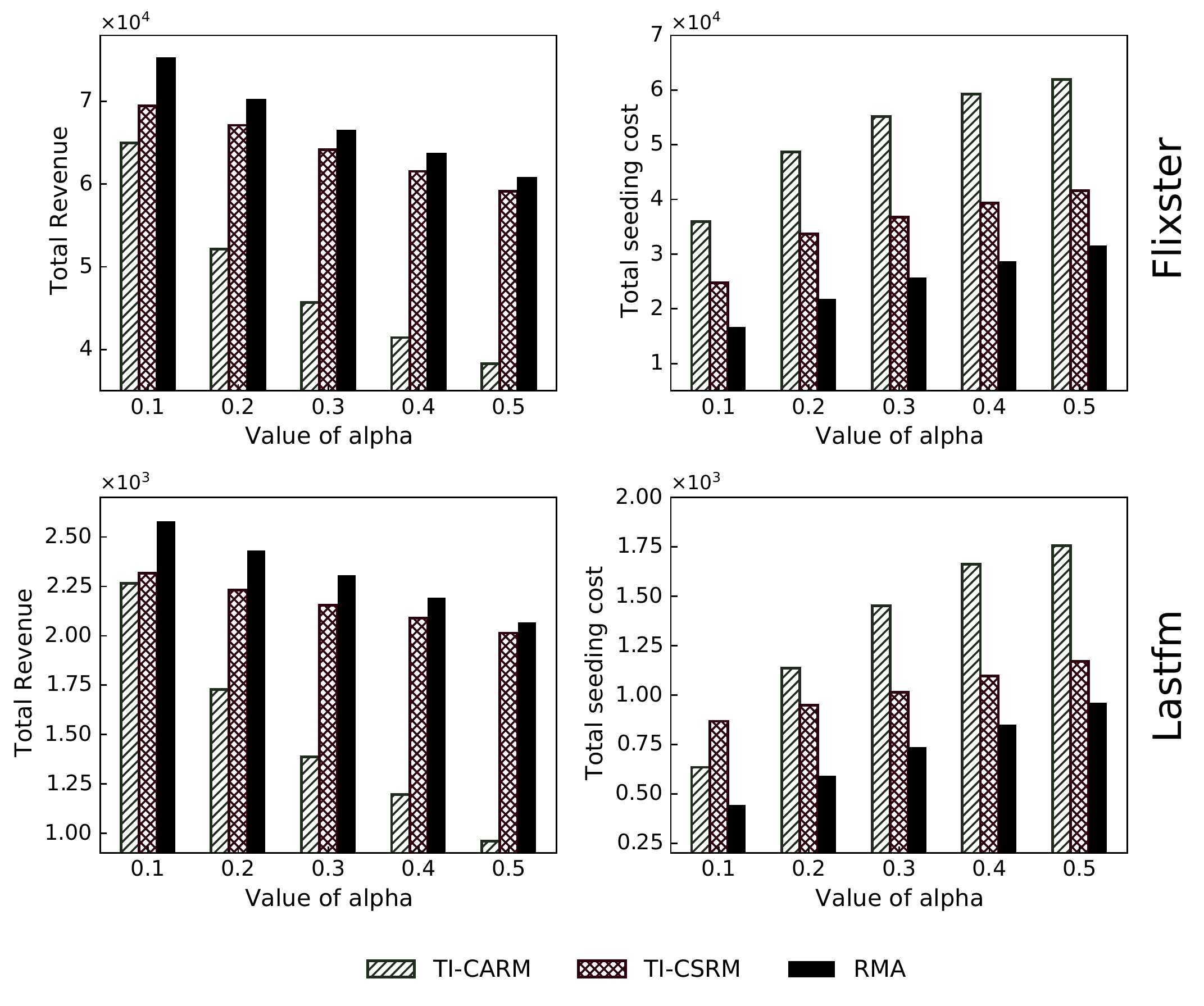}
        \vspace{-0.6cm}
        \caption{Total revenue (left) and total seeding cost (right) as a function of $\alpha$, on $\texttt{Fllxster}$ (top) and $\texttt{Lastfm}$ (bottom), for linear incentive model, by using SUBSIM.}
        \label{fig:subsim}
    \end{figure}


\begin{table}[t]
    \centering
    \caption{\small{Running time (seconds) by using SUBSIM }}
     \setlength{\tabcolsep}{2.30mm}{
    \begin{tabular}{|c|c|c|c|c|c|}
    \hline
    \texttt{Flixster} &$\alpha$=0.1&0.2&0.3&0.4&0.5 \\ \hline
    RMA&583&568&568&548&532\\ \hline
    TI-CARM&2653&1293&954&760&657\\ \hline
    TI-CSRM&12082&14985&15911&18676&19666\\ \hline
    \hline
    \texttt{Lastfm}&$\alpha$=0.1&0.2&0.3&0.4&0.5\\
    \hline
    RMA&22&19&20&21&17\\ \hline
    TI-CARM&84&75&65&59&54\\ \hline
    TI-CSRM&88&101&108&117&123\\ \hline
    \end{tabular}
    }
    \label{table:subsim}
\end{table}

\clearpage
}